%% file: main.tex
\documentclass[11pt]{article}

\usepackage{amsmath,amssymb,amsfonts}
\usepackage{amsmath,amssymb,mathtools}
\usepackage{xparse} 
\usepackage{tcolorbox} 

\usepackage[english]{babel}
\usepackage[letterpaper,hmargin=0.8in,vmargin=1in]{geometry}
\usepackage{palatino}
\usepackage{xspace,amsmath,mathtools,amssymb,amsfonts,amsthm}
\usepackage[svgnames,x11names,table]{xcolor}
\usepackage{algorithm}
\usepackage{algpseudocode}
\usepackage{float} 
\usepackage{multirow}
\usepackage{wrapfig}
\usepackage{caption}
\usepackage{subfig}
\usepackage{dsfont}
\usepackage[mathscr]{eucal}
\usepackage{hyperref}
\hypersetup{colorlinks=true,citecolor=blue,urlcolor=blue,bookmarks=false,hypertexnames=true}
\usepackage{array}
\usepackage{booktabs}
\usepackage{tabularx}

\usepackage{pgfplots}
\usepackage{tikz}
\usetikzlibrary{arrows.meta,calc,positioning,decorations.pathreplacing}
\usepackage{tablefootnote}

\usepackage{mdframed}

\algrenewcommand\algorithmicrequire{\textbf{Input:}}
\algrenewcommand\algorithmicensure{\textbf{Output:}}

\newtheorem{theorem}{Theorem}[section]
\newtheorem{definition}[theorem]{Definition}

\newtheorem{remark}[theorem]{Remark}
\newtheorem{lemma}[theorem]{Lemma}
\newtheorem{corollary}[theorem]{Corollary}

\DeclareMathAlphabet{\mathcal}{OMS}{cmsy}{m}{n}

\input{macros}
\begin{document}

\title{ 
\Proj: Latency-Friendly and Resilient Multi-Proposer Consensus
}

\author{
Álvaro Castro-Castilla\thanks{Nomos, Institute of Free Technology. \ \ 
\texttt{\{alvaro|marcin\}@nomos.tech}}
\and
Marcin Pawlowski$^\ast$\thanks{Jagiellonian University. }
\and
Hong-Sheng Zhou$^\ast$\thanks{Virginia Commonwealth University. \ \ 
\texttt{hszhou@vcu.edu}}
}

\maketitle

\setcounter{tocdepth}{2}
\setcounter{secnumdepth}{3}
\pagestyle{plain}

\input{abstract}

\newpage
\small
\tableofcontents
\normalsize

\newpage
\pagenumbering{arabic}

\input{sec1}
\input{sec1-related}

\input{sec1-organization}


\input{sec2}

\input{sec2-assumptions}

\input{sec2-tools}

\input{sec2-model}

\input{sec3}

\input{sec4}

\input{sec5}

\input{sec6}

\paragraph{Acknowledgements:}
The third author thanks Lei Fan for many valuable discussions on the DAG-based protocols studied in this work. He also thanks Lei for numerous insightful conversations over the years on the design and analysis of consensus protocols in general.

\bibliographystyle{alpha}
\bibliography{../crypto/abbrev3,../crypto/crypto,../crypto/extra,../crypto/extra3}

\appendix
\input{app2}

\input{app4}

\input{app6}

\end{document}

%% file: macros.tex
\newcommand{\SamplePayload}{\ensuremath{\textsc{SamplePayload}}\xspace}

\newcommand{\MaximumAntichain}{\ensuremath{\textsc{MaximumAntichain}}\xspace}
\newcommand{\ExactMaxAntichain}{\textsc{ExactMaxAntichain}\xspace}

\newcommand{\BranchW}{\ensuremath{{\mathrm{SubdagW}}}\xspace}
\newcommand{\SubdagW}{\ensuremath{{\mathrm{SubdagW}}}\xspace}

\newcommand{\Subdag}{{SubDAG}\xspace}
\newcommand{\subdag}{{subDAG}\xspace}
\newcommand{\subdags}{{subDAGs}\xspace}

\newcommand{\prefsubdag}{{preferred \subdag}\xspace}
\newcommand{\prefsubdags}{{preferred \subdags}\xspace}

\newcommand{\adversarialsubdag}{{adversarial \subdag}\xspace}

\newcommand{\Branch}{{\Subdag}\xspace}

\newcommand{\branch}{{\subdag}\xspace}
\newcommand{\branches}{{\subdags}\xspace}

\newcommand{\Subdagweight}{{\Subdag{} weight}\xspace}
\newcommand{\subdagweight}{{\subdag{} weight}\xspace}
\newcommand{\subdagweights}{{\subdag{} weights}\xspace}

\newcommand{\subdagscores}{{\subdag{} scores}\xspace}

\newcommand{\Proj}{\ensuremath{\mathsf{Areon}}\xspace}
\newcommand{\ProjIdeal}{\ensuremath{\mathsf{\Proj\text{-}Ideal}}\xspace}
\newcommand{\ProjBase}{\ensuremath{\mathsf{\Proj\text{-}Base}}\xspace}

\providecommand{\decay}{\varrho}

\newcommand{\secp}{\ensuremath{\kappa}\xspace}
\newcommand{\negl}{\ensuremath{\mathsf{negl}}\xspace}
\newcommand{\true}{\ensuremath{\mathsf{true}}\xspace}
\newcommand{\false}{\ensuremath{\mathsf{false}}\xspace}

\newcommand{\Undom}{\ensuremath{{F}}\xspace}
\newcommand{\Past}{\ensuremath{\mathrm{Past}}\xspace}
\newcommand{\Trim}{\ensuremath{\mathsf{Trim}}\xspace}

\newcommand{\DCP}{{DCP}\xspace}

\newcommand{\CTR}{\ensuremath{\mathsf{CTR}}\xspace}
\newcommand{\winner}{\ensuremath{\mathtt{winner}}\xspace}
\newcommand{\loser}{\ensuremath{\mathtt{loser}}\xspace}

\newcommand{\E}{\ensuremath{\mathbb{E}}\xspace}

\newcommand{\cP}{\ensuremath{\mathcal{P}}\xspace}

\newcommand{\cE}{\ensuremath{\mathcal{E}}\xspace}

\newcommand{\pk}{\ensuremath{\mathrm{pk}}\xspace}
\newcommand{\sk}{\ensuremath{\mathrm{sk}}\xspace}
\newcommand{\stake}{\ensuremath{\mathrm{stake}}\xspace}
\newcommand{\StakeTot}{\ensuremath{\mathrm{Stake}_{\mathrm{tot}}}\xspace}
\newcommand{\stakefrac}[1]{\ensuremath{\frac{\stake(#1)}{\StakeTot}}\xspace}

\newcommand{\id}{\ensuremath{\mathrm{id}}\xspace}
\newcommand{\Val}{\ensuremath{\mathcal{V}}\xspace}
\newcommand{\val}{\ensuremath{\mathrm{validator}}\xspace}
\newcommand{\slot}{\ensuremath{\mathrm{slot}}\xspace}
\newcommand{\block}{\ensuremath{{b}}\xspace}
\newcommand{\payload}{\ensuremath{\mathrm{payload}}\xspace}
\newcommand{\refs}{\ensuremath{\mathrm{refs}}\xspace}
\newcommand{\longref}{\ensuremath{\mathrm{longref}}\xspace}

\newcommand{\Blocks}{V}            
\newcommand{\Parties}{\mathcal{P}} 
\newcommand{\Anc}{\mathrm{Anc}}
\newcommand{\Desc}{\mathrm{Desc}}
\newcommand{\AncStar}{\mathrm{Anc}^{\!*}}
\newcommand{\DescStar}{\mathrm{Desc}^{\!*}}

\newcommand{\Conflicts}{\ensuremath{\mathrm{Conflicts}}\xspace}

\newcommand{\Tips}{\ensuremath{\mathrm{Tips}}\xspace}
\newcommand{\Lin}{\ensuremath{\mathsf{Lin}}\xspace}

\newcommand{\Eligibility}{\ensuremath{\textsc{Eligibility}}\xspace}
\newcommand{\VRFGen}{\ensuremath{\mathrm{VRF\text{-}Gen}}\xspace}
\newcommand{\VRFSign}{\ensuremath{\mathrm{VRF\text{-}Sign}}\xspace}
\newcommand{\VRFVerify}{\ensuremath{\mathrm{VRF\text{-}Verify}}\xspace}

\newcommand{\SIGGen}{\ensuremath{\mathrm{SIG\text{-}Gen}}\xspace}
\newcommand{\SIGSign}{\ensuremath{\mathrm{SIG\text{-}Sign}}\xspace}
\newcommand{\SIGVerify}{\ensuremath{\mathrm{SIG\text{-}Verify}}\xspace}
\newcommand{\Hash}{\ensuremath{\mathrm{Hash}}\xspace}

\newcommand{\Encode}{\ensuremath{\mathrm{enc}}\xspace}

\newcommand{\AntichainSelection}{\textsc{AntichainSelection}\xspace}
\newcommand{\Broadcast}{\textsc{Broadcast}\xspace}
\newcommand{\ForkChoiceUpdate}{\textsc{ForkChoiceUpdate}\xspace}

\newcommand{\wref}{\ensuremath{\mathrm{wref}}\xspace}
\newcommand{\CCA}{\ensuremath{\mathrm{CCA}}\xspace}

\newcommand{\TB}{\ensuremath{\mathsf{TB}}\xspace}
\newcommand{\ShortRefs}{\ensuremath{\mathrm{ShortRefs}}\xspace}
\newcommand{\LongRef}{\textsc{LongRef}\xspace}

\makeatletter
\newcounter{subsubsubsection}[subsubsection]

\newcommand\subsubsubsection{%
  \@startsection{subsubsubsection}{4}{\z@}%
    {3.25ex \@plus 1ex \@minus .2ex}
    {1.5ex  \@plus .2ex}
    {\normalfont\normalsize\bfseries}   
}

\newcommand{\l@subsubsubsection}{\@dottedtocline{4}{7em}{4em}}

\providecommand{\toclevel@subsubsubsection}{4}
\makeatother

\makeatletter
\renewcommand*{\@fnsymbol}[1]{%
  \ensuremath{%
    \ifcase#1\or
      \ast             
    \or \dagger        
    \or \ddagger       
    \or \mathsection   
    \or \mathparagraph
    \or \|
    \or \ast\ast
    \or \dagger\dagger
    \or \ddagger\ddagger
    \else\@ctrerr
    \fi}}
\makeatother

%% file: abstract.tex

\begin{abstract}
We present \Proj, a family of latency-friendly, stake-weighted, multi-proposer proof-of-stake consensus protocols. By allowing multiple proposers per slot and organizing blocks into a directed acyclic graph (DAG), \Proj achieves robustness under partial synchrony. Blocks reference each other within a sliding window, forming maximal antichains that represent parallel ``votes'' on history. Conflicting subDAGs are resolved by a closest common ancestor (CCA)-local, window-filtered fork choice that compares the weight of each subDAG---the number of recent short references---and prefers the heavier one. Combined with a structural invariant we call Tip-Boundedness (TB), this yields a bounded-width frontier and allows honest work to aggregate quickly.

We formalize an idealized protocol (\ProjIdeal) that abstracts away network delay and reference bounds, and a practical protocol (\ProjBase) that adds VRF-based eligibility, bounded short and long references, and application-level validity and conflict checks at the block level. On top of DAG analogues of the classical common-prefix, chain-growth, and chain-quality properties, we prove a backbone-style $(k,\varepsilon)$-finality theorem that calibrates confirmation depth as a function of the window length and target tail probability. We focus on consensus at the level of blocks; extending the framework to richer transaction selection, sampling, and redundancy policies is left to future work.

Finally, we build a discrete-event simulator and compare \ProjBase against a chain-based baseline (Ouroboros Praos) under matched block-arrival rates. Across a wide range of adversarial stakes and network delays, \ProjBase achieves bounded-latency finality with consistently lower reorganization frequency and depth. 
\end{abstract}

%% file: sec1.tex

\section{Introduction}
\label{sec:intro}

A recurring tension in blockchain design is achieving high throughput and low latency without compromising security or decentralization. Traditional single-proposer (leader-based) protocols---from Bitcoin's longest chain to many proof-of-stake (PoS) variants---extend a single linear chain one block at a time. This serialization simplifies reasoning but forces a conservative block rate to preserve the common prefix and chain quality properties under network delay. If blocks are proposed faster than they propagate, honest views diverge, forks proliferate, and safety and liveness degrade. As a result, single-proposer designs deliberately throttle throughput, which increases confirmation latency.

Single-leader systems are also brittle under adverse conditions: even moderate asynchrony or partitions can trigger deep reorganizations, and a slow or malicious leader can temporarily suppress blocks, harming liveness and fairness. The upshot is a fundamental scalability--security trade-off: increasing throughput or tolerating higher network latency tends to weaken safety guarantees.

\subsection{From single-leader chains to multi-proposer DAGs}

Single-leader chains serialize block creation and concentrate forks on a single tip, making them sensitive to delay spikes and strategic withholding. Multi-proposer DAGs distribute proposal opportunities across many validators and many tips, amortizing propagation delay and reducing the leverage of a temporary leader. The remaining challenge is to retain chain quality and finality without enforcing a total order at proposal time. Our design addresses this by using a sliding-window 
fork choice on a directed acyclic graph (DAG) that is local to the closest common ancestor (CCA) of competing tips: given two conflicting tips, we restrict attention to the two subDAGs above their CCA and prefer the one with more recent short references. 
We show that this 
CCA-based, windowed 
fork choice must be coupled with a new invariant, Tip-Boundedness (TB), which bounds the number of honest tips at any time and plays for DAGs the role that common-prefix and chain-growth play for longest-chain protocols.

\subsection{Our contributions}
\label{sec:results}

We present \Proj, a family of stake-weighted, multi-proposer PoS protocols designed for high-latency-tolerant (bounded-latency) finality and robustness under partial synchrony. We require a short-reference window $w \ge \Delta$.

A central message of our analysis is that classical chain invariants do not directly carry over to multi-proposer DAGs. We formalize DAG analogues of backbone properties---DAG Growth (DG), DAG Quality (DQ), and DAG Common Past (\DCP)---and show that they must be paired with a fourth invariant, Tip Boundedness (TB), which upper-bounds the number of short-reference tips in the window of length~$w$ (equivalently, the number of blocks in that window that can still be short-referenced by new blocks). Without TB, honest work can fragment indefinitely across many incomparable tips, violating either safety or liveness even when DG/DQ/\DCP hold.

\begin{itemize}
\item \textbf{\ProjIdeal (idealized protocol).} We formalize a  multi-proposer DAG protocol under synchronous conditions with unbounded referencing, using a CCA-local, window-filtered fork choice. We prove that it satisfies DG, DQ, \DCP, and TB, and derive persistence and liveness with confirmation depth $k = \Theta(w)$.
\item \textbf{\ProjBase (practical protocol).} We realize the design under partial synchrony with VRF sortition, 
bounded short references (window~$w$) that drive fork choice and confirmation, and weightless long references that preserve DAG connectivity.
We prove that under standard PoS assumptions (honest stake $H > 1/2$) and $w \ge \Delta$, \ProjBase satisfies DG/DQ/\DCP/TB and thus achieves persistence and liveness. Withheld work outside the window does not help the adversary.

\item \textbf{Backbone-style DAG invariants.}
We adapt classical backbone arguments~\cite{EC:GarKiaLeo15,EC:PasSeeShe17} to a multi-proposer DAG by introducing DAG-Growth (DG), DAG-Quality (DQ), and a DAG Common-Past property \DCP that lift the chain-growth, chain-quality, and common-prefix invariants to the DAG setting. We show that DG, DQ, and \DCP alone do not imply ledger safety or liveness in a DAG.

\item \textbf{Tip-Boundedness and $(k,\varepsilon)$-finality.}
We identify Tip-Boundedness (TB) as the missing invariant that links concurrent proposal (throughput) to bounded confirmation latency in DAG-based protocols. We prove that TB is necessary for block-confirmation liveness, and that DG, DQ, \DCP, and TB together yield a backbone-style persistence and liveness theorem with explicit $(k,\varepsilon)$-finality guarantees: once a block is $k$-deep under our fork choice, the probability that it is ever reverted is at most~$\varepsilon$.

\item \textbf{Latency-friendly convergence.} The CCA-local, sliding-window fork choice aggregates honest references quickly and resists long-range withholding, enabling fast convergence once a windowed margin appears. This explains the observed bounded-latency (high-latency-tolerant) finality in our simulations.
\item \textbf{Evaluation.}
We build a discrete-event simulator for \ProjBase and an optimized adversary that withholds and later releases 
private \subdags, 
and we instantiate a chain-based baseline by adapting Ouroboros Praos~\cite{EC:DGKR18} to the same setting.
Across a range of adversarial stakes and network delays, \Proj achieves consistently lower reorg frequencies at the same depths and faster stabilization than the Ouroboros Praos baseline at matched block-arrival rates (e.g., 14--15-block reorgs are about $600\times$ less frequent and 10-block reorgs about $8.7\times 10^{3}$ times less frequent).
\end{itemize}
Throughout this work we adopt a transaction-agnostic model: each block carries a single opaque application payload (e.g., a transaction), and we analyze consensus and finality purely at the block level; multi-transaction blocks and mempool sampling are left to future work (see Remark~\ref{rem:block-vs-tx}).

%% file: sec1-related.tex


\subsection{Related Work}
\label{sec:related}

Decentralized consensus has evolved from single-leader PoW/PoS chains to multi-proposer and DAG-based designs. We outline the most relevant lines and position our approach.

\medskip\noindent
\textbf{Proof-of-Stake chain protocols.}
Early PoS protocols such as Snow~White~\cite{FC:DaiPasShi19}, Ouroboros Praos~\cite{EC:DGKR18}, and Ouroboros Genesis~\cite{CCS:BGKRZ18} follow a longest-chain discipline with (effectively) one leader per slot via VRFs. Committee/BFT-style PoS systems--- Algorand~\cite{Algorand}, HotStuff~\cite{HotStuff}, and Dfinity~\cite{Dfinity} --- achieve rapid finality with quorum certificates at higher communication cost. Unpredictability mechanisms (e.g.,~\cite{FC:DebCabTse21,FKLTZ24,FKLTZ25}) and privacy-preserving stake~\cite{SP:KKKZ19} are orthogonal to our fork-choice rule but interact with eligibility and incentives. Backbone-style analyses~\cite{EC:GarKiaLeo15,EC:PasSeeShe17} and limits/attacks for longest-chain PoS (e.g.,~\cite{rationalattacksPOS}) motivate our windowed-DAG design.

\medskip\noindent
\textbf{Multi-proposer and DAG-style protocols.}
PHANTOM/GHOSTDAG~\cite{AFT:SWZ21} and Conflux~\cite{Conflux} combine Nakamoto-style work with DAG ordering; PRISM~\cite{Prism,EPRINT:ZhaChaLeo18} decomposes roles into parallel chains; other DAG voting/weighting approaches include~\cite{EPRINT:MorKulYok18} and Snow/Avalanche-style sampling~\cite{SnowFamily}. These works show that parallel proposal scales throughput, but their ordering/security arguments differ from our window-filtered, CCA-local approach.

\medskip\noindent
\textbf{Decoupled availability and ordering.}
Narwhal-Tusk~\cite{NarwhalTusk} decouples data availability from ordering to improve scalability. 
Our design is related but different: 
we remain within a \emph{single DAG} and use a sliding-window, CCA-local scoring that yields Tip-Boundedness and bounded-latency finality in partial synchrony.

\medskip\noindent
\textbf{Positioning.}
\Proj pairs a window-filtered fork choice with a Tip-Boundedness (TB) invariant.
DG/DQ/DCP alone do not secure a DAG ledger; TB is the additional structural condition that prevents 
frontier explosion 
and enables the $(k,\varepsilon)$-finality calibration we prove.
Unlike global optimization layers, our rule is local: it only inspects the subDAG above the current closest common ancestor (CCA) of competing tips and a sliding window of recent blocks, which admits a simple implementation. 
Our analysis deliberately abstracts away transaction batching and mempool policies by working in a single-payload-per-block model; this mirrors the backbone line~\cite{EC:GarKiaLeo15,EC:PasSeeShe17} and Ouroboros Praos~\cite{EC:DGKR18}, where the focus is on consensus and finality rather than throughput-oriented optimizations.

%% file: sec1-organization.tex

\subsection{Organization}
The remainder of the paper is organized as follows. 

Section~\ref{sec:prelim} sets up the execution model, notation, and DAG preliminaries, and formalizes the backbone-style invariants for DAGs---DAG Growth (DG), DAG Quality (DQ), DAG Common Past (\DCP), and Tip-Boundedness (TB).
Section~\ref{sec:ideal} presents \ProjIdeal: the multi-proposer DAG with a window-filtered
\emph{local} fork-choice rule and CCA-based conflict resolution, together with
its algorithms and proofs that it satisfies DG/DQ/\DCP/TB as well as ledger
safety and liveness.
Section~\ref{sec:base} instantiates {\ProjBase} under partial synchrony with VRF sortition, bounded short references (and a weightless long reference), and proves the practical counterparts of the invariants under \(w \ge \Delta\).
Section~\ref{sec:finality} develops a master \((k,\varepsilon)\)-finality theorem, derives parameter calibration and a deployment recipe, and maps depth to wall-clock settlement.
Section~\ref{sec:eval} describes our simulator and evaluation, including sensitivity to parallelism \(f\), adversarial stake \(\alpha\), and network delay \(\Delta\), 
and compares \Proj against Ouroboros Praos, used here as a representative single-leader PoS baseline.

Finally, Appendix~\ref{sec:concentration-app} provides the standard concentration bounds which are used in our analysis, Appendix~\ref{sec:base-app} gives supplemental material for Section~\ref{sec:base} (e.g., exact antichain computation), and Appendix~\ref{sec:eval-app} collects the full set of plots moved from Section~\ref{sec:eval} for readability.

%% file: sec2.tex


\section{Preliminaries}
\label{sec:prelim}

This section describes the execution model, timing assumptions, and the DAG structures used throughout the paper. We first establish basic notation for ancestry relationships and sliding window references, then formalize the core ledger properties (DAG Growth, DAG Quality, and DAG Common Past). Finally, we introduce \emph{Tip-Boundedness (TB)} and discuss its role in ensuring safety and liveness.

%% file: sec2-assumptions.tex

\subsection{Standing Assumptions \& Symbols}
\label{sec:assumption-symbol}
\label{sec:assumptions}
\label{sec:symbols}

\begin{table}[htp!]
  \centering
  \renewcommand{\arraystretch}{1.15}
  \begin{tabularx}{\linewidth}{@{} l X @{}}
    \toprule
    \textbf{Symbol} & \textbf{Meaning} \\
    \midrule
    $\Delta$ & Network delay bound (in slots); messages are delivered within $\Delta$ slots. \\
    $w$ & Short-reference window length (in slots); we require $w \ge \Delta$. \\
    $\mathcal{P}$ & Validator set (participants). \\
    $\lambda_h,\lambda_a$ & Expected numbers of honest/adversarial eligible blocks per slot. \\
    $\lambda$ & Total expected proposals per slot, $\lambda \approx \lambda_h+\lambda_a$ (ideal: public coin; base: VRF). \\
    $q$ & Honest coverage probability: when eligible, an honest block includes $\ge 1$ short-ref to a visible tip. \\
    $\beta$ & Tip-boundedness cap: with high probability $|\Tips_t|\le \beta$; $\beta=O(\lambda)$ (ideal), $\beta=O((\lambda_h+\lambda_a)\Delta)$ (base). \\
    $\Tips_t$ & Short-reference tips visible to an honest party at the end of slot $t$. \\
    $\Anc(b),\Desc(b)$&Ancestor/descendant sets of $b$ in the block DAG (view subscript omitted when clear from context). \\
    $\Anc^\ast(b),\Desc^\ast(b)$ &Ancestor/descendant closures of $b$ in the block DAG (including $b$ itself). \\
    $\Anc_V(b),\Desc_V(b)$ &Ancestors/descendants of $b$ in local view $V$ (includes genesis). \\

    $\CCA(i,j)$ & Closest common ancestor of $i,j$ in the current view (Def.~\ref{def:cca}). \\

    $\BranchW(x;c,w)$ & CCA-local subDAG weight of tip $x$ relative to anchor $c$ over window length $w$ (used in the fork-choice rule; see Section~\ref{sec:fork-choice}). \\
    $\CTR(i,j)$ & Conflicted-tips resolution rule comparing two short-reference tips $i,j$ (Def.~\ref{def:ctr}). \\

    $k_D$ & \DCP trimming parameter (Def.~\ref{def:DCP}). \\
    $k$ & Ledger finality depth. \\
    $\ell_{\mathrm{live}}$ & Liveness horizon (slots). \\
    \bottomrule
  \end{tabularx}
  \caption{Symbols used throughout the paper.}
  \label{tab:notation}
\end{table}

\noindent
\textbf{Security parameter.}
Let $\kappa \in \mathbb{N}$ be the security parameter; $\negl(\cdot)$ denotes negligible functions.
All probabilities are taken over the random coins of the protocol, the adversary's choices, and the
network schedule.

\medskip\noindent
\textbf{Network synchrony.}
There exists an unknown network delay bound $\Delta$ that holds throughout the execution:
any message broadcast by an honest party in slot $t$ is delivered to every honest party by the
end of slot $t+\Delta$.
We do not assume a global stabilization time; the adversary schedules deliveries at all times
subject only to the $\Delta$-delay constraint.

\medskip\noindent
\textbf{Participants and stake.}
The adversary controls at most an $\alpha$ fraction of the total stake (static or adaptive, as specified).
We assume an honest \emph{coverage probability} $q>0$, meaning that whenever an honest
validator is eligible and produces a block, it includes at least one short reference to a visible tip
with positive probability~$q$.

\medskip\noindent
\textbf{Eligibility rates.}
Honest/adversarial proposal intensities per slot are $(\lambda_h,\lambda_a)$; 
we write
$$\lambda := \lambda_h + \lambda_a$$
for the total expected number of block proposals per slot.

\medskip\noindent
\textbf{Window discipline.}
Fix a window length $w$ with
\[
  w \;\ge\; \Delta + \omega(\log \kappa)\,.
\]
(In practice, $w$ is chosen large enough to cover the maximum network delay with overwhelming
probability.)

\medskip\noindent
\textbf{Weight policy (short-refs only).}
Only short references to slots in $[t-w, t-1]$ contribute to fork-choice scores.
Long references always carry zero fork-choice weight; they are used solely to maintain DAG connectivity and to 
provide confirmation for blocks.

In the next subsection (\S\ref{sec:tools}) we describe the DAG structure, reference types, and the local fork-choice building blocks used throughout the rest of the paper.

%% file: sec2-tools.tex


\subsection{DAG basics and fork-choice building blocks}
\label{sec:dag-building-blocks}

\label{sec:basics-tools}
\label{sec:basics}
\label{sec:tools}

\subsubsection{DAG structure \& references}
\label{sec:DAG-structure}

\paragraph{Block and DAG Structure (canonical).}
Let $G=(V,E)$ be the reference DAG. For $u,v\in V$, write $u\rightsquigarrow v$ if there is a directed path $u\to\cdots\to v$ in $G$.

\medskip\noindent
\textbf{Block accessors.} Each block $b\in \Blocks$ carries the following fields:
\[
  \id(b)\in\{0,1\}^\kappa,\quad
  \slot(b)\in\mathbb{N},\quad
  \val(b)\in\Parties,\quad
  \payload(b),\quad
\]  \[
  \refs(b)\subseteq V,\quad
  \longref(b)\in V\cup\{\bot\},\quad
  y(b)\in\{0,1\}^\kappa,\quad
  \pi(b),\quad \sigma(b).
\]
Here $y(b)$ and $\pi(b)$ are respectively the output and the proof of a
verifiable random function (VRF) used for stake-proportional
eligibility and public-coin tie-breaking, and $\sigma(b)$ is a digital
signature on the block identifier under the validator's signing key.
The concrete VRF- and signature-based eligibility mechanism is
specified in Section~\ref{sec:notation-base}.

Here $\refs(b)$ are the short references: parents whose slots lie in the current short-reference window $W(t;w)=\{t-w,\ldots,t-1\}$. The value $\longref(b)$, when present, is a single long reference to a parent in a slot $<t-w$. 
Short references are the only edges that contribute to fork-choice weight: 
long references always carry weight $0$ in both the subDAG-weight function $\SubdagW(\cdot; \cdot, w)$ (Section~\ref{sec:fork-choice}) and the conflicted-tips resolution rule $\CTR(\cdot, \cdot)$ (Definition~\ref{def:ctr}; see also Table~\ref{tab:notation}). 
Nevertheless, long references are important for block confirmation: they extend the ancestor closure so that older blocks can still be confirmed via paths combining short and long references.
If no admissible short reference is available, the protocol requires $b$ to include exactly one long reference (which still has weight $0$).

Each block $b$ carries an application payload $\payload(b)$ and a reference set $\refs(b) \subseteq V$.
In this paper we treat $\payload(b)$ as a single opaque application-level object (e.g., a transaction), and we do not model its internal structure or the mempool-sampling policy; see Remark~\ref{rem:block-vs-tx} for a discussion of transaction redundancy and sampling.

In slot $t$, the short-reference window is $W(t;w)=\{t-w,\dots,t-1\}$. Thus, if $b$ is created in slot $t$, every short reference $r\in\refs(b)$ satisfies $1\le t-\slot(r)\le w$, 
so every reference edge points strictly backward in slot and the reference DAG is acyclic.
Long references primarily serve to maintain connectivity and, via ancestor closure, to preserve block confirmation:  when a block cannot attach to a sufficiently recent parent inside $W(t;w)$.

This can happen, for example, during temporary spikes in network latency relative to $w$
(e.g., short-lived routing anomalies or partial network partitions that momentarily delay
message delivery), 
 so that some tips fall outside the
short-reference window and must be reattached via long references. Our analysis, in line
with the Bitcoin backbone and Ouroboros Praos frameworks, assumes a partially
synchronous network with a delay bound~$\Delta$ and parameters chosen so that
$w \ge \Delta$; long references help reconnect and propagate voting weight through the DAG
once communication behaves according to this model again. 

Note that even though long references carry no fork-choice weight (for security and storage efficiency reasons), 
they still ensure that the DAG remains connected. As a result, an older block, while unweighted, can still become indirectly confirmed once it lies in the ancestor closure of some newer block that is short-ref-confirmed; long-refs simply provide the paths that connect such older blocks to the frontier.
All fork-choice and conflict-resolution rules are strictly local: all comparisons are performed relative to the closest common ancestor (CCA) of the competing tips, never relative to a global anchor or the full DAG.

\begin{definition}[Window and Short-Reference Window]\label{def:window}
For slot $t\in\mathbb{N}$ and window $w\in\mathbb{N}_{>0}$, define
\[
W(t;w)\;:=\;\{\tau:\ t-w \le \tau \le t-1\}\quad\text{(the last $w$ previous slots)}.
\]
Short references must target $W(t;w)$, i.e.,
\[
\ShortRefs(d;w)\;=\;\bigl\{\,r\in\refs(d):\ 1 \le \slot(d)-\slot(r) \le w\,\bigr\},\qquad
\wref(d;w)=|\ShortRefs(d;w)|.
\]
\end{definition}

\paragraph{Partial-order growth (protocol view).}
Let $G=(V,E)$ be the reference DAG. A block may be proposed at any time as long as its
references are acyclic and non-conflicting. Define reachability $u \preceq v$ iff $u=v$ or $u \rightsquigarrow v$.
All protocol rules (windowed referencing, scoring, and CCA-based conflict resolution) depend only on $\preceq$.
The protocol imposes no global total order and no linear-progress requirement.

\subsubsection{Ancestors, descendants, views, tips}

\medskip\noindent
\textbf{Ancestor and descendant sets.}
\[
  \Anc(b)\ :=\ \{\,a\in V:\ a\leadsto b\,\},\qquad
  \Desc(b)\ :=\ \{\,d\in V:\ b\leadsto d\,\}.
\]
Write $\AncStar(b)=\Anc(b)\cup\{b\}$ and $\DescStar(b)=\Desc(b)\cup\{b\}$. For a set $X\subseteq V$, define closures $\AncStar(X)=\bigcup_{x\in X}\AncStar(h)$ and $\DescStar(X)=\bigcup_{x\in X}\DescStar(x)$.

\medskip\noindent
\textbf{Views and view-relative ancestry.}
At the end of slot $t$, an honest party $P$ has a local view $V_t^P\subseteq V$. For $b\in V_t^P$,
\[
  \Anc_{V_t^P}(b)\ :=\ \{\,a\in V_t^P:\ a\leadsto b\text{ in the induced subgraph on }V_t^P\,\}.
\]
We use the shorthand $\Anc_V(\cdot)$ when the view $V$ is clear from context.

\paragraph{Windowed subgraph and tips.}
Let $V_w(t) = \{ b \in V^P_t : \slot(b) \ge t - w \}$ and let $G_t^{(w)}$ be the subgraph
induced by $V_w(t)$ with only short-reference edges.
The \emph{short-reference tips} are the blocks in the window that have no short-reference
descendants within the window:
\[
  \Tips_w(G_t)
    := \{\, b \in V_w(t) : \nexists d \in V_w(t) \text{ with a short-ref edge } b \to d \,\}.
\]
By default we write $\Tips_t$ for $\Tips_w(G_t)$; long references neither change tips nor
contribute weight.  
Equivalently, $\Tips_w(G_t)$ is the set of $\preceq$-maximal blocks
in $V_w(t)$ (Remark~\ref{rem:tips-default}).

\begin{remark}[Default meaning of tips]
\label{rem:tips-default}
Unless otherwise stated, ``tips'' always means short-reference tips $\Tips_w(G_t)$. Long
references never affect tip status or fork-choice weight. By construction, $\Tips_w(G_t)$ is
an antichain for the reachability order~$\preceq$: for any distinct $x,y \in \Tips_w(G_t)$ we
have neither $x \preceq y$ nor $y \preceq x$.
\end{remark}

\subsubsection{CCA, \subdagweights, preferred frontier (fork-choice relation)}
\label{sec:fork-choice}

\begin{definition}[Closest Common Ancestor (CCA)]
\label{def:cca}
For conflicting $i,j$, let $\mathrm{CA}(i,j)=\Anc(i)\cap\Anc(j)$.
Define $\CCA(i,j)=\arg\max_{c\in \mathrm{CA}(i,j)}\big(\mathrm{slot}(c),\, y(c)\big)$, 
i.e., the \emph{most recent} common ancestor (ties are broken by $y$). 
\end{definition}
Note that CCA is purely structural: it depends only on ancestry and $(\slot, y)$.
\subdagweights are applied later
inside the conflicted-tips resolution rule $\CTR(\cdot, \cdot)$ (the CCA-local fork-choice comparison between two competing tips; see Definition~\ref{def:ctr}).

\paragraph{Windowed references and subDAG weight.}
Recall from Definition~\ref{def:window} that for a block $d$ and window $w$ we write
\[
  \ShortRefs(d; w) = \{ r \in \refs(d) : 1 \le \slot(d) - \slot(r) \le w \}, \qquad
  \wref(d; w) = |\ShortRefs(d; w)|
\]
for the short references of $d$ that lie within the window.
For conflicting tips $i, j$ with $c = \CCA(i,j)$ and window $w$, define the CCA-local
subDAG weight
\[
  \SubdagW(i; c, w) := \sum_{d \in \Desc^*(i) \cap \Desc^*(c)} \wref(d; w),
\]
and analogously for $j$.

We call two \branches \emph{incomparable} if neither extends from the other, i.e.\ they have no ancestor--descendant relationship.

\medskip
\noindent\textbf{Conflict and compatibility among tips.}  
Two tips $x,y \in \Tips_w(G_t)$ are said to be \emph{conflicting} if their payloads
conflict according to the application-level predicate $\Conflicts(\cdot,\cdot)$; otherwise
we say they are \emph{compatible}. By Remark~\ref{rem:tips-default}, all tips in $\Tips_w(G_t)$ are
incomparable with respect to the structural reachability order $\preceq$; compatibility
and conflict here refer only to payloads, not ancestry.

\begin{definition}[Conflict Predicate \Conflicts]\label{def:conflicts}
We assume a black-box predicate $\Conflicts(b, H)$ that returns true iff the payload of block $b$ conflicts with the payloads of some blocks in $H$, according to the underlying application.  
We do not need to open this predicate in the consensus analysis.
\end{definition}

\begin{definition}[Conflicted-tips resolution]\label{def:ctr}
Let $b,b'$ be two incomparable short-reference tips in $\Tips_w(G_t)$
whose payloads conflict. Let $c := \CCA(b,b')$ and let $w$ be the
short-reference window length. Define their window-filtered \subdagweights as
\[
  W_b := \BranchW(b; c, w),\qquad
  W_{b'} := \BranchW(b'; c, w).
\]
The \emph{conflicted-tips resolution} rule
$\CTR(b,b') \in \{b,b'\}$ is:
\[
  \CTR(b,b') :=
  \begin{cases}
    b  & \text{if } W_b > W_{b'},\\[2pt]
    b' & \text{if } W_{b'} > W_b,\\[2pt]
    \text{the tip with larger } y(\cdot)
       & \text{if } W_b = W_{b'},
  \end{cases}
\]
where $y(\cdot)$ is the public-coin label used for deterministic
tie-breaking.
\end{definition}

\begin{definition}[Fork-choice preference and preferred frontier]
\label{def:preferred-frontier}
\label{def:undom}
Let $\Tips_w(G_t)$ be the short-reference tips at the end of slot $t$ and fix a window
length $w$.

\smallskip
\noindent
\textbf{Fork-choice preference on conflicting tips.}
For any pair of conflicting tips $x, y \in \Tips_w(G_t)$, we write $x \succ y$ if
$\CTR(x,y) = x$ and $y \succ x$ if $\CTR(x,y) = y$. The relation $\succ$ is left
undefined on non-conflicting pairs, which therefore remain incomparable.

\smallskip
\noindent
\textbf{Preferred frontier.}
The preferred frontier at time $t$ is the set of tips that are $\succ$-maximal within
their conflict class:
\[
  F_t := \{\, x \in \Tips_w(G_t) : \nexists\, y \in \Tips_w(G_t) \text{ such that
  $y$ conflicts with $x$ and } y \succ x \,\}.
\]
Equivalently, $F_t$ consists of the $\succ$-maximal tips in each conflict class of
$\Tips_w(G_t)$. Since $\succ$ is conflict-local and never compares compatible tips,
many non-conflicting tips may simultaneously lie in $F_t$.
\end{definition}

Operationally, the preferred frontier $F_t$ is realized by Algorithms~\ref{alg:ideal-local-fc} and~\ref{alg:base-local-fc} in the
idealized and practical protocols, respectively: they start from the short-reference tip
set $U_t = \Tips_w(G_t)$ and apply the conflicted-tips resolution rule $\CTR(\cdot,\cdot)$
(Definition~\ref{def:ctr}) to conflicting pairs until no conflicts remain, at which point the
remaining tips form $F_t$.  
Throughout the paper, a \emph{preferred subDAG} means any
subDAG that contains a tip from $F_t$; several such non-conflicting preferred subDAGs
may coexist in parallel.

\begin{remark}[Notation: structure vs.\ preference]
We use $\preceq$ (and its strict form $\prec$) exclusively for the structural reachability relation on the DAG: $u \preceq v$ means that $u$ is an ancestor of $v$ (or $u = v$). In contrast, the symbol $\succ$ is reserved for the fork-choice preference relation on tips, which is defined in Section~2.2.3 via CCA-local subDAG weights and the conflicted-tips resolution rule $\CTR(\cdot,\cdot)$. These relations are unrelated: $\preceq$ expresses ancestry, whereas $\succ$ expresses fork-choice preference among conflicting tips.
\end{remark}

\subsubsection{Short-ref and block confirmations; deterministic analysis \& anchors}
\label{sec:tx-confirmation}

We conceptually separate fork-choice weight from (logical) block confirmation. In a general deployment, a block would carry a batch of transactions, and a transaction is considered confirmed once the block that contains it lies in the ancestor closure of some visible block. In our single-payload-per-block model (Remark~\ref{rem:block-vs-tx}), each block stands for one opaque application payload (e.g., one transaction), so block-level confirmation and transaction-level confirmation coincide.

\paragraph{Confirming descendants and fork-choice weight.}
Short references drive fork-choice weight; confirmation is purely about reachability. A short-ref-confirming descendant of a block is exactly an edge that contributes to fork-choice weight under our rule. Long references never contribute to fork-choice weight. They are used solely to maintain connectivity and to ensure that older blocks---and thus their payloads---eventually lie in the ancestor closure of some visible block, along a path that may mix short and long references.

\paragraph{Deterministic linearization (analysis view).}
For the analysis in Section~\ref{sec:finality} we fix a deterministic linearization of an anchored subDAG. 
Given an anchor $a$ (e.g., a finalized checkpoint) and a window $w$, let $F_t$ be the preferred frontier under the local CCA-based fork-choice rule. We write $\Lin_a(G_t)$ for any topological order consistent with $\preceq$ on $\Anc^\ast(F_t) \cap \Desc^\ast(a)$, breaking ties by the public coin $y(\cdot)$. This linearization is used only as an analytical device to define ledger-style properties; it is not part of the protocol state or of any client operation. Anchors and the associated linearization appear only in the analysis of Section~\ref{sec:finality}; the protocol exports only a DAG (partial order).
For the remainder of the paper, we regard the analysis anchor $a$ as fixed but
otherwise arbitrary whenever we refer to $\Lin_a(G_t)$.

\begin{remark}[Protocol vs.\ analysis layer]\label{rem:protocol-vs-analysis}
The consensus protocol layer exports only a DAG and the preferred frontier $F_t$: 
block validity, short-reference windowing, and CCA-local fork choice are fully specified without any on-wire anchor. Anchors appear only in the security analysis of Section~\ref{sec:finality}: given a DAG view $G_t$, we fix a deterministic analysis-level linearization map (anchored linearization together with the $k$-depth finality predicate) that defines a ledger abstraction used only in the proofs. This analysis-level linearization is a deterministic projection of the DAG and cannot introduce or remove any consensus decisions. It does not modify protocol state nor block validity and is best thought of as a mathematical device, analogous to the ``$k$ confirmations'' convention in chain protocols. Actual client implementations may expose any view that is consistent with the DAG's partial order; we do not fix a specific client interface in this paper.
\end{remark}

\paragraph{Design summary (fork choice \& references).}
\begin{itemize}
\item
Fork choice is local and window-filtered: short references within the last w slots carry unit weight; long references always carry weight 0.
\item 
Fork choice is set-valued: 
instead of tracking the tip of a single linear chain as in longest-chain protocols, it maintains a preferred frontier \(F_t\) of \(\succ\)-maximal tips in each conflict class.

\item All comparisons between competing subDAGs are CCA-local: 
when comparing two tips $i,j$ with $c = \CCA(i,j)$, the subDAG weight of each tip is computed only from blocks in the region between $c$ and that tip, i.e., blocks that are descendants of $c$ and ancestors of the tip.

\end{itemize}

%% file: sec2-model.tex


\subsection{Security Model}
\label{sec:model}

We adopt the bounded-delay (partial synchrony) model: there exists an unknown constant $\Delta$ such that any message broadcast by an honest party in slot $t$ is delivered to every honest party by the end of slot $t+\Delta$. Parties need not know $\Delta$ at runtime; security parameters (e.g., window length $w$ and finality depth $k$) are calibrated as functions of $\Delta$ in the analysis.

The adversary $\mathcal{A}$ controls at most an $\alpha<\tfrac12$ fraction of stake, may adaptively corrupt parties, schedule messages subject to the $\Delta$-delay constraint, and inject blocks for corrupted parties. All honest parties start from a common genesis, with known public keys and initial stake allocation.

\paragraph{Execution.}
In each slot $t$, every honest party $P$ collects all messages delivered by time $t$, verifies them, checks its VRF-based eligibility, and if eligible, creates a block and broadcasts it. We write $V_t^P$ to denote $P$'s local DAG view at the end of slot $t$, and $L_t^P$ for the ledger induced by $P$'s fork choice at time $t$ (window-filtered as defined in Section~\ref{sec:basics-tools}). Formally, $L_t^P$ is an analysis-level abstraction of the DAG view $V_t^P$; the protocol state consists only of the DAG and the local fork-choice rule.

\subsubsection{Consensus Properties for DAGs}
\label{sec:DG-DQ-DCP}

\noindent
\textbf{Scope.} In this paper, DAG properties (DG/DQ/\DCP/TB) are evaluated with respect to the  \emph{preferred frontier} $\Undom_t$ 
(Def.~\ref{def:undom}) under the local CCA-based rule. Anchors and linearization are used only in Section~\ref{sec:finality} as an analysis device and are not part of protocol state. Thus, when reasoning about DG, DQ, \DCP, and TB, it suffices to work with the local fork choice and $\Undom_t$; only when we state ledger-style properties do we additionally consider the anchored linearization induced by $\Undom_t$. That linearization is fixed for the proofs and does not represent an operation that honest nodes or clients must perform.

\begin{figure}[htbp!]
\centering
\begin{tikzpicture}[x=1cm,y=0.8cm,>=Latex,thick]

  \fill[black] (0,0) circle (2pt);
  \node[right=0.4cm of {(0,0)}] {\scriptsize Generic block};

  \fill[green!40!black] (0,-1) rectangle (0.4,-0.6);
  \node[right=0.6cm of {(0.4,-0.8)}] {\scriptsize Honest/finalized prefix};

  \fill[blue!30] (0,-2) rectangle (0.4,-1.6);
  \node[right=0.6cm of {(0.4,-1.8)}] {\scriptsize Frontier / last $k_D$ layers};

  \fill[black!20] (0,-3) rectangle (0.4,-2.6);
  \node[right=0.6cm of {(0.4,-2.8)}] {\scriptsize Common past region};

  \fill[magenta!90] (0,-4) circle (2pt);
  \node[right=0.4cm of {(0,-4)}] {\scriptsize Block broadcast / inclusion};

\end{tikzpicture}
\caption{\textbf{Legend of colors used in Figures~\ref{fig:dg}--\ref{fig:liveness}.} 
Black = blocks; green = honest or finalized prefix; blue = frontier region; gray = common past; red = broadcast block event or inclusion.}
\label{fig:legend}
\end{figure}

\subsubsection*{DAG Growth (DG)}
\label{sec:DG}

In any window of $\ell$ slots, honest work ensures steady growth of at least one honest \prefsubdag (the DAG analogue of the chain-growth property in a linear blockchain, despite the DAG's parallelism). That is, at least $\tau_D \cdot \ell$ honest blocks created during that window end up on the past of an honest party's fork-choice tip (i.e., on an honest \prefsubdag).

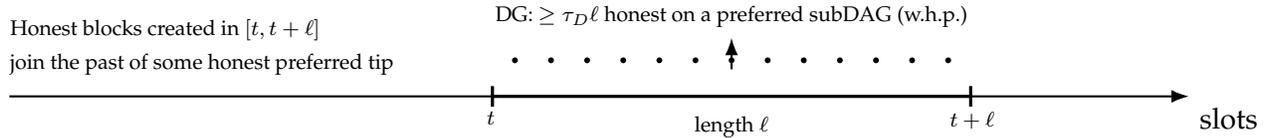
\begin{figure}[htp!]
\centering
\begin{tikzpicture}[x=0.16cm,y=0.6cm,>=Latex,thick]
  \draw[-{Latex}] (0,0) -- (98,0) node[below right] {slots};

  \draw[very thick,|-|] (40,0) -- (80,0);
  \node[below] at (40,-0.1) {\scriptsize $t$};
  \node[below] at (80,-0.1) {\scriptsize $t+\ell$};
  \node[below] at (60,-0.2) {\scriptsize length $\ell$};

  \foreach \x in {42,45,48,51,54,57,60,63,66,69,72,75,78} {
    \fill[black] (\x,0.8) circle (1.1pt);
  }
  \node[align=left] at (16,1.1) {\scriptsize Honest blocks created in $[t,t+\ell]$ \\[-1pt] \scriptsize join the past of some honest preferred tip};

  \draw[->] (60,0.6) -- (60,1.3);
  \node[above] at (60,1.3) {\scriptsize DG: $\ge \tau_D\ell$ honest on a \prefsubdag (w.h.p.)};
\end{tikzpicture}
\caption{\textbf{DAG Growth (DG).}
\small 
Over any $\ell$-slot interval, with probability $1-\negl(\secp)$ at least $\tau_D\ell$ honest blocks become ancestors of \emph{an honest party's fork-choice tip} (i.e., lie on some honest \prefsubdag) by the end of the interval. Black dots represent honest blocks produced during the interval.}
\label{fig:dg}
\end{figure}

\begin{definition}[DAG Growth]\label{def:DG}
Fix a security parameter $\secp$ and let $\negl(\cdot)$ denote a negligible function. Let $\Delta$ be the network delay bound. For any interval $I=[t,t+\ell-1]$ of $\ell$ consecutive slots , let $H_I$ be the set of honest blocks created in $I$ that, by the end of slot $t+\ell-1+\Delta$, are contained in the ancestor set of at least one honest party's fork-choice tip (equivalently, they lie on some honest party's \prefsubdag). We say $\mathrm{DG}[\tau_D,\ell]$  holds if for some constant $\tau_D\in(0,1]$ and for all such intervals
\[
\Pr\!\big[\,|H_I| \ge \tau_D\cdot \ell\,\big] \ge 1-\negl(\secp)\,.
\]
\end{definition}

\paragraph{Remark (how we measure progress).}
In a multi-proposer DAG, many honest tips can coexist and grow in parallel: the protocol state is \emph{partially ordered}. We therefore \emph{measure} steady ledger progress in a derived linearized view used only in the analysis: when we linearize an anchored \subdag, at least a $\tau_D$ fraction of honest blocks created over any $\ell$-slot interval appear in that linearization with high probability (DG-Lin). This focuses progress on some honest \prefsubdag \emph{in the analytical view}, without imposing global linearity on the protocol or constraining client behavior.

\paragraph{DG in the analysis view (DG-Lin).}
Fix any $\ell$-slot interval $I = [t, t + \ell - 1]$ and let $H_I$ be the
honest blocks created in $I$. At time $t' = t + \ell - 1 + \Delta$, consider
the anchored analysis-level linearization $\Lin_a(G_{t'})$ from
Section~\ref{sec:tx-confirmation}, where $a$ is the fixed analysis anchor introduced there
(e.g., a finalized checkpoint).
We say $\mathrm{DG\mbox{-}Lin}[\tau_D, \ell]$ holds if, except with negligible
probability,
\[
  \bigl|\{\, b \in H_I : b \text{ appears in } \Lin_a(G_{t'}) \,\}\bigr|
  \;\ge\; \tau_D \cdot \ell.
\]

\paragraph{Equivalence to the original DG (informal).}
Under our window-limited, CCA-local fork choice (and assuming TB), each
honest party’s preferred frontier $F_t$ at time $t$ induces a canonical
ledger-style view of its preferred subDAG: applying the anchored map
$\Lin_a(\cdot)$ from Section~2.2.4 to its local DAG view $G_t$ (with the
same fixed analysis anchor $a$) yields a linearization $\Lin_a(G_t)$.
In the original DG formulation, an honest block is counted as “grown”
if it lies on the past of some honest preferred tip, which is equivalent
to saying that it appears in $\Lin_a(G_t)$ for every honest party’s view
at the corresponding time.
Conversely, any block that is absent from all honest preferred subDAGs
cannot appear in the anchored linearization used for the analysis-level
ledger view.
Thus DG and DG-Lin differ only in presentation: DG is phrased in terms
of per-party preferred-tip views $\Lin_a(G_t)$, whereas DG-Lin is phrased
in terms of the global anchored linearization $\Lin_a(G_{t'})$ once the
interval $I$ has been fully delivered.

\subsubsection*{DAG Quality (DQ)}
\label{sec:DQ}

In any window of $\ell$ slots, the fraction of honest blocks among all blocks entering the “recent frontier” (within the last $w$ slots on the \prefsubdag) is at least $\mu_D$.

\begin{definition}[DAG Quality]\label{def:DQ}
Fix a security parameter $\secp$. For any interval $I=[t,t+\ell-1]$ , consider the set $B_I$ of blocks that enter the \prefsubdag's frontier (i.e., become descendants of the fork-choice tip) during $I$. Let $B_I^{\mathsf{H}}$ be the subset of $B_I$ created by honest parties. We say $\mathrm{DQ}[\mu_D,\ell]$ holds if
\[
\Pr\big[\,|B_I^{\mathsf{H}}| \ge \mu_D \cdot |B_I|\,\big] \ge 1-\negl(\secp)\,. 
\]
\end{definition}

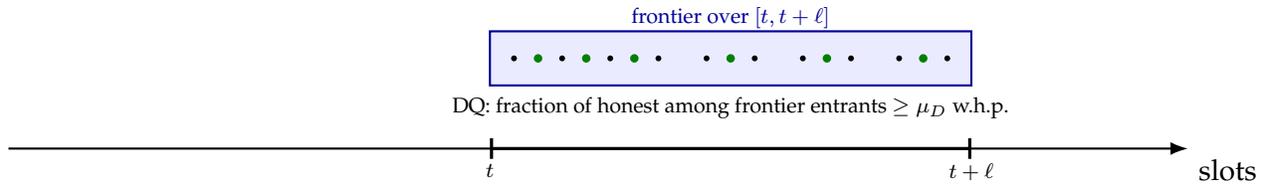
\begin{figure}[htp!]
\centering
\begin{tikzpicture}[x=0.16cm,y=0.6cm,>=Latex,thick]
  \draw[-{Latex}] (0,0) -- (98,0) node[below right] {slots};

  \fill[blue!8] (40,1.4) rectangle (80,2.6);
  \draw[blue!60!black] (40,1.4) rectangle (80,2.6);
  \node[blue!60!black] at (60,2.9) {\scriptsize frontier over $[t,t+\ell]$};

  \foreach \x in {42,46,50,54,58,62,66,70,74,78} {
    \fill[black] (\x,2.0) circle (1.1pt); 
  }
  \foreach \x in {44,48,52,60,68,76} {
    \fill[green!50!black] (\x,2.0) circle (1.6pt); 
  }

  \draw[very thick,|-|] (40,0) -- (80,0);
  \node[below] at (40,-0.1) {\scriptsize $t$};
  \node[below] at (80,-0.1) {\scriptsize $t+\ell$};

  \node[align=center] at (60,0.9) {\scriptsize DQ: fraction of honest among frontier entrants $\ge \mu_D$ w.h.p.};
\end{tikzpicture}
\caption{\textbf{DAG Quality (DQ).} 
\small 
In any $\ell$-slot window, the fraction of \emph{honest blocks} among those entering the advancing frontier is at least $\mu_D$ with probability $1-\negl(\secp)$. In the illustration, the blue band marks the advancing frontier; black dots represent all blocks that enter this frontier, and green dots highlight the subset created by honest parties.}
\label{fig:dq}
\end{figure}

\subsubsection*{DAG Common Past (\DCP)}
\label{sec:DCP}

For any two honest views of the DAG at times $t_1\le t_2$, if we remove the last $k_D$ layers from the current frontier, their remaining pasts are identical. Equivalently, all honest parties share the same set of ancestor blocks up to the most recent CCA of their tips.

\medskip

\noindent\textbf{Layers.} A \emph{layer} is the set of blocks with the same slot index. Trimming the last $k_D$ layers removes all blocks whose slot lies in the most recent $k_D$ slots.
Equivalently, all honest parties share the same set of ancestor blocks up to the most recent CCA of their tips: their trimmed pasts, denoted $\Past_P(t_1)$ and $\Past_Q(t_2)$, coincide once we delete the last $k_D$ layers.

\begin{definition}[DAG Common Past (\DCP)]
\label{def:DCPprime}
\label{def:DCP}

Let $\Undom_t$ denote the preferred frontier of honest nodes at times $t$  (Definition~\ref{def:preferred-frontier}),
and define 
$\Past(t):=\AncStar(\Undom_t)$.
There exists an integer $k_D \ge 0$ such that for all honest parties $P,Q$ and 
times $t_1\le t_2$,
after removing the most recent $k_D$ layers from $\Past_P(t_1)$ and 
$\Past_Q(t_2)$,
the former is a subset of the latter, except with negligible probability.
\end{definition}

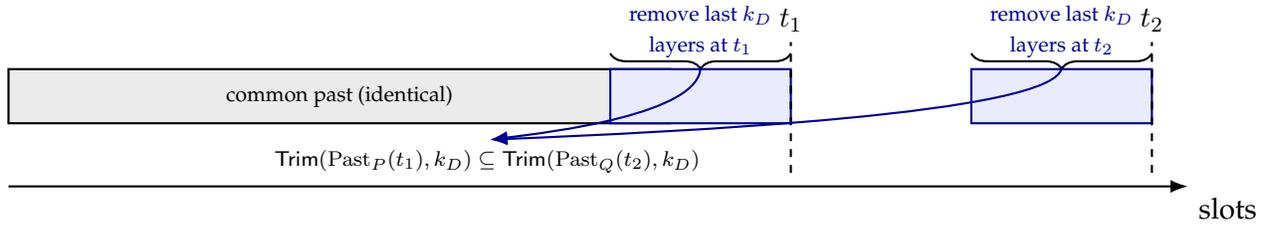
\begin{figure}[htbp!]
\centering
\begin{tikzpicture}[x=0.16cm,y=0.6cm,>=Latex,thick]
  \draw[-{Latex}] (0,0) -- (98,0) node[below right] {slots};

  \fill[black!8] (0,1.4) rectangle (55,2.6);
  \draw (0,1.4) rectangle (55,2.6);
  \node at (27.5,2.0) {\scriptsize common past (identical)};

  \fill[blue!8] (50,1.4) rectangle (65,2.6);
  \draw[blue!60!black] (50,1.4) rectangle (65,2.6);
  \node[blue!60!black] at (57.5,3.1) {\scriptsize 
  layers at $t_1$
  };

  \fill[blue!8] (80,1.4) rectangle (95,2.6);
  \draw[blue!60!black] (80,1.4) rectangle (95,2.6);
  \node[blue!60!black] at (87.5,3.1) {\scriptsize 
  layers at $t_2$
  };

  \draw[dashed] (65,0.3) -- (65,3.2) node[above] {$t_1$};
  \draw[dashed] (95,0.3) -- (95,3.2) node[above] {$t_2$};

  \draw[decorate,decoration={brace,amplitude=6pt}] (65,2.95) -- (50,2.95)
    node[midway,above=7pt,blue!60!black] {\scriptsize remove last $k_D$ };
  \draw[decorate,decoration={brace,amplitude=6pt}] (95,2.95) -- (80,2.95)
    node[midway,above=7pt,blue!60!black] {\scriptsize remove last $k_D$};

  \node[align=center] (trimEq) at (40,0.55)
    {\scriptsize 
    $\Trim(\Past_P(t_1), k_D) \subseteq \Trim(\Past_Q(t_2), k_D)$
    };

  \draw[->,blue!60!black] (57.5,2.65) .. controls (57.5,1.6) and (45,1.2) .. (trimEq.north);
  \draw[->,blue!60!black] (87.5,2.65) .. controls (87.5,1.6) and (55,1.2) .. (trimEq.north);
\end{tikzpicture}
\caption{\textbf{DAG Common Past (\DCP).} 
The \textcolor{black!60!white}{gray region} shows the common past that all honest parties agree on. 
For two honest parties $P$ and $Q$ at times $t_1\le t_2$,  let $\Past_P(t_1)$ and $\Past_Q(t_2)$ be the ancestor sets of their  preferred frontiers (Def.~\ref{def:DCPprime}). 
The \textcolor{blue!60!black}{blue bands} depict the most recent $k_D$ layers adjacent to these frontiers. After removing these $k_D$ frontier layers, the earlier trimmed past is contained in the later one, i.e., $\Trim(\Past_P(t_1), k_D) \subseteq \Trim(\Past_Q(t_2), k_D)$, with probability $1-\negl(\secp)$. 
}

\label{fig:dcp}
\end{figure}

\begin{remark}[Mergeable views] 
DCP does not require identical DAGs; honest views may be mergeable (no ledger conflicts) and coalesce as messages arrive.
\end{remark}

\begin{remark}[Consistency of DAG properties]
Definitions~\ref{def:DG} (DAG Growth),~\ref{def:DQ} (DAG Quality), 
and~\ref{def:DCPprime} (DAG Common Past (\DCP)) 
are all stated relative to an honest party's preferred frontier $\Undom_t$ and its past $\Past(t) = \AncStar(\Undom_t)$. 
This ensures consistency across the properties: \emph{DG} guarantees continued growth along at least one honest \prefsubdag
in the DAG,  \emph{DQ} guarantees that honest blocks dominate the frontier of that \branch,   
and 
\DCP
guarantees that honest parties' \prefsubdags have a long common past. Together, these properties provide a coherent DAG analogue of the chain-growth, chain-quality, and common-prefix properties from backbone analyses.
\end{remark}

\subsubsection{Ledger Properties}
\label{sec:ledger-properties}

\subsubsection*{Persistence (Safety)}
\label{sec:safety}
A protocol has \emph{persistence} with parameter $k$ if for any two honest times $t_1\le t_2$, the ledgers produced by the fork-choice rule at those times agree on all blocks that were at least $k$-deep at time $t_1$, except with probability $\negl(\secp)$.

\begin{figure}[htp!]
\centering
\begin{tikzpicture}[x=0.16cm,y=0.6cm,>=Latex,thick]
  \draw[-{Latex}] (0,0) -- (98,0) node[below right] {slots};

  \draw[dashed] (70,-2.0) -- (70,3.0) node[above] {$t_1$};
  \draw[dashed] (90,-2.0) -- (90,3.0) node[above] {$t_2$};

  \fill[green!12] (0,1.4) rectangle (45,2.6);
  \draw[green!40!black] (0,1.4) rectangle (45,2.6);
  \node[green!40!black,anchor=west] at (0.5,2.0) {\scriptsize $k$-deep prefix at $t_1$};

  \draw[green!40!black,dashed] (0,1.4) rectangle (45,2.6);
  \node[green!40!black] at (22.5,1.0) {\scriptsize remains identical at $t_2$};

  \draw[->] (46,2.0) -- (68,2.0) node[midway,above] {\scriptsize immutable region};
\end{tikzpicture}
\caption{\textbf{Persistence (Safety).} 
\small 
All but the last $k$ blocks in the ledger at time $t_1$ remain immutable and appear in every honest ledger at any later time $t_2\ge t_1$, except with probability $1-\negl(\secp)$. The green region shows the $k$-deep finalized prefix at $t_1$, which remains unchanged in later ledgers.}
\label{fig:persistence}
\end{figure}
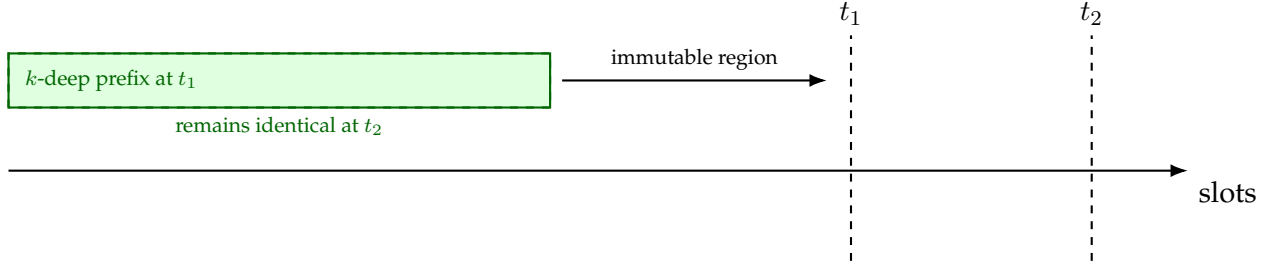

\begin{definition}[Persistence (Safety)]
Fix $k \in \mathbb{N}$. For an honest party $P$ and slot $t$, let $L_t^P$ denote the ledger 
 induced by applying the fork-choice rule to $V_t^P$ 
 (i.e., the sequence of blocks confirmed by time $t$ in $P$'s view).
 We say $\mathrm{Persistence}[k]$ holds if for all honest parties $P$, $Q$ and all slots $t_1 \le t_2$, 
 \[
\Pr\!\Big[\, L_{t_1}^P[1..n-k] = L_{t_2}^Q[1..n-k] \,\Big] \ge 1-\negl(\secp)\,,
\]
where $n = |L_{t_1}^P|$. 
\end{definition}

\begin{remark}[On DAG Growth and Persistence]
The proof of persistence  relies on DG to guarantee that honest blocks accumulate at a linear rate 
\emph{along some honest \prefsubdags}. Once a block is $k$-deep under a fork-choice tip of one such honest \prefsubdag, the honest margin behind it grows steadily, preventing any \adversarialsubdag from overturning it.
\end{remark}

\subsubsection*{Liveness}
\label{sec:liveness}

A protocol has \emph{liveness} with parameter $\ell_{\mathrm{live}}$ if any valid block broadcast continuously by an honest party  appears in the ledger of every honest party within $\ell_{\mathrm{live}}$ slots, except with probability $\negl(\secp)$.

\begin{figure}[htp!]
\centering
\begin{tikzpicture}[x=0.16cm,y=0.6cm,>=Latex,thick]
  \draw[-{Latex}] (0,0) -- (98,0) node[below right] {slots};

  \draw[very thick,|-|] (40,-0.8) -- (75,-0.8);
  \node[below] at (40,-0.84) {\scriptsize $t$};
  \node[below] at (75,-0.84) {\scriptsize $t+\ell_{\mathrm{live}}$};
  \node[below] at (57.5,-0.9) {\scriptsize inclusion window};

  \fill[magenta!90] (40,-0.8) circle (1.3pt);
  \node[below=4pt,magenta!90] at (40,-0.98) {\scriptsize $\block$ broadcast repeatedly};

  \draw[->,magenta!90] (68,-0.3) -- (68,1.0);
  \node[magenta!90,align=center] at (68,1.6) {\scriptsize $\block$ appears in $L_t^Q$\\[-1pt]\scriptsize for some $t'\in[t,t+\ell_{\mathrm{live}}]$};
\end{tikzpicture}
\caption{\textbf{Liveness.} 
\small
Any valid block $\block$ broadcast continuously by an honest party  is included in every honest ledger within $\ell_{\mathrm{live}}$ slots, except with probability $\negl(\secp)$. {\color{magenta}{Magenta}} marks the block broadcast event at slot $t$ and the inclusion window $[t,t+\ell_{\mathrm{live}}]$, with an arrow showing when the block enters an honest ledger.}
\label{fig:liveness}
\end{figure}
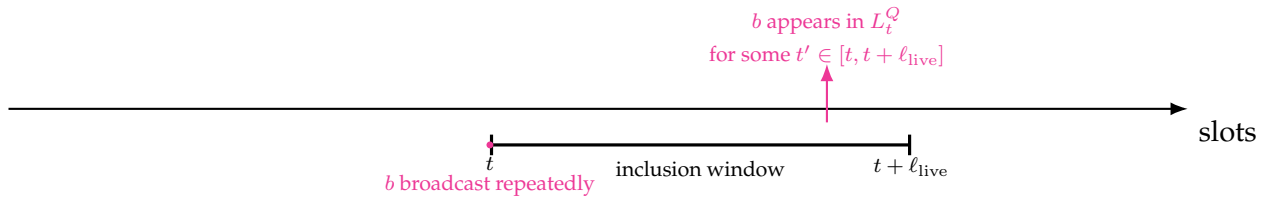

\begin{definition}[Liveness]
Fix $\ell_{\mathrm{live}} \in \mathbb{N}$. For any valid block $\block$ that is broadcast by an honest party $P$ in every slot $t$, let $\cE_{t'}^Q(\block)$ denote the event that honest party $Q$'s ledger $L_{t'}^Q$ at the end of slot $t'$ contains $\block$. We say $\mathrm{Liveness}[\ell_{\mathrm{live}}]$ holds if for all honest parties $P,Q$ and for all slots $t$,
\[
\Pr\!\Big[\, \exists\, t' \in [t,\,t+\ell_{\mathrm{live}}] : \cE_{t'}^Q(\block) \,\Big] \ge 1-\negl(\secp)\,. 
\]
That is, any block repeatedly injected by an honest party  is guaranteed to appear in every honest party's ledger within $\ell_{\mathrm{live}}$ slots, except with negligible probability.
\end{definition}

\begin{remark}[On DAG Growth and Liveness]\label{rem:DG-liveness}
Liveness requires that valid blocks eventually appear
in all honest ledgers. The revised DG property, formulated with respect to preferred subDAGs, ensures
continued growth along at least one of the honest preferred subDAGs. Together with TB (which prevents
the preferred frontier $F_t$ of honest nodes from fragmenting indefinitely) and DQ (which guarantees that
honest blocks dominate this frontier), this implies that any valid block that is repeatedly broadcast will be
absorbed into the preferred frontier of every honest party and become confirmed within bounded time.
\end{remark}

\subsubsection{Tip Boundedness and Its Necessity}
\label{sec:TB}

The preceding properties (DAG Growth, DAG Quality, DAG Common Past) capture progress, fairness, and eventual agreement on history.
These properties ensure that the underlying DAG admits a secure ledger abstraction. 
 However, these three alone are insufficient to guarantee a secure ledger. We introduce the additional invariant of \emph{Tip Boundedness} (TB).

\begin{definition}[Tip-Boundedness (TB)]
\label{def:TB}
Fix a security parameter $\secp$. Let $\Tips_t$ denote the set of tips visible to an honest party at the end of slot $t$. We say $\TB[\beta]$ holds if there exists a bound $\beta = \beta(H,\Delta,\lambda,w)$ such that for all $t $,
\[
\Pr[\,|\Tips_t| \le \beta\,] \ge 1-\negl(\secp)\,.
\]
Here $H$ is the honest stake (or eligibility) fraction, $\Delta$ is the network delay bound, $\lambda$ denotes the expected number of eligible block proposals per slot (parallelism), i.e., $\lambda=\lambda_h+\lambda_a$, 
 and $w$ is the short-reference window length.
\end{definition}

\begin{remark}[TB's role in finality (\S\ref{sec:finality})]\label{rmk:tb-to-finality}
We invoke Tip‐Boundedness (TB) in \S\ref{sec:finality} to bound the per‐window increments of the gap process (Def.~\ref{def:gap}), which enables the concentration step in the master $(k,\varepsilon)$-finality theorem (Thm.~\ref{thm:master-finality}) via standard Azuma--Hoeffding and Chernoff bounds. Concretely, TB together with $w \!\ge\! \Delta$ ensures that short-reference contributions per window are bounded, yielding a submartingale with bounded differences.
\end{remark}

\paragraph{Roadmap for TB.}
We establish TB in both models and invoke it in Persistence/Liveness:
\begin{itemize}
\item \textbf{Ideal model.} Lemma~\ref{lem:TB-ideal} proves $\TB[\beta]$ with $\beta = O(\lambda)$.
\item \textbf{Practical model.} Lemma~\ref{lem:TB-base} proves $\TB[\beta]$ with $\beta = O((\lambda_h+\lambda_a)\cdot\Delta)$ (equivalently $O(\lambda\Delta)$), assuming $w \ge \Delta$. 
\end{itemize}

\begin{remark}[Use of TB in Liveness and Persistence]
\label{rem:TB-usage}
TB bounds the frontier width and thus lower-bounds the per-tip confirmation rate. In the ideal-model and base-model security sections, we cite Lemma~\ref{lem:TB-ideal} or Lemma~\ref{lem:TB-base} when deriving concrete values for $\ell_{\mathrm{live}}$ and $k = k(w,\varepsilon)$. We write $k=k(w,\varepsilon)$ to emphasize its dependence on the window and tail target.
\end{remark}

\begin{remark}[Why this formulation]
Unlike DG, DQ, and \DCP which can be interpreted in a linearized setting, TB explicitly acknowledges the DAG's concurrency: multiple tips may exist at once. Bounding their number ensures that honest blocks regularly \emph{merge or coalesce} tips, preventing indefinite fragmentation of honest work. This makes DG, DQ, and \DCP meaningful in the DAG setting, because they now apply to a bounded set of \prefsubdags rather than an unbounded frontier.
\end{remark}

\paragraph{Why TB Is Necessary.}
If the number of honest tips is unbounded, then:

\begin{itemize}
\item \textbf{Safety problem:} 
Honest work fragments across many incomparable \branches (i.e., components that are incomparable under~$\preceq$). In fork choice, an adversary can concentrate its weight on one \branch to override the others. Even if DG, DQ, and \DCP hold, safety is violated because no \branch accumulates a \emph{stable majority} of weight. TB is necessary to limit the number of competing tips. By itself, TB does not guarantee that one \branch immediately holds a majority (honest weight could still be split across a few tips), but it creates conditions under which honest cross-referencing and fork-choice dynamics ensure that one \branch eventually dominates.

\begin{remark}[Stable majority nuance]
TB is thus a \emph{necessary but not sufficient} condition: it bounds fragmentation, and in combination with DG and DQ it yields eventual honest majority on one \branch, assuming honest nodes cross-reference tips once visible.
\end{remark}

\item \textbf{Liveness under wide frontiers (latency/throughput view):} When the number of incomparable honest tips grows, honest referencing effort spreads across more \branches. Even if every honest block is valid, short-reference weight is diluted: the per-tip confirmation rate drops, increasing the expected time-to-first-confirmation for some blocks. Under partial synchrony, this can push latency beyond any fixed $\ell_{\mathrm{live}}$ unless the frontier width is bounded. \emph{Tip Boundedness (TB)} enforces such a bound; with an honest block rate $\lambda_h$ and per-block short-reference probability $q>0$, the probability that a block has at least one confirming descendant by time $t$ is $\ge 1 - e^{-q\lambda_h t}$. Choosing $t = \lceil (1/(q\lambda_h)) \ln(1/\varepsilon)\rceil$ achieves the target liveness with failure probability at most $\varepsilon$.

\begin{remark}[TB across models]
Under synchrony, the tip bound $\beta$ scales with the number of blocks per slot, i.e., $\beta = \Theta(p\cdot|\mathcal{P}|)$ when each validator is eligible with probability $p$. Here $|\mathcal{P}|$ denotes the number of validators (see Table~\ref{tab:notation}).

Under partial synchrony with maximum delay $\Delta$, tips can accumulate over $\Delta$ slots before propagation equalizes views, giving $\beta = \Theta(\Delta\cdot p\cdot|\mathcal{P}|)$. This reconciles the ideal and practical statements and makes explicit the dependence on delay and issuance rate.
\end{remark}

\begin{remark}\label{rem:TX}\label{rem:block-vs-tx}
Transaction redundancy and sampling---deferred extension. In this paper we focus on consensus at the level of blocks. For simplicity, we assume that each block contains a single application payload (e.g., a transaction) and we treat block validity and conflicts at the block level.

A natural extension supports multiple transactions per block together with a mempool-sampling policy and tunable redundancy. Such an extension must ensure that transaction-level conflicts are handled in a way that is consistent with the block-level fork choice and finality analysis developed here, and that adversarially crafted conflicts cannot cause entire blocks to be discarded. Detailing such a mechanism is beyond the scope of this paper and is left for future work.
\end{remark}

\end{itemize}

\begin{lemma}[Necessity of TB for block-confirmation liveness]
\label{lem:TB-necessity}
If Tip-Boundedness (TB) does not hold, then there exists an admissible adversarial schedule and an honest block $h$ such that, with positive probability, $h$ never acquires any confirming (short-ref) descendants. Equivalently, the time until $h$ obtains its first short-reference-confirming descendant may be infinite, even though $h$ remains valid and is eventually gossiped to all nodes. Thus TB is necessary for the short-reference-based confirmation guarantees used in our liveness analysis.
\end{lemma}

\begin{proof}[Proof sketch]
Without TB, the adversary can keep the tip set size $T_t$ unbounded by releasing incomparable blocks that omit $h$ from their short references. Each honest proposer can include at most $R$ short references, so conditioned on the current tip set the probability that some honest block short-references $h$ in slot $t$ is at most $R/T_t$. The adversary can choose a schedule in which the sequence $(T_t)$ grows (or re-grows) so fast that $ \sum_{t=1}^\infty \frac{R}{T_t} < 1$.
By a simple union bound,
$$
  \Pr\big[\text{$h$ is ever short-referenced}\big]
  = \Pr\Big[\bigcup_{t \ge 1} \{\text{$h$ is short-referenced in slot $t$}\}\Big]
  \le \sum_{t \ge 1} \frac{R}{T_t} < 1.
$$
Hence with probability at least $1 - \sum_t R/T_t > 0$, no future block ever short-references $h$, so
$h$ never acquires any short-ref-confirming descendants.

\end{proof}

\begin{remark}[Validity vs.\ short-ref confirmation]
A block lacking short-ref-confirming descendants is still valid and, thanks to long references, will eventually
lie in the ancestor closure of some visible tip in every honest view (Section~\ref{sec:tx-confirmation}). Lemma~\ref{lem:TB-necessity} therefore
does not contradict eventual logical confirmation via long references; it isolates a failure mode of the
short-reference–driven fork-choice metric. TB rules out adversarial schedules in which an honest block never
acquires any short-ref-confirming descendants and hence never becomes $k$-deep under the preferred frontier,
which is the notion of confirmation used in our ledger liveness analysis.
\end{remark}

\paragraph{Remark: TB vs.\ \DCP and Relation to Long-Delay Attacks.}
DAG Common Past (\DCP) does \emph{not} imply Tip Boundedness (TB). \DCP ensures that if one trims away the most recent $k_D$ layers, then all honest views coincide on the older prefix. However, \DCP does not constrain how many incomparable tips may exist within those last $k_D$ layers. In particular, one can have a DAG where the deep past is fully common (\DCP holds) yet the active frontier contains $\Theta(\Delta \cdot \lambda)$ tips, where $\lambda = \lambda_h + \lambda_a$ is the total block proposal rate, or even unboundedly many tips. Thus, persistence and liveness may fail even though DG, DQ, and \DCP all hold: the honest work is fragmented across too many \branches to form a stable margin.

This parallels the \emph{long-delay attack} described by Pass, Seeman, and shelat in their analysis of Nakamoto consensus in asynchronous networks \cite[Section~8]{EC:PasSeeShe17}. When network delays are large relative to block production, the honest chain fragments into many concurrent forks, which an adversary can exploit to cause deep reorgs. Our TB property directly precludes such fragmentation in DAG protocols: it enforces a bound on the width of the tip set, thereby ensuring that honest references continually \emph{merge or coalesce} the frontier. In this sense, TB is the DAG analogue of the synchrony assumption needed to rule out long-delay attacks in chain-based protocols.

\begin{remark}[TB as the missing latency/throughput link]
In DAGs, throughput comes from concurrent proposal; without TB, that same concurrency can widen the frontier enough to dilute short-reference weight and stall confirmations. Bounding the tip set ensures that the honest short-reference mass concentrates quickly on a few \branches, restoring a bounded time-to-first-confirmation even at high issuance rates.
\end{remark}

%% file: sec3.tex


\section{\ProjIdeal: Idealized Multi-Proposer Protocol}
\label{sec:ideal}

This section isolates the core consensus dynamics of \ProjIdeal and explains
\emph{why} a multi-proposer DAG with a window-filtered, CCA-based local fork-choice attains fast, robust convergence under synchrony. 
 The fork-choice rule is set-valued: it maintains a preferred frontier \(F_t\) of conflict-free tips that are maximal under \(\succ\), rather than 
the tip of a single linear chain as in longest-chain protocols; a canonical tip is only chosen at the analysis layer (Section~\ref{sec:finality}).

We first list the assumptions, then present the block creation/reception and fork-choice
algorithms, followed by an informal walkthrough example. We end with a proof
roadmap highlighting where DAG Growth (DG), DAG Quality (DQ), DAG Common Past
(\DCP), and Tip Boundedness (TB) enter the persistence and liveness arguments.

Having introduced the DAG properties (DG, DQ, \DCP, TB) and their connection to ledger properties in Section~\ref{sec:prelim}, we now turn to protocol design. The goal is to present an idealized multi-proposer DAG protocol whose behavior can be analyzed directly against these properties. By abstracting away practical constraints such as network delay, bounded references, and transaction conflicts, the idealized model isolates the core consensus dynamics.

\subsection{Overview}
\label{sec:ideal-overview}

\subsubsection{Assumptions}
\label{sec:ideal-assumptions}

\ProjIdeal isolates the core consensus logic of a multi-proposer DAG protocol to enable crisp arguments about safety and liveness. 
The model makes the following simplifying assumptions:
\begin{itemize}
 \item \textbf{Synchronous delivery:} time proceeds in slots; all honest messages broadcast in slot $t$ are delivered to all honest parties by the end of slot $t$ (zero-latency abstraction).
 Under the synchronous network assumption, all honest validators share the same DAG view at the end of each slot; thus, for the purposes of analysis we may identify this common view with a single global DAG $G$.
 This identification is purely analytic; the protocol itself never relies on a globally agreed DAG.

 \item \textbf{Public-coin eligibility:} each validator $v$ is eligible in slot $t$ with probability $p$ \emph{independently of stake}. (Section~\ref{sec:base} realizes stake-proportional rates via VRF sortition.) 
 \item \textbf{Unbounded referencing:} when eligible, an honest block may reference \emph{any} existing block(s). Honest creators attempt to cite \emph{all} currently visible tips (equivalently, a maximum antichain).

\item \textbf{Bounded fork-choice window.}
Fork choice uses only short references from the last $w$ slots, as in
Section~\ref{sec:dag-building-blocks} (windowed references and \subdagweights): 
only these short references contribute to fork-choice scores and
block-confirmation depth. 
Long references may still be used to maintain
connectivity and to 
provide confirmation for older blocks, 
but they never affect \subdagweights or the preferred frontier. This bounded-weight discipline
prevents long-range withholding from dominating fork choice
(Section~\ref{sec:ideal-unbounded-vs-windowed}).

\item \textbf{No transaction conflicts:} To focus on consensus dynamics, each validator produces at most one block per slot by definition. In \ProjIdeal we ignore transaction dependencies/conflicts and treat blocks as abstract votes.

\end{itemize}

\paragraph{Intuition and example.}
Unbounded referencing ensures that any newly created honest block can \emph{pull in}
missing tips and rapidly merge honest work; the window-bounded scoring ensures that
ancient, withheld subgraphs cannot accumulate decisive weight. Concretely, suppose
slot $t$ produces two honest blocks $A$ and $B$ on different tips; in slot $t+1$,
an honest block $C$ may reference both, immediately collapsing the fork. Because
only recent references contribute to weight, a long-hidden adversarial 
\subdag cannot accumulate decisive weight once revealed: 
its old edges are disregarded by the fork-choice.

\subsubsection{Algorithms (for reference)}
The concrete algorithms that constitute the consensus protocol are presented in Section~\ref{sec:ideal-details}. Here we give a brief overview:

\begin{itemize}
 \item \textbf{Alg.~\ref{alg:ideal-antichain}} (\emph{Antichain Selection}): compute a maximum antichain of the current DAG. 
 \item \textbf{Alg.~\ref{alg:ideal-create}} (\emph{Block Creation}): if eligible in slot $t$, build a block that references the selected antichain and broadcast it.
 \item \textbf{Alg.~\ref{alg:ideal-reception}} (\emph{Block Reception}): integrate a valid incoming block; update $\Tips$ and fork-choice state.

\item \textbf{Alg.~\ref{alg:ideal-local-fc}} (\emph{Local Fork Choice / Preferred Frontier }):   
Starting from the short-reference tip set $U_t = \Tips_w(G_t)$, Algorithm~\ref{alg:ideal-local-fc} iteratively applies the conflicted-tips resolution rule $\CTR(\cdot,\cdot)$ to incomparable conflicting pairs, marking losers defeated and returning the preferred frontier $F_t$ (Def.~\ref{def:preferred-frontier}). The underlying notions --- windowed references, CCA, and \subdagweights --- are exactly those of Section~\ref{sec:dag-building-blocks}.

\end{itemize}

\paragraph{Exact antichain selection for honest nodes.}
In \ProjIdeal, honest validators compute the \emph{exact maximum
antichain} over the currently visible windowed DAG.  
Since the window contains only $O(\lambda w)$ blocks, the computation is
polynomial-time and can be performed using the Dilworth-based matching
algorithm described in Appendix~\ref{sec:appendix-dilworth}.  
This ensures that honest parties always choose the antichain of maximum
cardinality.

\subsubsection{Protocol Walkthrough (Honest Player View)}
\label{sec:ideal-walkthrough}

This subsection explains how an \emph{honest} validator executes \ProjIdeal in each slot, and how fork choice and conflict resolution interplay with block creation and reception.

\paragraph{Per-slot loop (local state machine). }
At the beginning of slot $t$, an honest player $P$ has local DAG $G_t=(V_t,E_t)$:
\begin{enumerate}
 \item \textbf{Eligibility check. } $P$ samples public coins for slot $t$ and determines eligibility with probability $p$.
 \item \textbf{Reference selection (Alg.~\ref{alg:ideal-antichain}). } $P$ deterministically computes an antichain $A_t\subseteq V_t$.
 \item \textbf{Block creation (Alg.~\ref{alg:ideal-block-creation}).}
If eligible, $P$ constructs $$b_t \gets \langle \id,\, \val = P,\, \slot = t,\, \payload,\, \refs = A_t \rangle$$
and broadcasts $b_t$, where $\payload$ is the single application payload sampled in Algorithm~\ref{alg:ideal-block-creation}.
 
 \item \textbf{Block reception and integration (Alg.~\ref{alg:ideal-reception}). }
 As blocks of slot $t$ arrive, $P$ verifies and inserts each block, updates $\Tips_{w}(G)$ and window summaries, and triggers \ForkChoiceUpdate.

\item 
 \textbf{Local fork choice (Alg.~\ref{alg:ideal-local-fc}). } 
 
\begin{enumerate}
 \item
 Scoring on the windowed DAG. Conceptually, $P$ uses the window-filtered \subdagweights
of Section~\ref{sec:basics}: when two tips $i, j$ need to be compared, $P$ evaluates their CCA-local scores
$\BranchW(i; c, w)$ and $\BranchW(j; c, w)$ at $c = \CCA(i, j)$. These weights are never
maintained as a global table; they are evaluated on demand inside the conflicted-tips resolution
rule $\CTR(\cdot,\cdot)$ (Def.~\ref{def:ctr}).

\item
Conflict resolution via $\CTR$ (internal to Alg.~\ref{alg:ideal-local-fc}). 
 Conceptually, the local fork choice maintains the preferred frontier $\Undom_t$ as in Def.~\ref{def:preferred-frontier} by iteratively applying the conflicted-tips resolution rule $\CTR(\cdot,\cdot)$ to conflicting tip pairs. When we compare two conflicting tips $\{i, j\}$, we simply apply $\CTR(i, j)$, which is exactly the same rule used internally by Algorithm~\ref{alg:ideal-local-fc}. 
 
 \end{enumerate}

\end{enumerate}

\subsubsection*{Additional Definitions Used in this Section}

\begin{definition}[Antichain; Maximal vs Maximum]
$A\subseteq V$ is an antichain if all elements are pairwise incomparable in $G$. Maximal means inclusion-maximal; maximum means largest cardinality.
\end{definition}

\begin{definition}[\MaximumAntichain\ on $G_{t,w}$]
This procedure returns the exact maximum-cardinality antichain of the
windowed DAG $G_{t,w}$, computed via the Dilworth-based matching
algorithm described in Appendix~\ref{sec:appendix-dilworth}.  It is used by
all honest validators in \ProjIdeal.
\end{definition}

\begin{definition}[\SamplePayload]
In the simplified model we treat each block as carrying a single opaque application payload (e.g., a transaction). We assume that the block creator runs some application-layer procedure \SamplePayload that returns either a valid payload or $\bot$. We do not model this procedure in detail.
\end{definition}

\begin{definition}[\ForkChoiceUpdate]
Maintain a cache comprising $\Tips_{w}(G)$, a per-validator latest-block map $L:\Val\to V\cup\{\bot\}$, and window reachability summaries sufficient to evaluate antichain tests and window-weight sums.
\end{definition}

\begin{definition}[\Broadcast]\label{def:broadcast}
We assume authenticated broadcast: any message sent by an honest party is delivered to every honest party by the end of slot $t$.
\end{definition}

\subsection{Protocol Details (Ideal)}
\label{sec:ideal-details}
\subsubsection{Idealized Antichain Selection and Block Creation}

Honest nodes compute the \emph{exact maximum antichain}.

\begin{algorithm}[H]
\caption{\AntichainSelection (Idealized)}
\label{alg:ideal-antichain}
\begin{algorithmic}[1]
\Require Current DAG view $G=(V,E)$ of the validator at the end of slot $t$.
\Ensure Antichain $A\subseteq V$
\State $A \gets \ExactMaxAntichain(G)$ 
\State \Return $A$
\end{algorithmic}
\end{algorithm}

\begin{algorithm}[H]
\caption{Block Creation (Idealized, per validator $v$ at slot $t$)}
\label{alg:ideal-create}\label{alg:ideal-creation}\label{alg:ideal-block-creation}
\begin{algorithmic}[1]
\Require Current DAG view $G=(V,E)$ of the validator at the end of slot $t$.
\If{$\Eligibility(v,t)=1$}
 \State 
 $\payload \leftarrow \SamplePayload()$
 \State $A \gets \AntichainSelection(G)$
 \State 
 $b \gets \langle \id,\,\val=v,\,\slot=t,\, \payload,\,\refs=A\rangle$
 \State Insert $b$ and edges $(r\to b)$ for $r\in A$ into $G$
 \State Update $\Tips_{w}(G)$ and reachability summaries
 \State \ForkChoiceUpdate($G$)
 \State \Broadcast($b$)
\EndIf
\end{algorithmic}
\end{algorithm}

\subsubsection{Idealized Block Reception}

Upon receiving a block $b \notin V$, the node first checks that all referenced parent
blocks are already present in $V$ and that inserting $b$ would not create a cycle.
If either check fails, $b$ is rejected. Otherwise, the node 
 inserts the edges
$(r \to b)$ for all $r \in \refs(b)$, updates tips and reachability summaries, and
then updates the fork-choice state.

\begin{algorithm}[H]
\caption{Block Reception (Idealized)}
\label{alg:ideal-reception}
\begin{algorithmic}[1]
    \Require Incoming block $b$, reference-DAG $G = (V, E)$
    \Ensure Either $b$ is accepted into $G$ or rejected
    \If{$b \in V$}
       \Return
    \EndIf
    \If{$\exists r \in \refs(b) \text{ such that } r \notin V$}
      \State 
      \textbf{reject} $b$ \Return \Comment{invalid: references unknown parent blocks}
    \EndIf
    \If{$\exists r \in \refs(b) \text{ such that } r \in \Desc^*(b)$}
      \textbf{reject} $b$ \Return \Comment{invalid: would create a cycle}
    \EndIf
    \State Insert $b$ into $G$: add vertex $b$ and edges $(r \to b)$ for all $r \in \refs(b)$
    \State Update $\Tips_w(G)$ and reachability summaries
    \State \ForkChoiceUpdate$(G)$ 
    \State \Broadcast$(b)$
  \end{algorithmic}
\end{algorithm}

\subsubsection{Local Fork Choice (\ProjIdeal)}
\label{subsec:ideal-local-fc}

\paragraph{Local fork choice in \ProjIdeal has two layers.}  
(i) A pairwise conflicted-tips resolution rule $\CTR(\cdot,\cdot)$ that, given two incomparable, conflicting tips, selects the better \branch using CCA-local, window-filtered weights.
(ii) A set-valued procedure that applies $\CTR$ to all conflicting tips in the short-reference tip set $U_t = \Tips_w(G_t)$ and outputs a conflict-free preferred frontier $\Undom_t \subseteq U_t$. We formalize these two components in Defs.~\ref{def:cfpf} and~\ref{def:ctr} and implement the frontier in Algorithm~\ref{alg:ideal-local-fc}.

In the idealized model, we realize the preferred frontier $F_t$ (Def.~\ref{def:preferred-frontier}) as follows.

\begin{definition}[Conflict-free preferred frontier]\label{def:ideal-preferred-frontier}\label{def:cfpf}\label{def:pf}
In \ProjIdeal we instantiate the preferred frontier $F_t$ from Definition~\ref{def:preferred-frontier} as follows. Given the windowed DAG $G_t$ and the short-reference tip set $U_t = \Tips_w(G_t)$, we repeatedly apply the conflicted-tips resolution rule $\CTR(\cdot,\cdot)$ (Definition~\ref{def:ctr}) to pairs of conflicting tips in $U_t$ until no conflicts remain. The surviving tips form the conflict-free preferred frontier $F_t \subseteq U_t$.
\end{definition}

Algorithm~\ref{alg:ideal-local-fc} realizes the preferred frontier of Def.~\ref{def:cfpf}. Starting from the short-reference tip set $U_t = \Tips_w(G_t)$, it repeatedly applies $\CTR(\cdot,\cdot)$ to every pair of incomparable, conflicting tips and marks the loser as defeated. 
The surviving tips are exactly the $\succ$-maximal ones under the pairwise relation induced by $\CTR$, hence form the conflict-free preferred frontier $\Undom_t$. 
This procedure is purely local, depending only on the local DAG view $G_t$ and the window parameter $w$.

\begin{algorithm}[H]
  \caption{Local Fork Choice / Preferred Frontier (Idealized)}
  \label{alg:preferred-frontier}\label{alg:ideal-local-fc}
  \begin{algorithmic}[1]
    \Require DAG view $G = (V, E)$, window $w$
    \Ensure Conflict-free preferred frontier $\Undom_t$
    \State $U \gets \Tips_w(G)$
    \ForAll{unordered pairs $\{i, j\} \subseteq U$ such that
      $\payload(i)$, $\payload(j)$ conflict}
      \State $\mathsf{winner} \gets \CTR(i, j)$ \Comment{Def.~\ref{def:ctr}}
      \State $\mathsf{loser} \gets \{i, j\} \setminus \{\mathsf{winner}\}$
      \State mark $\mathsf{loser}$ defeated
    \EndFor
    \State \Return $\{ x \in U : x \text{ is not defeated} \}$
  \end{algorithmic}
\end{algorithm}

\begin{remark}[Defeated tips]\label{rem:defeated-tips}
In Algorithm~\ref{alg:ideal-local-fc} (and also Algorithm~\ref{alg:base-local-fc} in Section~\ref{sec:base}),
when we \emph{mark a tip as defeated} we only mean that this tip is
eliminated from the current conflict class and will not appear in the
preferred frontier $\Undom_t$ at that time step. The underlying \branch and
all of its blocks remain part of the DAG~$G_t$: they can still be
referenced by future blocks and may still contribute short-reference
weight to other \branches. In particular, ``defeated'' does \emph{not}
mean that the \branch is pruned, invalidated, or removed from honest
views.
\end{remark}

\paragraph{Rationale. } 
Using $y(\cdot)$ prevents adversarial control via transaction ordering or content hashing; it is a stand-in for the VRF output in the practical protocol.

\paragraph{Tie-breaking and the role of $y(\cdot)$.}
In the ideal model, each block carries an independent public-coin label $y(\cdot)$ used solely for deterministic tie-breaking inside $\CTR(\cdot,\cdot)$: if two conflicting tips have equal CCA-local \subdagweight, the one with smaller $y(\cdot)$ wins. 
In \ProjBase (\S\ref{sec:base}), the same public coin is instantiated via per-block VRF outputs.

\paragraph{Finalization rule (pointer). }
Nodes finalize when a block is $k$-deep under the preferred frontier; $k$ and $\varepsilon$
are calibrated in \S\ref{sec:finality} (Def.~\ref{def:finalization-rule}, Thm.~\ref{thm:master-finality}).

\begin{remark}[Comparison with BFT-style finality]
Traditional BFT-based protocols (e.g., HotStuff, Casper FFG) achieve finality
through explicit voting and quorum certificates.
In contrast, our design provides \emph{probabilistic finality} directly from
the weight-based fork-choice and bounded-window analysis:
a block becomes irreversible once its depth ensures an honest margin
that no adversary with $\alpha < 1/2$ stake can overcome.
Thus, finality emerges organically from the ledger's growth properties,
without any auxiliary consensus layer.
\end{remark}

\subsection{On Unbounded Referencing vs.\ Bounded Fork Choice}
\label{sec:ideal-unbounded-vs-windowed}
\label{sec:ideal-window-justification}
\ProjIdeal allows \emph{unbounded referencing}: an honest block may cite \emph{all} existing tips (or a maximum antichain), which accelerates convergence and directly implies Tip Boundedness (TB) in the synchronous ideal.
However, using \emph{unbounded} weight in fork choice (i.e. counting all references since genesis) is unsafe even in a synchronous model because an adversary can \emph{withhold} a private \subdag for many slots and then reveal it at once. 
The revealed descendants then contribute a large cumulative score to an \adversarialsubdag, potentially overriding the honest \subdag despite minority stake.

\begin{lemma}[Unsafe infinite-horizon \branch weight]
Even under synchrony and honest majority, if fork choice scores competing \branches (rooted at conflicting tips) 
by \emph{unbounded} descendant counts since genesis, there exists a withholding strategy that can cause arbitrarily deep reversions with non-negligible probability.
\end{lemma}

\begin{proof}[Proof sketch]
Let the adversary with stake fraction $\alpha<1/2$ privately build a 
\subdag rooted at 
a recent CCA while the honest parties continue extending the public \subdag.
 After $L$ slots, the private \subdag accumulates $\Theta(\alpha L)$ references while the honest \subdag accumulates $\Theta((1-\alpha)L)$. If the scoring aggregates from genesis, the adversary can time its release to coincide with honest fragmentation at the frontier (many tips), so the \adversarialsubdag appears heavier due to concentrated weight. Since contributions are not window-limited, the \emph{entire} private accumulation applies, enabling deep reorgs. 
 Window-filtering neutralizes this leverage by discarding tips outside the last $w$ slots.

\end{proof}

\noindent\textbf{Conclusion. } In \ProjIdeal we \emph{retain unbounded referencing} (for rapid merging and TB), but we \emph{require a bounded window $w$ in the fork-choice metric}. This pairing preserves fast convergence while preventing long-range weight accumulation from withheld blocks.

\subsection{Security of the Idealized Protocol}
\label{subsec:ideal-security}
\label{sec:ideal-security}

\paragraph{Proof roadmap.}
The security proof follows the backbone style: 
(i) DG and DQ guarantee a steady supply of honest work into the preferred frontier; 
(ii) TB prevents indefinite fragmentation of honest tips; 
(iii) \DCP bounds view divergence on the common past $\Past(t) = \AncStar(\Undom_t)$  
after trimming the last $k_D$ layers.
Under these invariants, the windowed score gap between the
honest-\prefsubdag and any competing \branch has positive drift with bounded
increments. Azuma--Hoeffding then yields an exponential tail bound, implying that
a depth of $k=\Theta(w)+O(\log(1/\varepsilon))$ suffices for $(k,\varepsilon)$-finality. 
(See Section~\ref{sec:finality} for the master $(k,\varepsilon)$-finality theorem and parameter calibration.)

We now prove that \ProjIdeal satisfies the DAG properties (DG, DQ, \DCP, TB) and, consequently, \emph{ledger persistence (safety)} and \emph{liveness}. 
Throughout this section, let $H \in (1/2,1]$ denote the honest fraction of validators (or equivalently, of eligibility identities), and let $p \in (0,1]$ be the per-slot eligibility rate (public coin). In Section~\ref{sec:base} we show how stake-proportional VRF sortition realizes the same process with $H$ given by the honest stake fraction.
All statements hold \emph{except with negligible probability} in the security parameter $\secp$  
(via standard Chernoff bounds on sums of independent Bernoulli trials). 

Throughout this section, ``the fork choice'' always refers to the set-valued procedure that maintains the preferred frontier $\Undom_t$ (Defs.~\ref{def:preferred-frontier},~\ref{def:cfpf}), while any comparison between two conflicting tips is made via the underlying \CTR rule (Def.~\ref{def:ctr}).

\subsubsection{Consensus Properties}
\label{sec:consensus-properties}
\paragraph{Preliminaries.} 
Under synchrony and unbounded referencing, all blocks created in slot $t$ can reference \emph{all tips visible at the start of the slot}, and all honest parties receive the same set of blocks by the end of the slot.

\begin{theorem}[DAG Growth (DG) in \ProjIdeal]
\label{thm:ideal-DG}
Fix any interval of $\ell$ consecutive slots. With probability $1-\negl(\secp)$, at least $\tau_D\ell$ honest blocks are created and become ancestors of some honest tip by the end of the interval, for $\tau_D=\Theta(pH)$.
\end{theorem}

\begin{proof}[Proof sketch]
Chernoff concentration on the number of honest-eligible blocks in the interval gives $\Theta(pH\ell)$ honest blocks in the interval.
Under synchrony, each of these blocks is referenced by subsequent honest blocks in the next slot, thereby becomes an ancestor of an honest tip. 
\end{proof}

\begin{theorem}[DAG Quality (DQ) in \ProjIdeal]
\label{thm:ideal-DQ}
For any $\ell$-slot window, the fraction of honest blocks entering the preferred frontier is at least $\mu_D=\Theta(H)$ with probability $1-\negl(\secp)$.
\end{theorem}
\begin{proof}[Proof sketch]
Eligibility is driven by an i.i.d.~public coin across validators. Let $H$ denote the fraction of validators (equivalently, eligibility identities) that are honest. Then each eligible block is honest with probability $H$, so the expected honest fraction per slot is $H$. 
Since fork-choice is deterministic and purely local (CCA-based), the selection process does not bias adversarial blocks upwards; Chernoff bounds imply concentration around $H$ over $\ell$ slots.

\end{proof}

\begin{lemma}[TB in \ProjIdeal]\label{lem:TB-ideal}\label{lem:ideal-TB}
In the ideal model with synchrony and unbounded referencing, if each eligible honest block references \emph{all} currently visible tips (max-antichain policy), then for every slot $t$,
\[\Pr\!\left[\,|\Tips_t| \;\le\; c\cdot \lambda\,\right]\;\ge\;1-\negl(\secp),\]
for a universal constant $c$ and expected per-slot eligibility $\lambda$ (honest + adversarial).
In particular, $\TB[\beta]$ holds with $\beta=O(\lambda)$.
\end{lemma}
\begin{proof}[Proof sketch]
At the end of slot $t$, all honest parties have identical views (synchrony). Any tips from slot $t-1$ become covered in slot $t$ because every honest eligible block references all visible tips. Hence the only tips at time $t$ are births of slot $t$. The total number of births in slot $t$ is at most the number of eligible validators in that slot, whose expectation is $\lambda$ (honest + adversarial); a Chernoff bound shows that this quantity is $O(\lambda)$ with high probability, yielding the claimed tip bound.
\end{proof}

\begin{theorem}[DAG Common Past (\DCP) in \ProjIdeal]
\label{thm:ideal-DCPprime}
There exists $k_D=O(1)$ such that for any two honest parties $P, Q$ and times $t_1 \le t_2$, after removing
the most recent $k_D$ layers from $\Past_P(t_1)$ and $\Past_Q(t_2)$, the earlier trimmed past is contained in
the later one, except with negligible probability; i.e.,
$\Trim(\Past_P(t_1), k_D) \subseteq \Trim(\Past_Q(t_2), k_D)$.
\end{theorem}

\begin{proof}[Proof sketch]
Under the broadcast assumption of Definition~\ref{def:broadcast} (i.e., the $\Delta=0$ special case of our bounded-delay model), all honest parties have identical DAG views at the end of each slot. Combined with TB, this implies that the frontier $\Undom_t$ is the same for all honest parties (up to the last $O(1)$ layers). 
Equivalently, after trimming the most recent $k_D = O(1)$ layers, all honest views coincide on the remaining prefix, which yields \DCP.
\end{proof}

\subsubsection{Ledger Properties (Persistence and Liveness)}
\label{subsec:ideal-ledger-props}

We now show that \ProjIdeal's  CCA-local, window-filtered fork choice, coupled with CCA-based conflict resolution, yields ledger safety and liveness.

\begin{theorem}[Persistence (Safety) of \ProjIdeal]\label{thm:ideal-persistence}
Let $k,\varepsilon$ be any parameters that satisfy the master finality guarantee of \S\ref{sec:finality} (i.e., Thm.~\ref{thm:master-finality}, instantiated by Cor.~\ref{cor:ideal-finality}). If a block $B$ is $k$-deep under the preferred frontier of an honest party by the end of some slot, then with probability at least $1-\varepsilon$ the block $B$ remains in the ledger of every honest party forever after.
\end{theorem}

\begin{proof}
This is an immediate consequence of the $(k,\varepsilon)$-finality rule in \S\ref{sec:finality}. Under the ideal-model assumptions, Cor.~\ref{cor:ideal-finality} instantiates Thm.~\ref{thm:master-finality}; thus, once $B$ is $k$-final (Def.~\ref{def:finalization-rule}), the probability that $B$ is ever removed is at most $\varepsilon$. Persistence follows.
\end{proof}

\begin{theorem}[Liveness of \ProjIdeal]\label{thm:liveness-areon-ideal}
\label{thm:liveness-ideal}
Fix a slot~$t$ and let $b$ be an honestly created valid block in slot~$t$. Under the standing ideal-model assumptions of Sections~\ref{sec:ideal-assumptions} and~\ref{sec:ideal-security} (synchronous delivery, public-coin eligibility with honest fraction $H > 1/2$, unbounded referencing, and a bounded fork-choice window), with probability at least $1 - \negl(\secp)$ there exists
$\ell_{\mathsf{live}} = \tilde{O}\!\left(\frac{1}{pH}\right)$
such that by slot $t + \ell_{\mathsf{live}}$ the block $b$ is $k$-deep under the preferred frontier of every honest party.
\end{theorem}

\begin{proof}[Proof sketch]
Let $\cP$ denote the validator set and let $\lambda_h$ be the expected number of
honest-eligible validators per slot (Table~\ref{tab:notation}), so in our public-coin model
$\lambda_h = \Theta(pH \cdot |\cP|)$. In each slot, the probability that at least one honest block is
created is therefore a constant $\Omega(pH)$.  
Under synchrony and unbounded referencing, once $b$ is visible it is short-referenced by subsequent honest blocks within $O(1/(pH))$
slots, and the DAG Growth (DG) property implies that the number of honest confirming
descendants of $b$ then grows linearly in time. DAG Quality (DQ) and Tip-Boundedness (TB)
ensure that this growth is not diluted across an unbounded frontier: a constant fraction of new
weight continues to land on preferred honest subDAGs. Combining these properties with the
bounded fork-choice window yields a standard martingale argument (as in Section~\ref{sec:finality}) showing
that, except with negligible probability, $b$ accumulates at least $k$ short-ref-confirming
descendants in every honest view within $\tilde{O}(1/(pH))$ slots. Equivalently, there exists
$\ell_{\mathrm{live}} = \tilde{O}(1/(pH))$ such that by slot $t + \ell_{\mathrm{live}}$ the block $b$ is $k$-deep
under the preferred frontier of every honest party.
\end{proof}

\begin{remark}[Block liveness as a consequence]\label{rem:tx-liveness-from-block-liveness}
In the simplified presentation, each block carries a single application payload (e.g., a transaction). Under the standard re-broadcast assumption on honest parties, block liveness immediately implies the usual block-liveness guarantee: every valid payload eventually appears in a $k$-deep block under the preferred frontier of all honest parties. A full treatment for multi-transaction blocks and an explicit mempool-sampling policy is left to a transaction-layer extension of this work.
\end{remark}

%% file: sec4.tex

\section{\ProjBase: A Practical DAG-Based PoS Protocol}
\label{sec:base}

\paragraph{High-level overview.}
\ProjBase instantiates the ideal design under the standard bounded-delay model (unknown $\Delta$). It retains the window-filtered  
fork-choice, and introduces VRF sortition for eligibility, a strict short-reference window of size $w\!\ge\!\Delta$, and a weightless long reference to preserve connectivity. 
We explain the module boundaries (eligibility, block creation, reception, fork-choice, conflict resolution), give an execution walkthrough, and state the precise assumptions used by the proofs.

\subsection{Overview} 

\subsubsection{Assumptions and Instantiation of \ProjIdeal}
\label{subsec:base-overview}
\ProjBase is a practical instantiation of \ProjIdeal (Section~\ref{sec:ideal}) under \emph{the bounded-delay model} with \emph{VRF sortition}, \emph{bounded short references} (window $w$), and a single long reference for connectivity. 
Blocks are required to satisfy an abstract payload-validity predicate; we do not model transaction or UTXO semantics explicitly. 
We retain the same CCA-local, window-filtered fork-choice and conflict rule as in \ProjIdeal,
instantiated with VRF-based eligibility and bounded referencing.

\begin{itemize}
 \item \textbf{Network.} Honest messages are delivered within an unknown delay bound $\Delta$ at all times.
 \item \textbf{Eligibility.} Each validator $v$ is eligible in slot $t$ iff $(y,\pi) = \VRFSign(\sk^{\mathsf{vrf}}_v;\, \Proj \| \Encode(t))$\footnote{If the implementation uses an explicit parallelism index $\psi\in\{1,\dots,\Psi\}$, include it in the VRF domain separator: $\VRFSign_{\sk_v}(\Proj\,\|\,\psi\,\|\,\Encode(t))$.}  satisfies $y<p(t,v)$ (Alg.~\ref{alg:block-creation} and Section~\ref{sec:notation-base}
 for the concrete VRF-based eligibility mechanism).

 \item \textbf{References.} In slot $t$, short references must lie in the \emph{window} $W(t;w)=[t-w,t-1]$; 
 in addition, if there exists any visible block older than $t-w$, the block is required to include a single long reference~$\ell$ to such an older block, and this long reference carries zero fork-choice weight.

 \item \textbf{Fork-choice.}  We use the same CCA-local, window-filtered fork choice as in \ProjIdeal:
  on the local windowed DAG $G_{t,w}$, short references inside the last $w$ slots carry
  unit weight, long references carry zero weight, and 
  Algorithm~\ref{alg:base-local-fc}  (\textsc{Local Fork Choice}) applies the conflicted-tips resolution rule $\CTR(\cdot,\cdot)$ to all incomparable conflicting tips in $\Tips_w(G_t)$ and returns the conflict-free preferred frontier $\Undom_t$.
 
 \item  \textbf{Conflicts.}
When two incomparable blocks imply conflicting ledger states, the node resolves them
via the same CCA-local conflicted-tips resolution rule $\CTR(\cdot,\cdot)$
used inside \textsc{Local Fork Choice} (Defs.~\ref{def:cfpf},~\ref{def:ctr}).
Concretely, given tips $i,j$, let $c = \CCA(i,j)$ and compare the window-filtered
\subdagweights $\BranchW(i;c,w)$ and $\BranchW(j;c,w)$; if one is larger, $\CTR$ returns
that tip, and if they tie $\CTR$ selects the tip with smaller VRF label $y(\cdot)$.
This is exactly the pairwise restriction of the fork-choice relation.
\end{itemize}

\paragraph{Why a short window and a weightless long reference?}
The short window blocks long-range withholding attacks: edges older than $w$ slots
do not contribute to the fork-choice score. The optional long reference reconnects
stragglers beyond the window \emph{without} affecting scores, ensuring the DAG
remains connected while keeping the security analysis local to the window.

\paragraph{Parameter discipline.} Throughout the analysis we assume
\begin{equation}
\label{eq:base-params}
w \;\ge\; \Delta \;+\; \omega(\log\secp) \qquad\text{and}\qquad 
H>\tfrac{1}{2}+\varepsilon
\end{equation}
for some constant $\varepsilon>0$ (consistent with {\em Standing Assumptions} in Sec.~\ref{sec:prelim}). 
Here $\Delta$ denotes an upper bound on message delay in the network,  
and $H$ is the fraction of total stake held by honest parties. The slack $\omega(\log\secp)$ is used for concentration bounds.

\paragraph{Honest short-reference coverage (\(q>0\)).}
Conditioned on being eligible in slot \(t\), an honest validator places at least one
\emph{short-reference} to a currently visible tip with probability at least \(q>0\).
Let $\mathcal{E}_{v,t}$ be the event that, when eligible in slot $t$, validator $v$ includes at least one short reference to a visible tip.
We assume \(\Pr[\mathcal{E}_{v,t}\mid v \text{ eligible in } t] \ge q.\)

The choice of tip(s) may be deterministic or randomized; the assumption only enforces that, 
whenever some visible tip lies inside the short-reference window, at least one such tip is short-referenced upon eligibility (and long references are used only when no visible tip lies in the window). 
This condition underlies the \emph{Tip-Boundedness} bound in the practical model (Lemma~\ref{lem:base-TB})  and the resulting liveness bounds; if honest blocks short-reference all visible tips, then $q = 1$.

As in Remark~\ref{rem:TX}, throughout this section we adopt the simplifying assumption that
each block carries a single opaque application payload (e.g., a transaction), and we treat
validity and conflicts at the block/payload level.  Extensions supporting multiple
transactions per block and concrete mempool-sampling policies are left to future work.

\subsubsection{Protocol Explanation (Honest Player View)}
\label{subsec:base-walkthrough}
We describe the per-slot state machine of an \emph{honest} node $P$; we explicitly indicate where Algorithms~\ref{alg:block-creation}–\ref{alg:base-local-fc} are invoked.

\paragraph{Start of slot $t$.} Node $P$ maintains a local DAG $G_t=(V_t,E_t)$ and a ledger state view induced by its \emph{preferred frontier}. It also caches reachability summaries and a per-validator latest-block map.

\begin{enumerate}
\item \textbf{Eligibility \& payload selection.}
$P$ computes $(\mathsf{ok},\pi,y) \gets \Eligibility(P,t)$.
(In \ProjBase, the same per-block VRF output $y$ also serves as the canonical public coin
for tie-breaking when needed.)
If $\mathsf{ok} = \true$, $P$ selects a candidate application payload $\payload_t$ from the
mempool (details of this sampling procedure are abstracted away in our simplified model).

\item \textbf{Reference selection (short \& long). } $P$ computes a \emph{deterministic} large antichain $R\subseteq V_{t,w}$ of recent blocks (e.g., \MaximumAntichain). 
  
\item\textbf{Block creation and gossip (Alg.~\ref{alg:base-block-creation}).}
If eligible and after validity checks, $P$ assembles
\[
  b \gets \langle \id,\, \val = P,\, \slot = t,\,
         \payload = \payload_t,\,
         \refs = R \cup \{\ell\},\, y,\, \pi,\, \sigma \rangle,
\]
signs it, inserts it into its local view, and gossips $b$ to the network.

\item \textbf{Block reception (Alg.~\ref{alg:block-reception}, continuous during slot).}
For each received block $b'$: verify its VRF proof and signature; fetch any missing referenced blocks (parents); check that
(i) all short references satisfy the window constraint, 
(ii) there is at most one long reference,
(iii) no cycles are created, 
(iv) the short-reference set of $b'$ within the window forms an
antichain in $G_{t,w}$, and 
(v) no payload conflicts with ancestors (according to the abstract conflict predicate), i.e., $\neg\,\Conflicts\bigl(b', \Anc^*(\mathrm{refs}(b'))\bigr)$.
Reject $b'$ if any check fails; otherwise integrate it and update the cached summaries and $\Tips(G)$.

\item \textbf{Local fork-choice (Alg.~\ref{alg:base-local-fc}, end of slot). } 
 At the protocol layer, $P$ runs \textsc{Local Fork Choice} (Alg.~\ref{alg:base-local-fc}) on $\Tips_w(G)$
to obtain the  
preferred frontier 
$\Undom_t$. 

This frontier is the only fork-choice state used by the protocol.

(Analysis view; not part of the protocol). For the purposes of the finality analysis in Section~\ref{sec:finality}, we associate to each local view $G_t$ and preferred frontier $\Undom_t$ an abstract ledger obtained by linearizing $\AncStar(\Undom_t) \cap \DescStar(a)$ for some analysis anchor a (e.g., genesis or the last finalized checkpoint), as in Section \ref{sec:tools}. This construction is purely analytical and does not affect protocol behavior or specify a client operation.

\item
\textbf{Conflict resolution via CCA (on demand).}    When two incomparable tips $i,j$ imply conflicting states, $P$ simply runs the 
   conflicted-tips resolution rule $\CTR(i,j)$ (Def.~\ref{def:ctr}), which compares their
   CCA-local window-filtered \subdagweights and uses the VRF label $y(\cdot)$ to break ties.
   This is consistent with the fork choice, because $\CTR$ is exactly the pairwise restriction
   of the relation that Algorithm~\ref{alg:base-local-fc} uses to construct the preferred frontier $\Undom_t$.

\end{enumerate}

\begin{remark}[Recall: \prefsubdags]
 Recall from Def. 2.4 that a \prefsubdag is any \branch containing a tip in the preferred frontier $F_t$; our fork choice therefore outputs a set $F_t$ of tips rather than a single chain tip.
\end{remark}

\subsection{Protocol Details}
\label{sec:protocol-base}
\subsubsection{Notation and Objects}
\label{sec:notation-base}
We reuse the notation from Section~\ref{sec:ideal}. In addition:

\paragraph{Ledger validity and conflicts.}
We reuse the abstract conflict predicate $\Conflicts(b,H)$ from
Definition~\ref{def:conflicts}, where informally $\Conflicts(b,H)$ holds if the payload
of block $b$ would violate ledger validity when applied on top of the state induced by
the ancestor closure of $H$ (for example, by double-spending in a UTXO ledger).
The concrete validity rules and the mempool-sampling process are left to the application
layer and are not modeled in this paper.

\paragraph{Reference window and references.}
Fix a slot $t$ and window size $w$. Let $W(t;w)=\{t-w,\ldots,t-1\}$ be the short-reference window and let $V_{t,w}$ denote the set of blocks that are visible at slot $t$ and were created in slots indexed by $W(t;w)$.

For any new block $b$ created at slot $t$, we use two primitives:
\begin{itemize}
  \item \textbf{Short-reference set.} $\refs(b)\subseteq V_{t,w}$ is the short-reference set of $b$. In implementation (Alg.~6), $\refs(b)$ is chosen as an exact maximum antichain of currently visible tips inside $V_{t,w}$.

\item \textbf{Long reference.}
The field $\longref(b) \in (V \setminus V_{t,w}) \cup \{\bot\}$ is a long reference used to
maintain connectivity to older parts of the DAG.
If $b$ has no admissible short-reference parent in the current window $V_{t,w}$, then the protocol \emph{requires} a long reference, i.e., $\longref(b) \neq \bot$ and it must point to some ancestor of $b$ outside $V_{t,w}$.
Otherwise (if $b$ has at least one admissible short-reference parent), we set $\longref(b) = \bot$.

Here, $\ell \gets \LongRef(G,t,w,R)$ denotes a deterministic long-reference selection
procedure that, given the current DAG $G$, slot $t$, window $w$, and short-reference set $R$, returns either $\bot$ or a single ancestor $\ell$ with $\slot(\ell) < t-w$, used only for connectivity 
and for extending the ancestor closure used in block confirmation
(long references carry zero fork-choice weight).

Consequently, long references never affect \subdagscores, tip status, or CCA comparisons in the fork-choice rule. They do, however, extend the ancestor closure used for 
block confirmation:
when evaluating $\Conflicts(\cdot,\cdot)$, we treat 
a block as already confirmed if it is reachable via a path that may mix short and long references.

\end{itemize}

\paragraph{Counting short references in the window.}
When evaluating at slot $t$, we write
\[
  \ShortRefs_t(d;w)\ :=\ \bigl|\{\,u\in V_{t,w}\ :\ (u\to d)\ \text{is a short-reference edge}\,\}\bigr|
\]
for the number of short references to $d$ coming from blocks inside the current window. When $t$ and $w$ are clear from context, we abbreviate $\ShortRefs(d;w)$.

All fork-choice computations (Algorithm~\ref{alg:base-local-fc}) and CCA-local conflict 
comparisons compute \subdagweights by aggregating only short-reference edges, i.e., they only count $\ShortRefs_t(\cdot; w)$. 
The CCA itself is computed from the full ancestry relation (which is defined using all references), but once the CCA is fixed, \subdagweights aggregate only short references.

\paragraph{VRF-based eligibility and block authentication.}
We instantiate eligibility and block authentication using a verifiable
random function (VRF) and a standard digital signature scheme.

Let $(\VRFGen,\VRFSign,\VRFVerify)$ be a secure verifiable random
function (VRF)~\cite{FOCS:MicRabVad99,PKC:DodYam05}, and let
$(\SIGGen,\SIGSign,\SIGVerify)$ be a secure digital signature scheme.
Each validator $v$ obtains a VRF key-pair
\[
  (\sk^{\mathsf{vrf}}_v, \pk^{\mathsf{vrf}}_v)
  \gets \VRFGen(1^\kappa)
\]
and a signing key-pair
\[
  (\sk^{\mathsf{sig}}_v, \pk^{\mathsf{sig}}_v)
  \gets \SIGGen(1^\kappa)
\]
that are fixed in the genesis state.

For slot $t$, let $p(t,v)$ denote the per-slot success probability for
validator $v$ (typically
$p(t,v) = \phi \cdot \stake(v)/\StakeTot$ for a global
parameter~$\phi$). We define the VRF-based eligibility procedure
$\Eligibility(v,t)$ used in Algorithm~\ref{alg:base-block-creation} as follows. Validator $v$
computes
\[
  (y,\pi) \gets
  \VRFSign\bigl(\sk^{\mathsf{vrf}}_v;\,
    \Proj \parallel \Encode(t)\bigr),
\]
interpreting $y$ as a real in $[0,1)$. If an explicit parallelism index
$\psi \in \{1,\ldots,\Psi\}$ is used, we instead domain-separate as
\[
  (y,\pi) \gets
  \VRFSign\bigl(\sk^{\mathsf{vrf}}_v;\,
    \Proj \parallel \psi \parallel \Encode(t)\bigr).
\]
We set
\[
  \mathsf{ok} \gets (y < p(t,v))
\]
and define
\[
  \Eligibility(v,t) \text{ to output } (\mathsf{ok},\pi,y).
\]

Whenever $\mathsf{ok} = \true$ and $v$ creates a block $b$ in slot $t$
(Algorithm~5), it sets
\[
  y(b) \gets y, \qquad \pi(b) \gets \pi,
\]
and deterministically derives a block identifier
\[
  \id(b) \gets
  \Hash\bigl(v \parallel t \parallel \refs \parallel \payload
              \parallel y \parallel \pi\bigr).
\]
The block is authenticated via a digital signature
\[
  \sigma(b) \gets\SIGSign\bigl(\sk^{\mathsf{sig}}_v;\, \id(b)\bigr),
\]
and $b$ stores the triple $(y(b),\pi(b),\sigma(b))$ in its header, as in
Section~\ref{sec:DAG-structure}. 

\medskip

On receiving a block $b$ attributed to validator $v$, a node first checks
eligibility and authenticity before performing any structural or window
checks. It verifies the VRF proof by computing
\[
  \mathsf{ok}' \gets\VRFVerify\bigl(\pk^{\mathsf{vrf}}_v;\,\Proj \parallel \Encode(\slot(b)), y(b), \pi(b)\bigr)
\]
and verifies the digital signature
\[
  \SIGVerify\bigl(\pk^{\mathsf{sig}}_v;\,
    \id(b), \sigma(b)\bigr).
\]
If either check fails or $\mathsf{ok}' = \false$ (equivalently,
$y(b) \ge p(\slot(b),v)$), then $b$ is rejected.

\medskip

\noindent\textbf{Conceptual block generation and verification.}
Conceptually, each block in \ProjBase is produced and authenticated in
two layers.

\emph{Eligibility and public randomness via VRFs.}
In each slot $t$, every validator $v$ runs the VRF-based eligibility
procedure above. The VRF output $y$ determines whether $v$ is eligible
in slot $t$ according to its stake weight, and the same output $y$ is
recorded in the block as $y(b)$ and later reused as a public coin for
tie-breaking in the fork-choice rule (Algorithm~\ref{alg:base-local-fc}).

\emph{Authentication via digital signatures.}
Given eligibility, $v$ selects references (a maximum antichain within
the short-reference window, plus an optional long reference) and a
payload, and deterministically derives a block identifier
$\id(b)$ from $(v,t,\refs,\payload,y,\pi)$. The validator then signs this
identifier with its long-term signing key to obtain $\sigma(b)$. The
signature binds the VRF output and the referenced parents to $v$'s
identity and prevents adversaries from forging or tampering with blocks.

On reception (Algorithm~\ref{alg:block-reception}), nodes verify both the VRF proof and the
digital signature before performing any structural or window checks.
This mirrors the design of Ouroboros Praos~\cite{EC:DGKR18}: the VRF
implements stake-proportional, unpredictable leader election and
provides per-slot randomness, while a standard signature scheme
authenticates the block and links it to the elected leader.

\subsubsection{Block Creation (\ProjBase)}
\label{sec:base-creation}

\paragraph{\Subdagweight (protocol-level).} \ProjBase reuses the CCA-local \subdagweight from
Section~\ref{sec:basics}: when comparing two competing tips $i, j$ with $c = \CCA(i, j)$ and window
$w$, the \subdagweight of $i$ is
\[
  \BranchW(i; c, w) := \sum_{d \in \Desc^*(i) \cap \Desc^*(c)} \wref(d; w), 
\]
and analogously for $j$. 
Here $\wref(d; w)$ denotes the number of short references to $d$ within the window~$w$, as defined in Section~\ref{sec:basics}.

\begin{algorithm}[H]
\caption{Block Creation (\ProjBase)}
\label{alg:block-creation}\label{alg:base-block-creation}
\label{alg:base-create}
\begin{algorithmic}[1]
\Require Ledger state $S$, validator $v$, key $\sk(v)$, stake $\stake(v)$, slot $t$, window $w$

\Ensure New block $b$ or $\bot$ if ineligible
\State $(\mathsf{ok}, \pi, y) \gets \Eligibility(v,t)$
\If{not $\mathsf{ok}$} \Return $\bot$ \EndIf
\State $\payload \gets \SamplePayload(S; \rho = \Hash(y\| t))$ \Comment{Abstract payload sampling; may return $\bot$}

\State $R \gets \MaximumAntichain(G_{t,w})$ \Comment{Deterministically select large antichain of recent blocks}

\State $\ell \gets \LongRef(G, t, w, R, \payload)$ 

\If{$\Conflicts(\payload,\AncStar(R\cup\{\ell\}\setminus\{\bot\}))$}
\State set $\payload \gets \bot$ \Comment{In the simplified model we simply drop a conflicting payload}
\EndIf

\State $\refs \gets 
\begin{cases}
R \cup \{\ell\} & \text{if }\ell\neq\bot,\\
R & \text{otherwise. }
\end{cases}$
\State $\id \gets \Hash(v \,\|\, t \,\|\, \refs \,\|\, \payload \,\|\, y \,\|\, \pi)$\Comment{Define block identifier deterministically}
\State $\sigma \gets\SIGSign\bigl(\sk^{\mathsf{sig}}_v;\, \id\bigr)$

\State $b \gets \langle \id, v, t, \payload, \refs, y, \pi, \sigma\rangle$

\State Insert $b$ and edges $(r \to b)$ for all $r \in \refs(b)$ into $G$
\State Update $\Tips(G)$, reachability summaries, and latest-block map $L$
\State \ForkChoiceUpdate($G$)
\State Gossip $b$ to peers

\State \Return $b$
\end{algorithmic}
\end{algorithm}

\begin{remark}
In this simplified model, consensus is defined purely at the level of blocks and references.
The payload sampling mechanism $\SamplePayload$ is treated as an abstract oracle:
honest nodes may drop conflicting payloads without affecting the evolution of the block DAG.
Thus, even an adversary that injects conflicting transactions into the mempool can at most
reduce application-level throughput, but cannot stall block production or violate the safety
and liveness properties proved in this paper.
\end{remark}

\subsubsection{Block Reception and Integration}

\begin{algorithm}[H]
\caption{Block Reception (\ProjBase)}
\label{alg:block-reception}
\begin{algorithmic}[1]
\Require Incoming block $b$, local DAG $G=(V,E)$, current slot $t$
\Ensure $b$ is accepted into $G$ or deferred/rejected
\If{$b\in V$} \Return \EndIf

\ForAll{$x \in \refs(b)$ \textbf{with} $x \notin V$}
  \State wait until block $x$ arrives and is added to $V$
  \Comment{{\small defer processing of $b$ until all parents are present}}
\EndFor

\State $\mathsf{ok}' \gets \VRFVerify\bigl(\pk^{\mathsf{vrf}}_v;\,\Proj \parallel \Encode(\slot(b)), y(b), \pi(b)\bigr)$
\If{\textbf{not} $\mathsf{ok}'$ \textbf{or not}  $\SIGVerify\bigl(\pk^{\mathsf{sig}}_v;\,\id(b), \sigma(b)\bigr).$} 
 \State \textbf{reject} $b$ \Comment{VRF proof or signature invalid}
\EndIf
\If{$y(b) \ge p(\slot(b),\val(b))$}
 \State \textbf{reject} $b$ \Comment{VRF output above threshold}
\EndIf

\If{$\exists r\in \refs(b)$ with $(\slot(b)-\slot(r) \le 0)$} 
 \State \textbf{reject} $b$ \Comment{forbid same-slot or future-slot refs}
\EndIf

\If{$|\{r\in \refs(b): \slot(b)-\slot(r) > w\}| > 1$}
 \State \textbf{reject} $b$ \Comment{at most one long-ref (distance $> w$)}
\EndIf

\If{$\exists r\in \refs(b)$ such that $r\in \DescStar(b)$}
 \State \textbf{reject} $b$ \Comment{cycle detected}
\EndIf

\If{$\exists\, r,r'\in \refs(b)$ with $r\neq r'$ such that ( $r \preceq r'$ or $r' \preceq r$ ) in $G$}
 \State \textbf{reject} $b$ 
 \Comment{$\refs(b)$ not an antichain in $(G,\preceq)$}
\EndIf

\If{$\Conflicts(b, \AncStar(\refs(b)))$}
 \State  \textbf{reject} $b$ \Comment{payload conflict with ancestor(s)}
\EndIf

\State Insert $b$ and edges $(r \to b)$ for all $r\in\refs(b)$ into $G$
\State Update $\Tips(G)$, reachability summaries, and latest-block map $L$
\State $\ForkChoiceUpdate(G)$
\end{algorithmic}
\end{algorithm}

\begin{remark}[Late arrivals and window bounds]
\label{rem:late-arrivals}
A node may receive a block $b$ with $\slot(b) < t$. The reception checks are slot-local: for every short reference $r$ of $b$, we enforce $1\le \slot(b)-\slot(r)\le w$ and allow at most one long-ref with $\slot(b)-\slot(\ell)\ge w+1$ (which carries zero weight). Thus a late-arriving $b$ is accepted iff its own short-refs lie in $W(\slot(b);w)$ and no reference targets $\slot(b)$.
\end{remark}

\subsubsection{Local Fork Choice (\ProjBase)}

In \ProjBase, the per-block VRF output serves as the canonical public coin used for
tie-breaking in the conflicted-tips rule. Algorithm~\ref{alg:base-local-fc} is the
practical instantiation of the preferred-frontier rule of
Section~\ref{sec:dag-building-blocks} and Algorithm~\ref{alg:ideal-local-fc}: starting
from the short-reference tips $U_t = \Tips_w(G_t)$, it applies the conflicted-tips
resolution rule $\CTR(\cdot,\cdot)$ to incomparable, conflicting tips, now using the
VRF labels $y(\cdot)$ as the public-coin tie-breaker.

The surviving tips of Algorithm~\ref{alg:base-local-fc} form the preferred frontier
$F_t$ (Defs.~\ref{def:preferred-frontier} and~\ref{def:ctr}): a set of pairwise
non-conflicting tips such that no conflicting tip is preferred over any of them by
the local CCA-based fork-choice rule. 
Several incomparable, non-conflicting tips can survive simultaneously, so Algorithm~\ref{alg:base-local-fc} returns 
the conflict-free preferred frontier $F_t$ (the frontier of \prefsubdags) rather than collapsing to a single \prefsubdag.

\begin{algorithm}[H]
\caption{Local Fork Choice /Preferred Frontier  (\ProjBase)}\label{alg:base-local-fc}
\textbf{Input:} local view $G=(V,E)$; window $w$; VRF public coin $y(\cdot)$\\
\textbf{Output:} preferred frontier $\Undom_t$
\begin{algorithmic}[1]
\State $U \gets \Tips_{w}(G)$
\ForAll{unordered incomparable pairs $\{i,j\}\subseteq U$ \textbf{such that} 
        $\payload(i)$ and $\payload(j)$ conflict}
  \State $\winner \gets \CTR(i,j)$ 
  \Comment{{\footnotesize Conflicted-tips resolution (Def.~\ref{def:ctr}), using VRF-based labels $y(\cdot)$ for tie-breaking}}

  \State $\loser \gets \{i,j\}\setminus\{\winner\}$
  \State mark $\loser$ defeated
\EndFor
\State \textbf{return} $\{x\in U : x \text{ not defeated}\}$ \Comment{{\footnotesize The notion of a ``defeated'' tip is exactly as in
Remark~\ref{rem:defeated-tips}}}

\end{algorithmic}
\end{algorithm}

\paragraph{Eligibility and VRF-based tie-breaking.}
In \ProjBase, the public coin $y(\cdot)$ used in $\CTR(\cdot,\cdot)$ is instantiated
via per-block VRF outputs. Whenever two incomparable tips $i,j$ have the same
window-filtered \subdagweight at their closest common ancestor, the rule
$\CTR(i,j)$ keeps the tip with the smaller VRF label $y(\cdot)$ and marks the other
as defeated. This affects only the membership of tips in $F_t$: 
defeated tips and the \branches rooted at them remain in the DAG and may still 
contribute to block confirmation, 
but they are no longer part of the preferred frontier.

\paragraph{Weight and tie-breaking policy.}
\ProjBase simply reuses the fork-choice policy of Section~\ref{sec:dag-building-blocks}:
\subdagweights count only short references from the last $w$ slots, while long references
carry zero fork-choice weight and serve only 
connectivity and block confirmation.
The per-block VRF output $y(\cdot)$ is used solely as the public coin in $\CTR(\cdot,\cdot)$ for breaking ties between conflicting tips 
of the same weight.

\begin{remark}[Algorithmic obligations supporting TB]
\label{rem:alg-TB}
The local fork-choice rule relies on three operational conditions from the design
of Section~\ref{sec:dag-building-blocks}:
(i) short references are restricted to the last $w$ slots;
(ii) long references carry zero fork-choice weight 
and are included only as mandatory fallbacks when no visible tip lies inside the short-reference window; 
and (iii) whenever eligible, an honest creator includes at least one short reference to a currently visible tip with probability $q > 0$.
Item~(iii) is the coverage condition used in Lemma~\ref{lem:TB-base}.
\end{remark}

\begin{remark}[Maintaining scores in practice]
Conceptually, \subdagscores can be evaluated on demand from the windowed DAG.
Implementations may instead maintain incremental $\mathsf{score}(x)$ values along
their \prefsubdags as new blocks arrive; the protocol specification does not
fix a particular data structure.
\end{remark}

\begin{remark}[\Branch pruning in the ledger view]
When two conflicting blocks appear on different \branches, the fork-choice
rule resolves the conflict via $\CTR(\cdot,\cdot)$ as above. The loser \branch is
pruned from the \emph{ledger view} (kept in the DAG) but is not permanently
discarded until the honest \prefsubdag reaches finality depth $k$
(Section~\ref{sec:finality}).
Ledger view here refers to any analysis-level linearization consistent with the DAG; we do not fix a specific client implementation. 
\end{remark}

\paragraph{Finalization heuristic (pointer).}
We adopt the global rule from Section~\ref{sec:finality}: finalize a block once it
is $k$-deep under the preferred frontier (Def.~\ref{def:finalization-rule}). The
choice of $k$ and the wall-clock mapping $T_{\mathrm{final}} = k\cdot\tau_{\mathrm{slot}}$
follow Theorem~\ref{thm:master-finality} and Corollary~\ref{cor:practical-finality};
deployment tuning is centralized in Section~\ref{sec:finality-deployment}, and we
avoid restating constants here to keep a single source of truth.

\subsection{Security of the Practical Protocol}
\label{sec:base-security}

\paragraph{Proof roadmap under the standard bounded-delay model (unknown $\Delta$). }
Assuming honest stake $H{>}1/2$, VRF unpredictability, and $w\!\ge\!\Delta$, we establish DG/DQ/DCP/TB analogously to the ideal case. The only change is that TB scales with $(\lambda_h+\lambda_a)\Delta$ due to message delay: honest tips may momentarily widen but remain $O((\lambda_h+\lambda_a)\Delta)$ w.h.p. The same window-gap martingale argument then yields $(k,\varepsilon)$-finality with $k=\Theta(w)+O(\log(1/\varepsilon))$.

We prove the DAG properties and ledger properties for \ProjBase under~\eqref{eq:base-params}. 
We write $\lambda_h \coloneqq \mathbb{E}[\#\text{honest-eligible per slot}] = \Theta(pH|\mathcal{P}|)$, where $|\mathcal{P}|$ is the number of validators.

\paragraph{Cryptographic assumptions.}
As in Ouroboros Praos~\cite{EC:DGKR18}, in this section we assume a
secure verifiable random function (VRF) $(\VRFGen,\VRFSign,\VRFVerify)$
and a secure digital signature scheme
$(\SIGGen,\SIGSign,\SIGVerify)$: the adversary cannot forge a valid VRF
proof or a valid signature for an honest validator except with negligible
probability. Equivalently, one may view the protocol as running in a
hybrid model with ideal functionalities for VRFs and signatures, as in
Praos. Under these assumptions, every valid block attributed to $v$ must
have been produced by $v$ in a slot in which $v$ was VRF-eligible, and
our combinatorial analysis focuses solely on the resulting leader
schedule and DAG evolution.

\subsubsection{Consensus properties: DG, DQ, TB, DCP}
\begin{theorem}[DAG Growth (DG) in \ProjBase]
\label{thm:base-DG}
For any interval of $\ell$ slots for the entire execution, at least $\tau_D\ell$ honest blocks become ancestors of some honest tip by time $t{+}\ell{-}1{+}\Delta$, with probability $1-\negl(\secp)$, for $\tau_D=\Theta(pH)$.
\end{theorem}
\begin{proof}[Proof sketch]
In each slot $t$ the number of honest eligible is $\mathrm{Binomial}(|\mathcal{P}|,p\cdot \stakefrac{\cdot})$; 
 Chernoff concentration gives $\Omega(\lambda_h)$ honest blocks per slot over $\ell$ slots. Delivery within $\Delta$ ensures these blocks are referenced by subsequent honest blocks inside the window $w\ge \Delta$, hence they become ancestors of honest tips.
\end{proof}

\begin{theorem}[DAG Quality (DQ) in \ProjBase]
\label{thm:base-DQ}
Fix any $\ell\ge \Omega(\log\secp)$ for the entire execution. The fraction of honest blocks among those entering the preferred frontier during the interval is at least $\mu_D=\Theta(H)$ w.h.p.
\end{theorem}
\begin{proof}[Proof sketch]
VRF eligibility is stake-proportional; the expected honest fraction per slot is $H$. 
The window-filtered local scoring (short references inside the last $w$ slots only) does not amplify adversarial inputs beyond the window; 
by linearity and concentration over $\ell$ slots the realized fraction stays within a small deviation of $H$.
\end{proof}

\begin{lemma}[TB in \ProjBase]
\label{lem:TB-base}
\label{lem:base-TB}
Using the coverage assumption of Remark~\ref{rem:alg-TB}(iii) with parameter $q$,
under the standard bounded-delay model (unknown $\Delta$), short-reference window $w\ge \Delta$, and an
honest short-reference policy that includes at least one visible tip with probability $q>0$ whenever eligible,
the number of tips visible to an honest party by the end of any slot $t$ satisfies
\[\Pr\!\left[\,|\Tips_t|\;\le\; c\cdot (\lambda_h+\lambda_a)\cdot \Delta\,\right]\;\ge\;1-\negl(\secp),\]
for a universal constant $c$, where $\lambda_h,\lambda_a$ are honest and adversarial per-slot creation rates.
Equivalently, $\beta = O(\lambda\cdot \Delta)$.
\end{lemma}
\begin{proof}[Proof sketch]
\emph{Visibility window.} 
Because messages are delivered within at most $\Delta$ slots, any block older than $\Delta$ slots is visible to all honest parties. With $w \ge \Delta$, such a block remains inside the short-reference window long enough to be short-referenced by honest blocks, so it cannot remain a tip beyond age $w$.

\noindent
\emph{Births vs deaths.} Tips are born within the last $\Delta$ slots (worst-case skew)
and die once they receive a short-descendant. Honest eligible blocks provide an independent per-slot coverage probability
at least $q\lambda_h$ for each visible tip. 

\noindent
\emph{Bound.} Across a $\Delta$-slot horizon, tip births concentrate around
$(\lambda_h+\lambda_a)\Delta$; a birth-death coupling plus Chernoff bounds yields the stated w.h.p.\ cap.
\end{proof}

\begin{theorem}[DAG Common Past (\DCP) in \ProjBase]
\label{thm:base-DCPprime}\label{thm:base-DCP}
There exists $k_D=\Theta(\Delta)$ such that, for any two honest parties $P$, $Q$ and times $t_1 \le t_2$,  
after removing the last $k_D$ layers 
from $\Past_P(t_1)$ and $\Past_Q(t_2)$, the earlier trimmed past is contained in the later
one (and in fact the trimmed pasts coincide) with high probability; formally,
$\Trim(\Past_P(t_1), k_D) \subseteq \Trim(\Past_Q(t_2), k_D)$.
\end{theorem}

\begin{proof}[Proof sketch]
By Lemma~\ref{lem:TB-base}, the preferred frontier has bounded width $O((\lambda_h+\lambda_a)\Delta)$ w.h.p.  
Once a block is older than $\Theta(\Delta)$ slots, it has been delivered to all honest parties and short-referenced from their
respective frontiers within the window $w \ge \Delta$. 
Therefore, the CCAs of the preferred frontiers 
$\Undom_t$ across honest parties move forward monotonically, and any divergence is confined to the last $\Theta(\Delta)$ layers.
Trimming these layers from $\Past_P(t)$ and $\Past_Q(t)$ yields identical past sets, giving \DCP with $k_D = \Theta(\Delta)$.
\end{proof}

\subsubsection{Ledger properties: Persistence and Liveness}

\begin{theorem}[Persistence (Safety) of \ProjBase]\label{thm:base-safety}
Let $k,\varepsilon$ be any parameters that satisfy the master finality guarantee of \S\ref{sec:finality}
(Thm.~\ref{thm:master-finality}, instantiated by Cor.~\ref{cor:practical-finality}). If a block $B$ is
$k$-deep under the preferred frontier of an honest party by the end of some slot, then with probability
at least $1-\varepsilon$ the block $B$ remains in the ledger of every honest party forever after.
\end{theorem}
\begin{proof}
Immediate from Def.~\ref{def:finalization-rule} and Cor.~\ref{cor:practical-finality}.
\end{proof}

\begin{theorem}[Liveness of \ProjBase]
\label{thm:areon-base-liveness}
\label{thm:base-liveness}
Fix a slot $t$ and let $b$ be an honestly created valid block in slot $t$ whose payload does not conflict with 
the preferred frontier and any of its descendants. 
Under the standing assumptions of this section (bounded-delay network with parameter $\Delta$ and window size $w \ge \Delta$, VRF-based eligibility with honest stake fraction $H > 1/2$), there exists
$\ell_{\mathsf{live}} = \tilde{O}\!\left(\frac{1}{p H} + \Delta\right)$
such that, except with negligible probability, by slot $t + \ell_{\mathsf{live}}$ the block $b$ is $k$-deep under the preferred frontier of every honest party (and hence appears in every honest ledger view).
\end{theorem}

\begin{proof}[Proof sketch]
With probability $\Omega(p H)$ per slot, some honest validator is eligible. Once the block $b$
is created in slot $t$, it is delivered to all honest parties within at most $\Delta$ slots by the
bounded-delay assumption. From that time onward, each honest eligible block includes, with
independent probability at least $q > 0$, a short reference into the currently visible preferred
frontier.

By Lemma~\ref{lem:base-TB} (Tip-Boundedness), the number of honest tips is bounded by
$\beta = O((\lambda_h + \lambda_a)\Delta)$ with high probability, so a constant fraction of new
honest short references lands on descendants of $b$ in every honest view. Using the same
martingale-style argument as in the proof of Theorem~\ref{thm:liveness-ideal} (adapted to the bounded-delay
setting), we obtain that, except with negligible probability, within $\tilde{O}(\frac{1}{p H})$ additional
slots the block $b$ accumulates at least $k$ short-ref-confirming descendants under the preferred
frontier of every honest party. Accounting for the initial $\Delta$-slot diffusion delay gives the
claimed bound $\ell_{\mathsf{live}} = \tilde{O}(\frac{1}{p H}+ \Delta)$.
\end{proof}

\subsubsection{Coupling to the Idealized Analysis}
\label{subsec:base-coupling}
\ProjBase is a refinement of \ProjIdeal: unbounded referencing is replaced with \emph{window-bounded} short references plus a  long-ref; public coins are realized with VRFs; synchrony is relaxed to the bounded-delay model with $w\ge \Delta$. The proofs above adapt the arguments of Section~\ref{sec:ideal} by (i) adding $\Delta$-slot delivery delays in the TB and 
\DCP lemmas; (ii) ensuring that window-filtered weights ignore stale withheld work; and (iii) using stake-proportional VRF concentration for DQ. Thus, each ideal lemma has a practical analogue with slightly weaker (but still constant) parameters.

The practical \DCP parameter $k_D = \Theta(\Delta)$ plays the same role as the constant $k_D$ in the ideal model:
it bounds how far back the CCAs of honest preferred frontiers may differ. This allows us to apply the
same windowed gap process and finality theorem from Section~\ref{sec:finality}.

\subsubsection{Design Soundness: Remarks and Additional Concerns}
\label{subsec:base-remarks}
\begin{remark}[Window size vs. delay]
If $w<\Delta$, honest tips can \emph{expire} from the short-reference window before cross-linking, violating TB and undermining DCP and safety. Hence the constraint $w\ge \Delta$ is \emph{necessary}.
\end{remark}

\begin{remark}[Long-ref policy]
Long references are used to maintain connectivity and block reachability: they extend the
ancestor closure used for parent fetching and for checking which blocks 
are (logically) confirmed by a tip.
However, long references always carry weight~$0$ and are ignored
by the fork-choice ranking itself: they do not contribute to $\BranchW(\cdot)$ and do not change
the outcome of pairwise comparisons in $\CTR(\cdot,\cdot)$. When deciding whether a candidate tip
can be extended, or when validating a newly received block, nodes still run the usual
conflict checks against the full ancestor closure reachable via both short and long references;
in particular, any block that is 
only reachable via a long reference and thus block confirmed by that tip
must also be conflict-free
with respect to the block's payload. Thus long refs influence \emph{validity} (which blocks can be validated and confirmed), but not the \emph{ordering} of tips in fork choice. Implementations
should deterministically choose the long reference~$\ell$ (e.g., the last finalized checkpoint,
or the $\prec$-max ancestor of the smallest short-ref) to avoid adversarial steering and to
enable fast parent fetching.
\end{remark}

\subsubsection{Discussion of Security and Parameters}

\noindent
\textbf{Safety:} The combination of antichain reference rules, 
conflict checks based on the abstract predicate $\Conflicts(\cdot,\cdot)$,
and the CCA-local fork choice ensures that once honest blocks gain a majority of in-window
weight over any \adversarialsubdag, the honest \subdag will prevail.

\smallskip

\noindent
\textbf{Liveness:} Fast convergence in \ProjBase comes not only from producing more honest blocks but also from eliminating honest forks. Multiple proposers increase the rate at which honest votes (short references) are cast, while the antichain reference rule and the CCA-local window-filtered fork choice continually merge visible honest tips and enforce Tip-Boundedness (TB). As a result, honest weight does not disperse across many long-lived honest \subdags but quickly concentrates on a small set of preferred tips. 
With $w$ chosen sufficiently large relative to $\Delta$ and occasional longrefs used only to maintain connectivity, 
any valid block proposed by an honest party eventually becomes $k$-deep under the preferred frontier of every honest party (and hence appears in every honest ledger view).

Intuitively, the relevant ``speed'' parameter is the honest vote rate per effective tip, roughly $\lambda_h / \beta$; \ProjBase improves both the numerator (more honest proposers) and the denominator (a TB-bounded number of honest forks).

\smallskip

\noindent
\textbf{Resistance to Withholding:} Because \subdagweights are always computed using only in-window references below the
relevant CCA when comparing tips, an adversary cannot secretly build a long private
\branch that overtakes the honest public \subdag upon revelation.

\smallskip

\noindent
\textbf{Parameter Choices:} A window on the order of tens of slots (e.g., $w\approx 30$) often suffices to cover network latency for a moderate honest block rate.

\subsection{Cross-Section Correspondence (Sections \ref{sec:ideal} vs \ref{sec:base})}
The following summary illustrates how the idealized protocol maps to \ProjBase:
\begin{itemize}
 \item \textbf{Eligibility:} Ideal uses implicit Bernoulli coins per slot; Base uses VRF-based eligibility (see Algorithms~\ref{alg:ideal-create}, \ref{alg:block-creation}).
 
 \item \textbf{Antichain selection:} Ideal selects an exact max antichain (Alg.~\ref{alg:ideal-antichain}); 
 Base uses the exact maximum antichain within the window, computed via the Dilworth-based matching algorithm described in Appendix~\ref{sec:appendix-dilworth}.  This matches the idealized reference selection in Section~\ref{sec:ideal} and ensures that parent selection is optimal for honest validators.

 \item \textbf{Long reference:} Ideal has none ($\LongRef$ always returns $\bot$); 
 Base includes at most one deterministic long reference to an ancestor outside the window, used only when no admissible short-reference parent is available in the current window.

\item \textbf{Block creation:} 
Ideal treats blocks as abstract votes and does not model application-level payload validity or VRF-based eligibility explicitly (Alg.~\ref{alg:ideal-create}); 
Base reinstates VRF-based eligibility and enforces application-level payload validity and conflict checks when constructing blocks (Alg.~\ref{alg:block-creation}), and may add such a long reference to ensure connectivity.

\item \textbf{Block reception:} Ideal only checks acyclicity and parent presence (Alg.~\ref{alg:ideal-reception});
Base verifies VRF/signature, window/antichain constraints, and 
payload validity via the abstract conflict predicate $\Conflicts(\cdot,\cdot)$ (Alg.~\ref{alg:block-reception}).

 \item 
 \textbf{Fork choice:} Ideal uses local fork choice with unit weights (Alg.~\ref{alg:ideal-local-fc}); 
 Base uses unit-weight references only; stake is already encoded through VRF eligibility probabilities, so no additional stake-weighting of references is used or needed (Alg.~\ref{alg:base-local-fc}).
In both cases, the fork choice is set-valued and maintains a preferred frontier $\Undom_t$,
obtained by applying $\CTR(\cdot,\cdot)$ to all incomparable, conflicting tips.

 \item 
 \textbf{Conflict resolution:} 
 Ideal and Base use the same CCA-local rule, implemented
  internally inside the fork-choice algorithms (Alg.~\ref{alg:ideal-local-fc} for Ideal and Alg.~\ref{alg:base-local-fc} for Base). 

\end{itemize}

%% file: sec5.tex

\section{Finality in Our Designs}
\label{sec:finality}

In this section, we first fix the ledger properties and the finalization rule (Def.~\ref{def:finalization-rule}).
We then formalize the stochastic driver (the gap process) and prove a master $(k,\varepsilon)$-finality theorem
that applies to both the ideal design of \S\ref{sec:ideal} and the practical design of \S\ref{sec:base}.
We conclude with a deployment recipe and an empirical methodology that operationalizes the parameters and
reports wall-clock settlement times.

\subsection{Finality in chains vs.\ DAGs: theory \& practice}\label{subsec:chains-vs-dags}

\paragraph{(i) Longest \emph{chain} protocols (Bitcoin, Praos).}
In linear chains there is a single tip; backbone analyses (CP/CG/CQ) imply that a block
that is $k$-deep on that chain is reversed with probability that decays exponentially in $k$.
Operationally, deployments expose a depth parameter $k_{\text{chain}}$ (e.g., $6$ for Bitcoin on long timescales)
and map it to wall clock via $T_{\text{final}}=k_{\text{chain}}\cdot\tau_{\text{slot}}$.

\paragraph{(ii) Why this does \emph{not} lift naively to longest DAGs.}
DAG protocols admit many concurrent tips.
Two failure modes appear if one tries to reuse chain proofs verbatim:
(1) \emph{Fragmentation.} Honest work can split across many incomparable tips; without a bound,
no \branch accumulates a stable honest margin even if CP/CG/CQ analogues hold locally.
(2) \emph{Long-range withholding.} If fork choice aggregates weight since genesis,
an adversary can reveal a large private 
\subdag 
that overpowers the public \subdag even under synchrony.

\paragraph{(iii) Our treatment for longest DAGs (this paper).}
We pair two ingredients:
(A) a \emph{window-filtered, CCA-local} fork choice that counts only short references from the last $w$ slots, and
(B) a \emph{Tip-Boundedness} (TB) invariant that caps the frontier width (w.h.p.).
Together they yield a \emph{windowed gap process} with positive drift and bounded differences, so the tail probability that a $k$-deep block is ever removed decays as $\exp(-\Theta(k/w))$. 
We expose the same ``$k$-depth'' operational rule as chains, but with the DAG-calibrated bound
\begin{equation}\label{eq:dag-finality-scaling} 
  k \;\ge\; C_1\,w \;+\; C_2\,\ln\!\Big(\tfrac{1}{\varepsilon}\Big) \quad\Longrightarrow\quad
  \Pr[\text{$k$-deep block ever reverts}] \le \varepsilon,
\end{equation}
and $T_{\text{final}}(\varepsilon)=k\cdot\tau_{\text{slot}}$. Here, depth $k$ is measured \emph{along short-reference edges within the window and under the current CCA anchor (analysis-level anchor; not protocol state)}; long references carry zero fork-choice weight (see Remark~\ref{rem:depth-window}).

\smallskip

\noindent
\emph{Meaning of $C_1,C_2$.} $C_1$ absorbs worst-case delay fluctuation and TB/coverage slack;
once the window length $w$ is chosen large enough to dominate the network delay bound $\Delta$
(cf.~Section~\ref{sec:model}), $C_1$ becomes an absolute $O(1)$ constant independent
of the concrete latency. 
$C_2$ is a universal constant coming from the exponential concentration bound (Azuma--Hoeffding / Chernoff-type) for the windowed gap process. 
Both depend monotonically on $(H-1/2)$, the honest short-ref coverage $q$, and the per-slot honest rate $\lambda_h$; larger values
increase drift and allow smaller $k$.

\subsection{Ledger Properties, Tip-Boundedness, and the Finalization Rule}

We adopt the standard ledger properties Persistence (safety) and Liveness from the blockchain backbone line. 
Persistence requires that once a block is sufficiently deep in the ledger prefix of an
honest party, all honest views agree on it forever after. Liveness requires that any valid block whose
payload does not conflict with the ledger eventually appears in every honest ledger.

In our DAG setting, Tip-Boundedness (TB) is the key structural invariant that underpins these
properties: it prevents the honest-preferred frontier from fragmenting across too many incomparable tips and ensures that honest short-reference weight concentrates on a small number of
preferred subDAGs (see Section~\ref{sec:model} for the general TB discussion).

\begin{definition}[Finalization rule]
\label{def:finalization-rule}
\label{def:finality-rule}
Fix a window width $w \ge \Delta$ and the protocol's window-filtered fork-choice rule. Let $F_t$ denote the preferred frontier of an honest node at time $t$. Anchoring is used only in the analysis layer to define the depth predicate; it is not part of the protocol state.

We say that a block $B$ is $k$-final at time $t$ if $B$ lies at distance at least $k$ beneath $F_t$ when distances are measured along short-reference edges (i.e., within-window references) under the current CCA-based anchor used in the analysis. The protocol achieves $(k,\varepsilon)$-finality if, for any honest node and any time $t$ at which $B$ is $k$-final, the probability that $B$ is ever removed from that node's ledger view after time $t$ is at most $\varepsilon$.
\end{definition}

\begin{remark}[Depth convention and anchors]\label{rem:depth-window}
Throughout this section, ``depth'' is always measured along short-reference edges within the last $w$ slots between a block and the preferred frontier $F_t$, under the CCA-based anchor used in the analysis. Operationally, the protocol layer maintains only the DAG and the local preferred frontier $F_t$; the anchor and the associated linearization of $\Anc^*(F_t)$ are purely analytical devices, analogous to the ``$k$-confirmation'' convention in longest-chain protocols. Implementations are free to adopt this $k$-depth rule at the client side or to expose any equivalent interface consistent with the DAG’s partial order.
\end{remark}

\subsection{Stochastic Driver: The Windowed Gap Process}
\label{subsec:gap-process}

\paragraph{Canonical honest tip.}
For the purposes of analysis, we fix once and for all a deterministic tie-breaking rule that,
given a non-empty frontier $F_t$, selects a canonical tip $i_t \in F_t$ (for example, the tip with
the smallest $(\slot, y(\cdot))$ pair in lexicographic order). All statements referring to
``the honest-preferred tip'' use $i_t$ as this canonical representative.

We now make explicit the stochastic variable that concentrates in our analysis. At time $t$,
let $F_t$ be the preferred frontier under the local CCA-based fork choice, and let $j$ be any
competing tip. Let $c = \CCA(i_t, j)$ be the closest common ancestor of the canonical honest
tip $i_t$ and $j$. Write $\Desc^*(x)$ for the set of descendants of $x$ reached by short
references (windowed edges). Let $\wref(d; w) \in \{0, 1\}$ indicate whether a block $d$
contributes short-reference weight within the current window of width $w$.

\begin{definition}[Windowed gap]\label{def:gap}
Fix a time $t$ and let $i_t$ be the canonical honest tip. For any competing tip $j$ and
anchor $c = \CCA(i_t, j)$, over a window of width $w \ge \Delta$ the (windowed) gap random
variable between $i_t$ and $j$ is
\begin{equation}
  G
  \;=\;
  \sum_{d \in \Desc^*(i_t) \cap \Desc^*(c)} \wref(d; w)
  \;-\;
  \sum_{d \in \Desc^*(j)   \cap \Desc^*(c)} \wref(d; w).
\end{equation}
\end{definition}

\paragraph{Intuition for Definition~\ref{def:gap}.}
The windowed gap $G$ measures the short-term advantage of the canonical honest \subdag of $i_t$ over any competitor $j$.
Each short reference acts as a ``vote'' for the \branch it descends
from.  Because $\wref(d;w)$ counts only references created within the most recent
$w$ slots, $G$ captures which \branch has accumulated more honest support in the
recent execution window.  Concretely, $G>0$ means that, among descendants of the
anchor $c$, the \branch of $i$ has received more short references than that of
$j$ within the window $W(t;w)$.

The use of descendant intersections ensures that only references relevant to the
competition between $i$ and $j$ are counted, while Tip-Boundedness guarantees
that each slot contributes at most $O(\beta)$ new references, yielding bounded
increments $|G_{t+1}-G_t|\le O(\beta)$.  Since honest blocks appear with rate
$\lambda_h$ and reliably reference the preferred frontier, the gap has a
positive drift $\mathbb{E}[G_{t+1}-G_t \mid \mathcal{F}_t]\ge \gamma \lambda_h
w$, so an honest \subdag steadily accumulates margin.  These two properties
(positive drift and bounded increments) allow Azuma--Hoeffding concentration,
leading to the exponential finality tail proved in Theorem~\ref{thm:master-finality}.

We assume (as part of the standing assumptions) that:
(i) VRF eligibility outcomes across validators and across slots are mutually independent and are revealed only upon issuance;
(ii) the honest coverage probability per slot is bounded below by a constant $q>0$ (i.e., some honest short reference appears with probability $\ge q$);
(iii) the network satisfies partial synchrony with delay parameter $\Delta$, and TB holds with high probability as a consequence of $w\ge \Delta$ and coverage.

Under these assumptions there exists $\gamma=\gamma(H,q)>0$ such that, conditioning on the past (the natural filtration),
\begin{equation}
  \mathbb{E}[\,G\,] \;\ge\; \gamma\,\lambda_h\,w \quad\text{and}\quad
  \{G_\ell\}_{\ell\ge 0}\ \text{forms a submartingale with bounded increments}.
  \label{eq:gap-drift}
\end{equation}

\begin{remark}[On $\beta$ from Tip-Boundedness]\label{rem:tb-beta}
Here $\beta$ denotes the \emph{Tip-Boundedness cap} on the number of visible short-reference tips
$|\mathrm{Tips}_t|$ (Def.~\ref{def:TB} in \S\ref{sec:model}). In our models, $\beta=O(\lambda)$ in the ideal case
(Lemma~\ref{lem:ideal-TB}) and $\beta=O((\lambda_h+\lambda_a)\Delta)$ in the practical case (Lemma~\ref{lem:base-TB}).
\end{remark}

We rely on \DCP to ensure that the CCA $c$ used to define the gap between $i$ and any competitor $j$ lies in
the common past of all honest views, up to the trimming parameter $k_D$. Intuitively, once we ignore
the last $k_D$ layers, all honest parties agree on the subDAG below $c$, so the windowed gap process
measured from $c$ is well-defined and comparable across honest views.

\begin{lemma}[Bounded differences via TB]\label{lem:bd-consolidated}
Let $G_t$ be the windowed gap (Def.~\ref{def:gap}). Suppose TB$[\beta]$ holds and $w \ge \Delta$.
Then there exists a constant $C_\beta = O(\beta)$, depending only on the TB parameter and the per-slot
short-reference multiplicity, such that
\[
  |G_{t+1} - G_t| \le C_\beta \quad \text{a.s. for all } t.
\]
Together with the honest-rate drift
$\mathbb{E}[G_{t+1}-G_t \mid \mathcal{F}_t] \ge \gamma \lambda_h$ for some
$\gamma = \gamma(H,q) > 0$, this bounded-differences condition allows us to apply
the Azuma--Hoeffding inequality to $(G_t)$ and derive an exponential tail for the time
needed to accumulate a margin $\Theta(k)$.
\end{lemma}

\noindent In particular, TB yields the uniform one-step bound
\begin{equation}
  |G_{t+1}-G_t| \;\le\; C_\beta \quad\text{a.s.},
  \label{eq:gap-step}
\end{equation}
where $C_\beta = O(\beta)$ depends only on the TB cap $\beta$ on $|\Tips_t|$ and the
per-slot multiplicity cap. Here $H>1/2$ is the honest fraction and $\lambda_h$ the honest
per-slot issuance rate. Bounded increments follow from TB (the number and effect of
short references per window are bounded) and from the window filter.

\subsection{Master $(k,\varepsilon)$-Finality Theorem}\label{subsec:master-finality}

The next theorem packages the concentration of the gap process into a universal finality guarantee that
we instantiate for both designs.

\begin{theorem}[Master $(k,\varepsilon)$-finality]\label{thm:master-finality}
Assume $H>1/2$, $w\ge \Delta$, honest coverage $q>0$, and honest rate $\lambda_h>0$.
There exist absolute constants $C_1,C_2>0$ such that if
\begin{equation}
  k \;\ge\; C_1\,w \;+\; C_2\,\ln\!\left(\tfrac{1}{\varepsilon}\right),
  \label{eq:k-form}
\end{equation}
then the protocol achieves $(k,\varepsilon)$-finality under Def.~\ref{def:finalization-rule}.
\end{theorem}

\begin{remark}[Choosing the window $w$ and the rule-of-thumb ``$w = k$'']
Assume a bounded-delay network with delay bound $\Delta$. Our master $(k,\varepsilon)$-finality
theorem (Theorem~\ref{thm:master-finality}) is stated under the requirement that the short-reference window dominates
the delay, for example
\[
  w \;\ge\; \Delta + C_0 \log \kappa
\]
for a suitable absolute constant $C_0$. Once $w$ is chosen large enough to cover the delay bound,
the constant $C_1$ in Equation~\eqref{eq:k-form} absorbs worst-case delay fluctuation and the slack in the
Tip-Boundedness and coverage assumptions, and can be treated as an $O(1)$ constant independent
of the concrete latency.

For any such window width $w$ and target tail probability $\varepsilon$, Theorem~\ref{thm:master-finality} shows that it
suffices to take
\[
  k \;\ge\; C_1 w + C_2 \ln \frac{1}{\varepsilon}.
\]
In many deployments it is convenient to tie the confirmation depth $k$ to the window length when
mapping depth to wall-clock latency, and a rule-of-thumb such as $k \approx w$ (up to an $O(1)$
factor) is often reasonable. Our analysis, however, does not require $k$ and $w$ to coincide: for a
fixed tail target $\varepsilon$, increasing $w$ beyond the smallest value that safely dominates $\Delta$
never decreases the required depth $k$, because the linear term $C_1 w$ in Equation~\eqref{eq:k-form} grows with
$w$. Practically, one should choose the smallest window $w$ that reliably covers the network delay,
and then set $k$ using the formula above.
\end{remark}

\begin{proof}[Proof sketch]
Consider any time $t$ and a block $B$ that is $k$-final (Def.~\ref{def:finalization-rule}).
For $B$ to be removed later, some competitor \branch must overtake the preferred frontier far enough back to erase $k$ levels.
Anchoring at the CCA, the windowed gap process $G$ (Def.~\ref{def:gap}) experiences a positive drift
$\mathbb{E}[G]\ge \gamma\lambda_h w$ (Eq.~\eqref{eq:gap-drift}).
By TB and windowing, per-window changes are bounded, so $(G_t)$ is a submartingale with bounded differences
(Lemma~\ref{lem:bd-consolidated}).
Applying the Azuma--Hoeffding inequality to $(G_t)$
bounds the probability that the competitor cancels the gap needed to erase $k$ levels.
This yields an exponential tail in $k$ with scale proportional to $w$, giving~\eqref{eq:k-form}.
\end{proof}

\begin{corollary}[Ideal design]\label{cor:ideal-finality}
Under the modeling assumptions of \S\ref{sec:ideal}, Theorem~\ref{thm:master-finality} holds.
In particular, there exist constants $c_1,c_2>0$ such that
\begin{equation*}
  k \;=\; c_1\,w \;+\; c_2\,\ln\!\left(\tfrac{1}{\varepsilon}\right)
\end{equation*}
suffices for $(k,\varepsilon)$-finality.
\end{corollary}

\begin{corollary}[Practical design]\label{cor:practical-finality}
Under the network and eligibility assumptions of \S\ref{sec:base},
Theorem~\ref{thm:master-finality} holds with the same scaling.
If slot duration is $\tau_{\mathrm{slot}}$, the wall-clock settlement bound is
$T_{\mathrm{final}}(\varepsilon)=k\cdot \tau_{\mathrm{slot}}$.
\end{corollary}

\begin{remark}[Parameter transparency]\label{rem:params2}
Eq.~\eqref{eq:k-form} decomposes $k$ additively as $C_1 w + C_2 \ln(1/\varepsilon)$.
At a fixed tail target $\varepsilon$, \emph{increasing $w$ never decreases $k$}; the linear term adds security mass
but also increases the required depth. Improvements from larger honest fraction $H$, coverage $q$, or honest rate
$\lambda_h$ appear by \emph{shrinking the constants} $C_1$ and $C_2$ (greater drift, tighter concentration), not by
making $w$ larger. Practically: choose the smallest safe $w\!\approx\!\Delta_{\max}$, then pick $k$ via $k=C_1 w + C_2\ln(1/\varepsilon)$.
\end{remark}

\subsection{Deployment Recipe and Wall-Clock Mapping}
\label{sec:finality-deployment}

\paragraph{Recipe.} (i) Choose $w \ge \Delta$; (ii) fix a tail target $\varepsilon$ (e.g., $10^{-6}$); (iii) estimate
$\lambda_h$ and $H$ for the target deployment; (iv) set
$
  k = c_1 w + c_2 \ln\!\bigl(1/\varepsilon\bigr)
$
using either conservative analytical $c_1, c_2$ or fitted values from simulation (next subsection);
(v) publish $T_{\mathsf{final}}(\varepsilon) = k \cdot \tau_{\mathsf{slot}}$.

\smallskip
\noindent\textbf{Operational notes.}
\begin{itemize}
  \item The linear-in-$w$ term reflects that the effective security mass per window scales with the
        window width.
        
  \item Choosing $w$ only slightly above $\Delta$ is latency-friendly; larger $w$ increases
  robustness to bursty delays. A deployment that fixes $w$ at an estimated upper bound
  $\Delta_{\max}$ on the network delay trades additional latency
  (since $k = \Theta(\Delta_{\max}) + \Theta(\ln(1/\varepsilon))$)
  for the convenience of treating $C_1$ as an absolute $O(1)$ constant across operating
  regimes.

  \item
  In the ideal model of Section~\ref{sec:ideal}, the tie-breaking label $y(\cdot)$ is an explicit
        public coin; in the practical protocol Areon-Base, $y(\cdot)$ is instantiated by VRF outputs.
        Under the standard pseudorandomness assumption for the VRF, these labels are independent
        across slots and validators, which provides exactly the independence needed for the
        Chernoff/Azuma concentration step in Theorem~\ref{thm:master-finality}.
\end{itemize}

\paragraph{Interpreting $C_1$ and $C_2$.}
$C_1$ absorbs structural slack due to delay skew and TB; 
if one uses $w$ chosen as an upper bound $\Delta_{\max}$ on the network delay,
$C_1$ is effectively an absolute $O(1)$ constant. $C_2$ is the Chernoff/Azuma tail constant controlling the $\ln(1/\varepsilon)$ term.
Both improve (decrease) as $(H-1/2)$, the honest coverage $q$, and the honest rate $\lambda_h$ increase.

\paragraph{Worked example.}
Suppose $\tau_{\mathrm{slot}}=1\,\mathrm{s}$, choose $w=30$, and target $\varepsilon=10^{-6}$. 
If one fits $C_1\approx 1.8$ and $C_2\approx 2.7$, 
then $k \approx 1.8\cdot 30 + 2.7\cdot \ln 10^{6} \approx 54+37 \approx 91$,
hence $T_{\mathrm{final}}\approx 91\,\mathrm{s}$.

In short, $(w,\varepsilon,\tau_{\mathrm{slot}}) = (30,10^{-6},1\,\text{s})$ and
$(C_1,C_2) \approx (1.8,2.7)$ map to $(k,T_{\mathrm{final}}) \approx (91,91\,\text{s})$.

\subsection{Empirical Methodology for Probabilistic Finality}
\label{sec:finality-empirical}

We complement the analysis with a protocol-agnostic measurement procedure.

\begin{definition}[Empirical settlement function]\label{def:settlement}
Let $\widehat{P}_{\mathrm{reorg}}(d)$ be the fraction of runs in which a block at depth $d$ is later reorganized.
Define the empirical settlement depth for tail $\varepsilon$ as
\begin{equation*}
  d^*(\varepsilon)=\min\{\,d:\ \widehat{P}_{\mathrm{reorg}}(d)\le \varepsilon\,\},
  \qquad T_{\mathrm{final}}(\varepsilon)=d^*(\varepsilon)\cdot \tau_{\mathrm{slot}}.
\end{equation*}
\end{definition}

\paragraph{Procedure.}
(1) Fix $(H,\Delta,w,\tau_{\mathrm{slot}})$ and a workload; (2) run Monte-Carlo with a strong withholding adversary
consistent with the model; (3) record the reorg-depth distribution and estimate $\widehat{P}_{\mathrm{reorg}}(d)$;
(4) fit $\widehat{P}_{\mathrm{reorg}}(d) \approx A e^{-\decay d}$
on a log scale;
(5) report $d^*(\varepsilon)$ and $T_{\mathrm{final}}(\varepsilon)$; (6) cross-validate that
$d^*(\varepsilon)$ grows $\Theta(w)+\Theta(\ln(1/\varepsilon))$ across parameter sweeps.

\paragraph{On the fit and mapping to $\varepsilon$.}
We fit $\widehat{P}_{\mathrm{reorg}}(d) \approx A e^{-\decay d}$, where $A$ is the prefactor and $\decay$ the empirical
decay rate. Solving $A e^{-\decay d}\le \varepsilon$ gives the \emph{predictive mapping}
\[
d^*(\varepsilon)\;\approx\;\frac{\ln(A/\varepsilon)}{\decay},\qquad
T_{\mathrm{final}}(\varepsilon)\;=\;d^*(\varepsilon)\cdot\tau_{\mathrm{slot}}.
\]
Under Theorem~\ref{thm:master-finality}, $\decay$ should scale like $\Theta(1/w)$ (holding other parameters fixed); verifying this by sweeping $w$ is a useful sanity check. 

\begin{remark}[Consistency check]\label{rem:consistency}
Empirically observed exponential tails and the scaling of $d^*(\varepsilon)$ with $w$ provide a direct check of
Theorem~\ref{thm:master-finality}. Discrepancies typically indicate either (i) $w<\Delta$, breaking TB,
or (ii) underestimated adversarial scheduling capabilities.
\end{remark}

%% file: sec6.tex

\section{Implementation and Evaluation}
\label{sec:eval}

This section reports the empirical evaluation of \ProjBase. We built a discrete-event simulator and compared \ProjBase to Ouroboros Praos with a focus on reorg depth under increasing network delay and chain parallelization. The simulator simulates all required components for ledger growth and reorg behavior under configurable adversarial stake, network delay, and block-production parallelism (number of concurrent proposers). For better readability, some figures are included in Appendix~\ref{sec:eval-app}. 

\paragraph{Clarification on maximum antichain computation.} Honest validators compute the \emph{exact maximum antichain} over the local window $G_{t,w}$, which contains only $O(\lambda w)$ blocks; this computation is polynomial-time via Dilworth’s theorem. In contrast, an adversary attempting to maximize disruption must explore a global combinatorial search space involving withheld blocks, multiple release orders, and induced DAG shapes.  This requires evaluating many distinct antichain configurations across exponentially many possibilities. Thus, the computational burden applies only to the adversary, not honest nodes.

\subsection{Implementation Overview}
The simulator models validators, network propagation, and the \ProjBase{} fork-choice. Each block stores $(\id,\val,\slot,\refs)$ which are essential to model the chain growth as in Alg.~\ref{alg:block-creation}. References are maintained as adjacency lists of the ref-DAG. The simulator exposes:
\begin{itemize}
\item \textbf{Eligibility:} VRF-based slot eligibility $\Eligibility(v,s)$ (Sec.~\ref{sec:notation-base}). Multiple proposers per slot emerge from Bernoulli trials across validators. It uses stake-weighted selection with stake distribution derived from a Pareto distribution.

\item \textbf{Reference selection:}  Within window $w$, each honest producer selects the exact maximum antichain of visible parents (as defined in Section~\ref{sec:ideal}), computed via the Dilworth-based algorithm of Appendix~\ref{sec:appendix-dilworth}. Honest validators always use this exact maximum-antichain procedure in the protocol. 

\item \textbf{Long-ref:} One \emph{long-ref} $\ell$ to a block outside the window (if an unreachable block is detected). Long-refs carry \emph{zero} weight in fork-choice (Sec.~\ref{sec:base}).
 
\item \textbf{Conflict resolution:} CCA-based selection using the same window-filtered \subdagweights $\BranchW(\cdot; c, w)$ as in Sections~\ref{sec:prelim} and~\ref{sec:base} (Alg.~\ref{alg:base-local-fc}). Short references inside the last w slots contribute unit weight and long references carry zero weight.

\item \textbf{Adversary:} An adaptive, withholding adversary that
  (i) hides its blocks;
  (ii) concentrates stake in a single identity;
  (iii) observes the honest network with zero delay;
  and (iv) the adversary attempts to build a heavier \subdag at every single honest block, using different approaches. 
\end{itemize}

\paragraph{Metrics and axes.}
We report 
(i) \emph{reorg length} (depth of deepest reorg experienced by any honest block),
(ii) \emph{finality proxy} (age threshold after which reorg probability $<10^{-5}$),
and (iii) \emph{concurrent tips} $|\Tips(G)|$ over time; we sweep the production factor $f$ and adversary/network parameters.

\subsection{Experimental Setup}\label{sec:setup}
Unless noted, default parameters are:
\begin{center}
\begin{tabular}{lcl}
\toprule
Mean network delay & = & $7.5$ s\\
Reference window $w$ & = & $30$ slots\\
Adversary stake & = & $0.30$ \\
Honest validators & = & $1000$ \\
\bottomrule
\end{tabular}
\end{center}

Honest nodes build best-effort antichains given current visibility; block dependencies (payload-conflict constraints) are enforced. We model a worst case in which the adversary is a single, perfectly coordinated entity with zero observation delay of the honest DAG, yielding a conservative baseline.

Our adversary uses two attack strategies: an ILP-based strategy (worst-case exponential time; see Section~\ref{sec:ilp}) and a separate heuristic strategy. We simulate both and compare their effectiveness; for our parameter ranges, the ILP-based attack is strictly more powerful, while still feasible to run.

\smallskip

\noindent\emph{Caveat.} Increasing $f$ speeds up convergence but also changes the baseline rate; comparisons to single-leader baselines must fix an equalized block-arrival process to avoid bias.

\subsection{Parallelism Bias}
We control block-production parallelism via the parameter $f$, which scales the block production rate relative to the base value.

\begin{figure}[htbp!]
\centering
\includegraphics[width=1\textwidth]{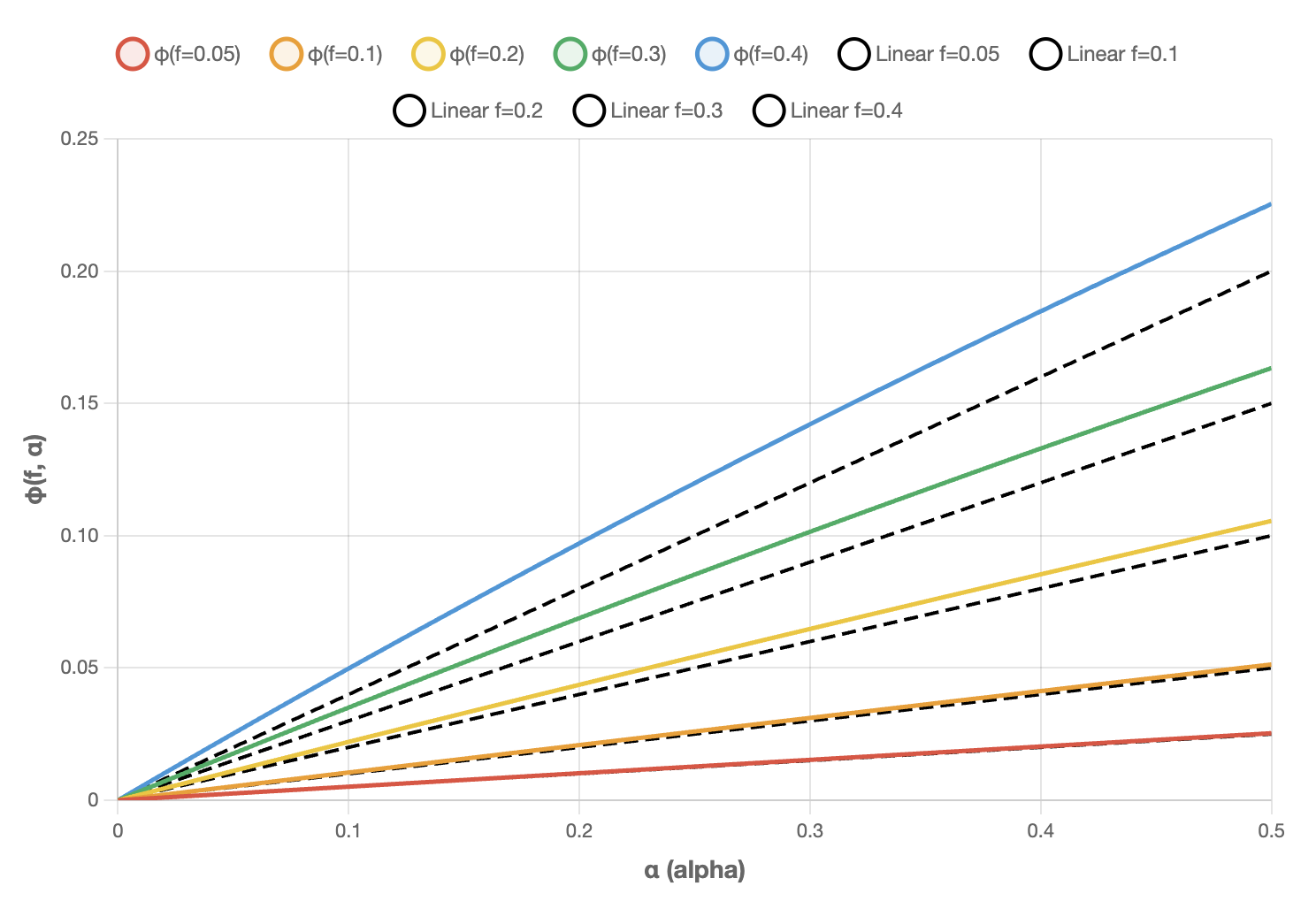}
\caption{Relationship between $f$, node stake $\alpha$, and the probability of winning a slot $\phi(f,\alpha)$.}
\label{fig:f-parallelism-bias}
\end{figure}
Figure~\ref{fig:f-parallelism-bias} shows that increasing $f$ advantages the adversary by raising its chance to win a slot beyond its stake proportion. This effect is stronger when adversarial stake is concentrated in a single entity rather than split across identities.

We choose this method for its simplicity; the resulting overestimation of adversarial capability is acceptable for our purposes. Future work will reduce this bias with an alternative parallelism mechanism.

\subsection{ILP vs.\ Global Heuristic}
\label{sec:ilp}
\label{sec:ILP}

ILP-based optimization (exhaustive search within the $2w$-slot sliding attack window)
yields rarer, longer reorgs but is computationally expensive. The global heuristic scales to many more attacks but produces shorter reorgs on average, consistent with its weaker optimization power.

\subsection{Parallelism}

\begin{figure}[H]
\centering
\subfloat[5x parallelization, $f=0.25$]{\includegraphics[width=1\textwidth/2]{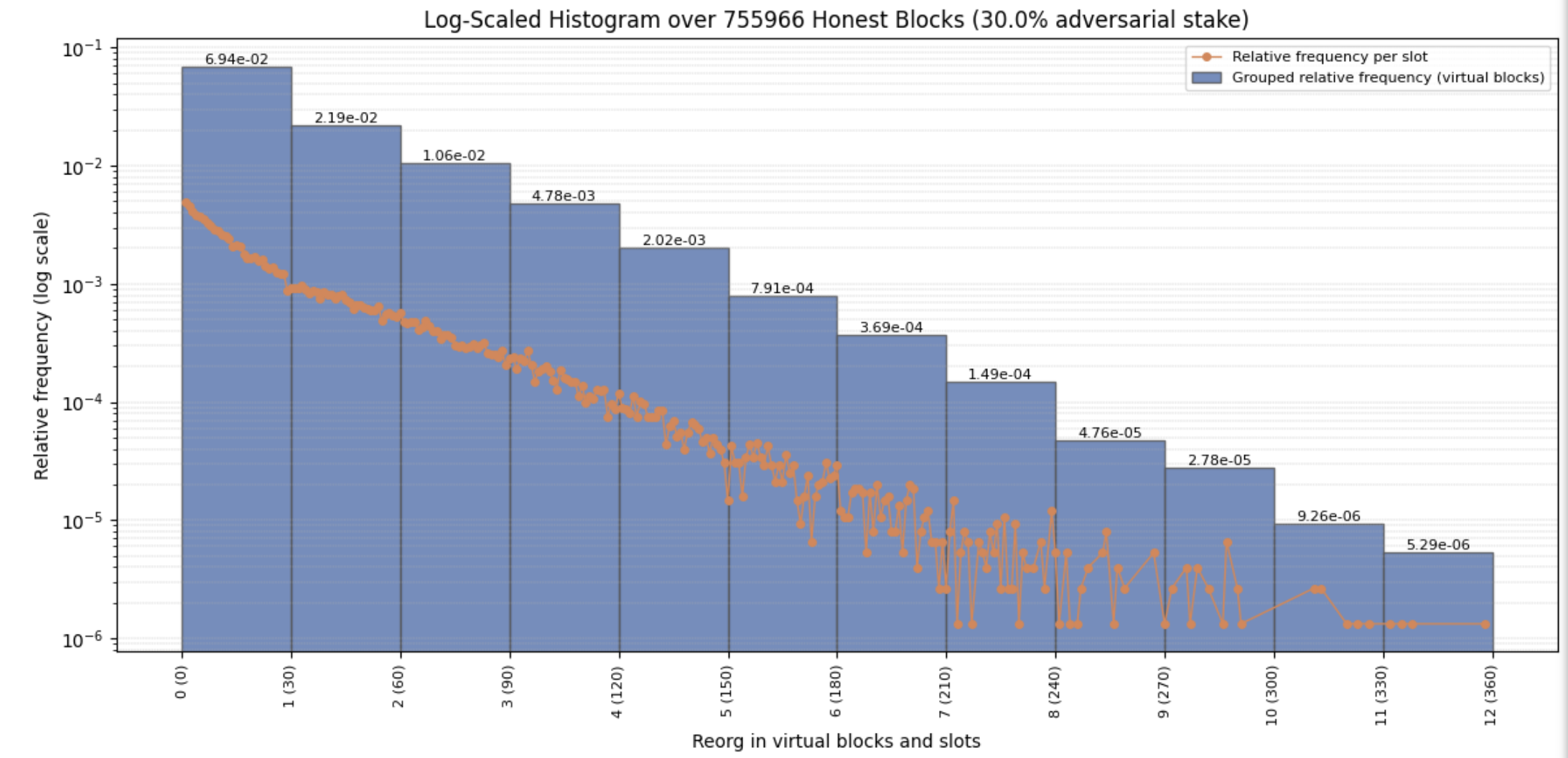}}
\subfloat[6x parallelization, $f=0.30$]{\includegraphics[width=1\textwidth/2]{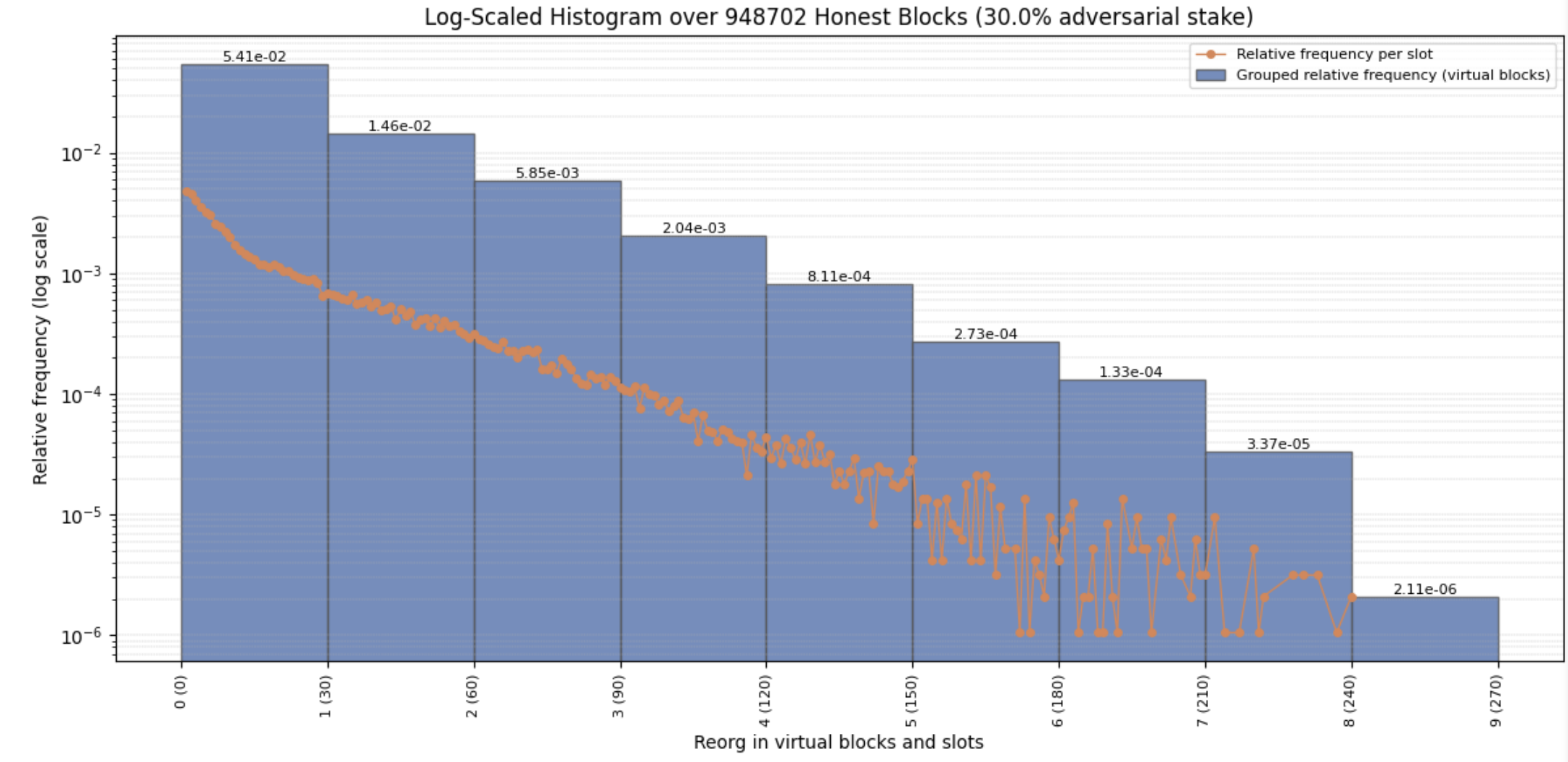}}\\
\subfloat[8x parallelization, $f=0.40$]{\includegraphics[width=1\textwidth/2]{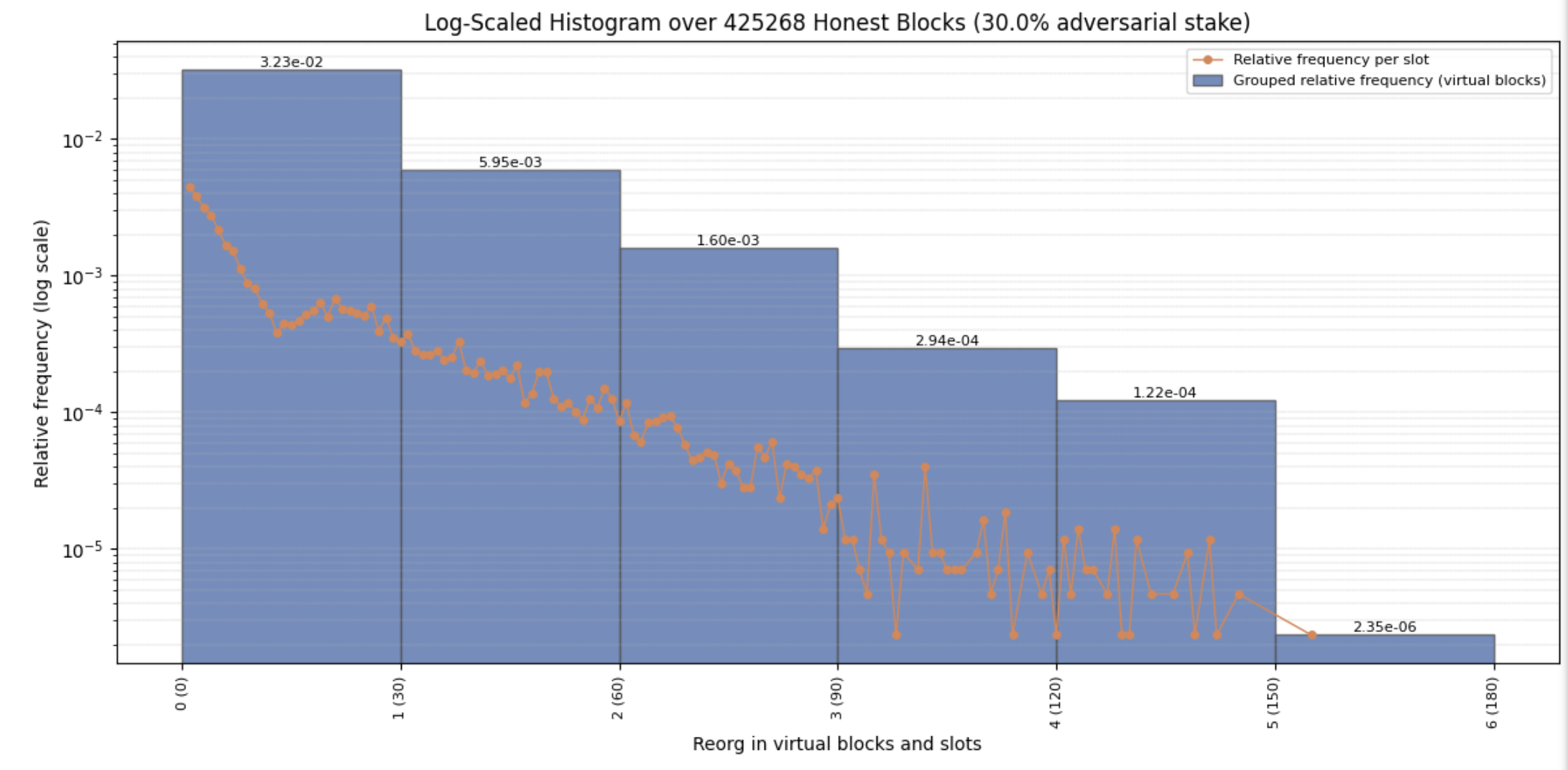}}
\subfloat[10x parallelization, $f=0.50$]{\includegraphics[width=1\textwidth/2]{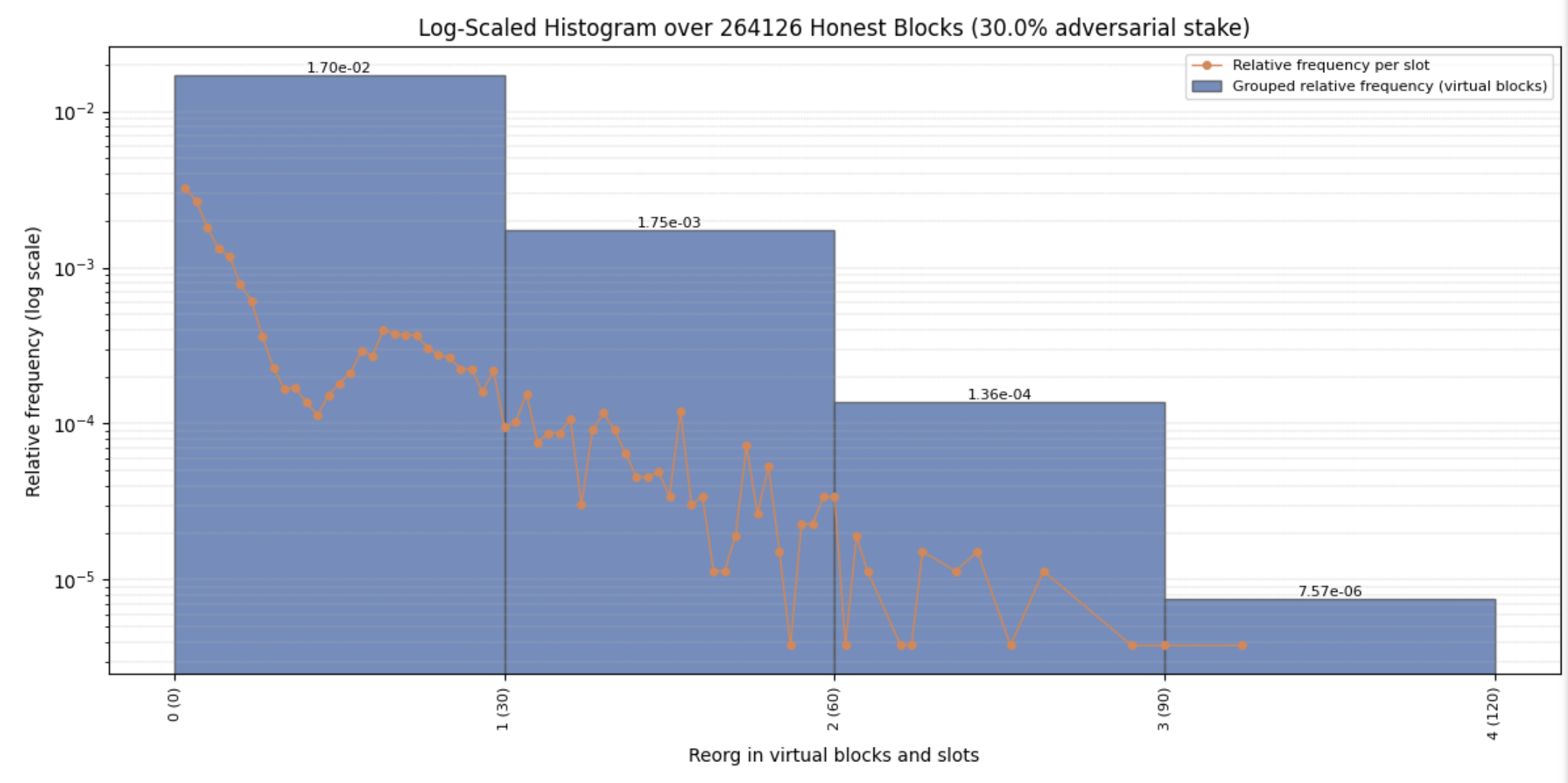}}
\caption{Impact of parallelization level on reorg depth for \ProjBase{} with 30\% adversarial stake.}
\label{fig:parallelization}
\end{figure}
Empirically, higher parallelism (larger $f$) correlates with shorter reorgs, especially in large-scale heuristic runs that enable more attack attempts (Fig.~\ref{fig:parallelization}). ILP-based local optima can surface rare long reorgs in small samples, but across $f\in\{0.25,0.30,0.40,0.50\}$ tail depth decreases as $f$ increases, aligning with the “more votes, faster convergence” intuition. (See Appendix~\ref{sec:eval-app}, Fig.~\ref{fig:c2-delay} for sensitivity across $f$ and mean broadcast delay.)

\subsection{Adversarial stake}

\begin{figure}[H]
\centering \includegraphics[width=\textwidth]{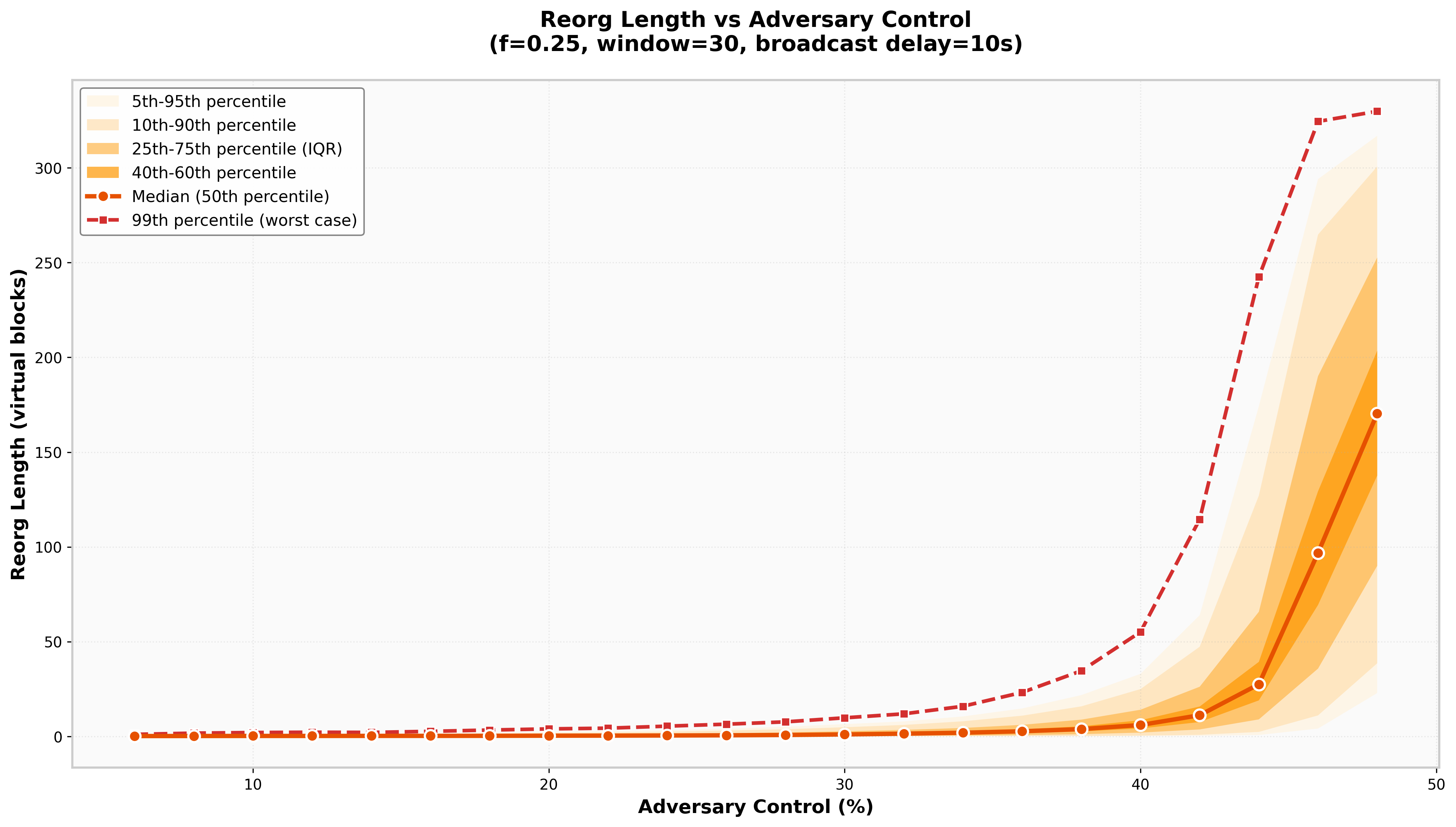}
\caption{Dynamics of the adversarial stake impact on the reorg depth for \ProjBase.}
\label{fig:reorg-adversary-dynamics}
\end{figure}
As $\alpha$ approaches $0.5$, both the frequency and length of reorgs increase (Fig.~\ref{fig:reorg-adversary-dynamics}). Even at $\alpha=0.40$, reorg lengths remain bounded (e.g., $\le 55$ virtual blocks with $>99\%$ probability) under $w=30$. The tail grows sharply as $\alpha\to 0.5$, reflecting loss of safety at majority. At $\alpha=0.45$, we observe a median reorg depth of around $100$ virtual blocks. For $\alpha<0.3$, typical reorgs stay shallow (depth $\le 10$), while extremes depend on the relationship between $\Delta$ and $w$.

Note that the top of the plot tapers off because it reaches the iteration-length limit: when a curve hits this ceiling, the adversary maintains a reorg that persists up to the established finality time (i.e., the adversary “wins” over the entire finality window). Once a curve hits this ceiling, we already treat it as an adversarial win over the whole finality horizon; the apparent plateau is therefore an artifact of the finite simulation horizon, not a genuine bound on reorg depth.

Figure~\ref{fig:adversarial} (Appendix~\ref{sec:eval-app}) presents a detailed comparison of how adversarial stake affects reorg depth in \ProjBase ($f=0.25$) versus Praos ($f=0.05$).

\subsection{Network delay}

\begin{figure}[H]
\centering \includegraphics[width=\textwidth]{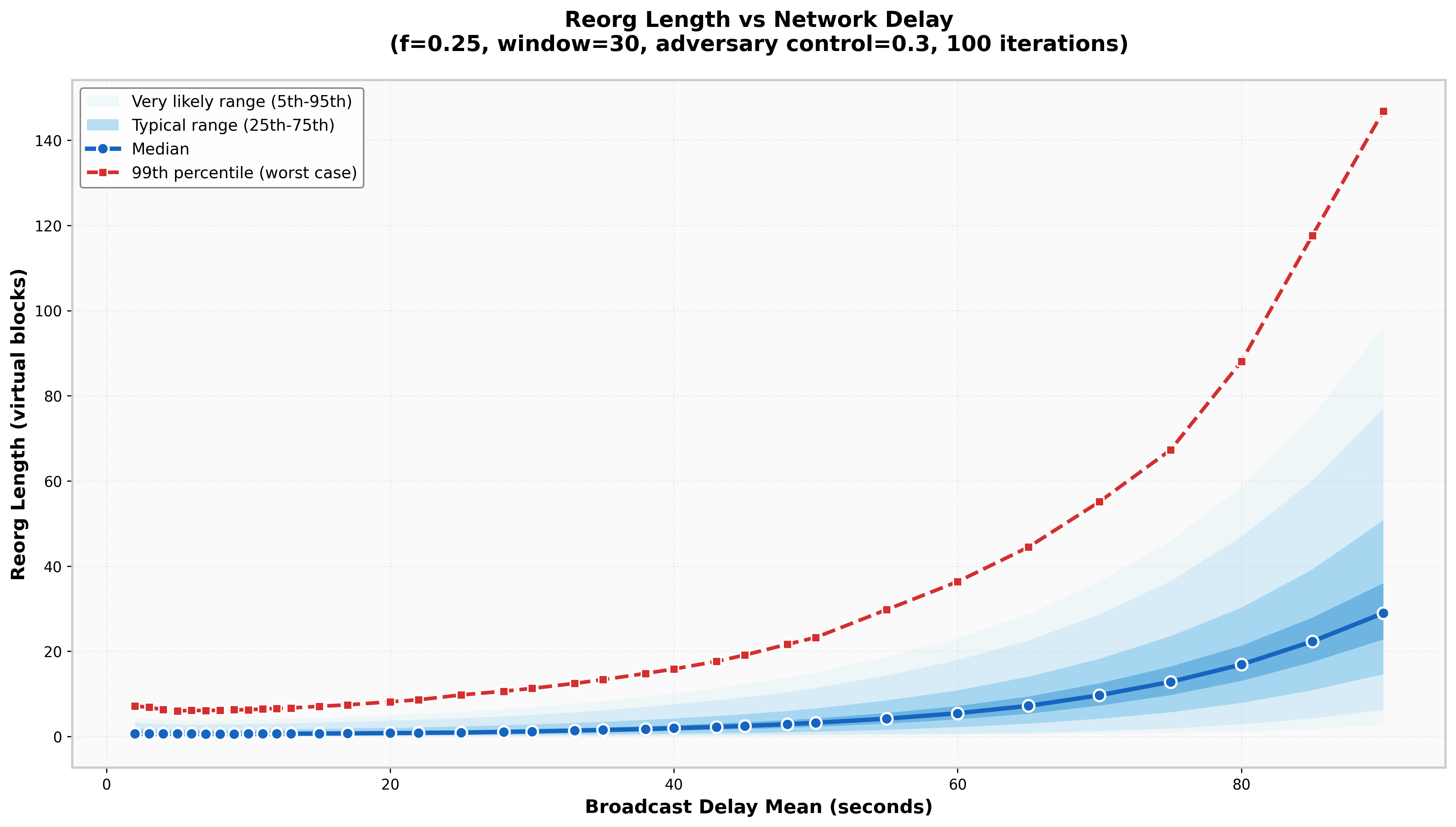}
\caption{Dynamics of the network delay impact on the reorg depth for \ProjBase.}
\label{fig:reorg-delay-dynamics}
\end{figure}
Holding $\alpha=0.3$, increasing $\Delta$ relative to slot time induces more transient forks (Fig.~\ref{fig:reorg-delay-dynamics}). However, if $w\ge\Delta$ (Appendix~\ref{sec:eval-app}, Fig.~\ref{fig:window}), these forks resolve within $w$ slots and the reorg-depth distribution shifts only slightly. If $w<\Delta$, reorgs can grow unbounded (e.g., $w=30$ and $\Delta=60$ slots, Appendix~\ref{sec:eval-app}, Fig.~\ref{fig:window}g), with tips oscillating indefinitely and successful attacks emerging. This confirms the necessity of $w\ge\Delta$ for security (Appendix~\ref{sec:eval-app}, Fig.~\ref{fig:window}), matching the TB requirement.

Figure~\ref{fig:c2-delay} (Appendix~\ref{sec:eval-app}) compares how network delay affects reorg depth in \ProjBase and Praos.  \ProjBase consistently exhibits lower reorg probabilities than Praos at $f=0.05$, and reorg depth decreases as $f$ increases.

\subsection{Finality time}
Using our finality proxy (age after which reorg probability $<10^{-k}$), finality occurs within tens of virtual blocks as long as $\alpha<0.5$ and $w\ge\Delta$ (Appendix~\ref{sec:eval-app}, Fig.~\ref{fig:window}). For example, with $\alpha=0.3$, $\Delta=0.5$ s, and $w=30$, we reach $k=6$ in roughly $30$--$40$ virtual blocks; at $\alpha=0.45$, about $60$ virtual blocks. In secure regimes, finalization is consistently $O(w)$.

\subsection{Results Summary}
Across benchmarks, we find: 
All Praos comparisons below use an equalized block-arrival process per Section~\ref{sec:setup} to avoid baseline-rate bias.
\begin{itemize}
\item \textbf{Stability at 30\% adversary.} The network consistently resists optimized reorg attempts; binned histograms of reorg lengths are time-equivalent to Praos blocks.

\item \textbf{Quantitative comparison to Praos.} Under the strongest adversarial optimization tested (ILP-based, optimal sliding attack window, exhaustive local), Praos produces roughly 15-block reorgs with frequency $\approx 10^{-2}$ (Appendix~\ref{sec:eval-app}, Fig.~\ref{fig:praos}), whereas \ProjBase{} produces 14--15-block reorgs with frequency $\approx 1.5\times 10^{-5}$ ($\sim 600\times$ improvement). For 10-block reorgs, Praos is $\approx 7\times 10^{-1}$ vs. \ProjBase{} $\approx 8\times 10^{-5}$ ($\sim 8750\times$ improvement).

\item \textbf{Effect of production rate $f$.} Increasing $f$ raises short-term instability but \emph{reduces} reorg length.  A higher information rate increases local variance but speeds convergence.

\item \textbf{Stronger adversaries.} With $40$--$45\%$ adversarial control, \ProjBase{} maintains bounded reorgs and compares favorably to Praos under the same conditions.

\item \textbf{Near-majority and majority.} At $49\%$, the system degrades gracefully; at $\ge 51\%$, the attacker dominates (as expected), but the induced distributions remain informative for parameterization and risk analysis.
\end{itemize}

%% file: app2.tex


\section{Standard Concentration Bounds}
\label{sec:concentration-app}

We recall the standard concentration inequalities used in our analysis. Proofs can be found in textbooks on randomized algorithms and probabilistic methods, e.g., Mitzenmacher and Upfal~\cite{MitzenmacherUpfal} and Dubhashi and Panconesi~\cite{DubhashiPanconesi}.

\begin{lemma}[Chernoff bound]
\label{lem:chernoff}
Let $X_1,\dots,X_n$ be independent $\{0,1\}$-valued random variables and $X = \sum_{i=1}^n X_i$ with $\mu = \E[X]$. Then for any $\delta \in (0,1)$,
\[
  \Pr[X \le (1-\delta)\mu]
  \le \exp\!\Bigl(-\tfrac{\delta^2}{2}\,\mu\Bigr)
  \quad\text{and}\quad
  \Pr[X \ge (1+\delta)\mu]
  \le \exp\!\Bigl(-\tfrac{\delta^2}{3}\,\mu\Bigr).
\]
More generally, for any $\delta > 0$,
\[
  \Pr[X \ge (1+\delta)\mu]
  \le \exp\!\Bigl(-\tfrac{\delta^2}{2+\delta}\,\mu\Bigr).
\]
\end{lemma}

\begin{lemma}[Azuma--Hoeffding inequality]
\label{lem:azuma}
Let $(M_t)_{t=0}^T$ be a martingale with respect to a filtration $(\mathcal{F}_t)$, and assume that the differences are uniformly bounded: 
for some $c>0$,
\[
  |M_t - M_{t-1}| \le c \quad\text{almost surely for all } t=1,\dots,T.
\]
Then for any $\lambda > 0$,
\[
  \Pr\bigl[\,|M_T - M_0| \ge \lambda\,\bigr]
  \;\le\;
  2 \exp\!\Bigl(-\frac{\lambda^2}{2 c^2 T}\Bigr).
\]
\end{lemma}

%% file: app4.tex


\section{Exact Antichain Computation via Dilworth's Theorem}
\label{sec:appendix-dilworth}
\label{sec:base-app}

In this appendix we recall a standard polynomial-time algorithm for computing the \emph{exact maximum-cardinality} antichain of the window-DAG \(G_{t,w}\), which is the reference-selection method used by honest validators in Sections~\ref{sec:ideal} and~\ref{sec:base}. The algorithm is based on Dilworth's theorem:

\begin{theorem}[Dilworth, 1950]
In any finite partially ordered set $(P,\preceq)$, the size of the largest antichain equals the minimum number of chains in a partition of $P$.
\end{theorem}

This equivalence enables a polynomial-time algorithm: by constructing a bipartite graph representing comparabilities in the DAG's transitive reduction, the maximum antichain problem reduces to a maximum matching problem. Solving this yields the exact maximum antichain, which honest validators may use for reference selection.

\paragraph{Transitive Reduction.}
Since redundant edges do not affect comparability, we may safely compute the \emph{transitive reduction} of $G_{s,w}$ before applying Dilworth’s algorithm. This reduces the input size and improves runtime.

\paragraph{Complexity.}
Computing the exact maximum antichain requires solving a maximum matching instance with complexity roughly $O(|V|^{2.5})$ using Hopcroft–Karp. While practical for small window sizes, this is less efficient than greedy heuristics for large $w$.

\paragraph{Summary.}
The Dilworth-based algorithm described above is the method used by all honest validators in \ProjIdeal and \ProjBase: honest parties \emph{always compute the exact maximum antichain} within the windowed DAG.  This ensures that reference selection maximizes fan-in and matches the protocol specification in Section~\ref{sec:ideal}.

For experimental evaluation in Section~\ref{sec:eval}, we optionally include a \emph{greedy maximal-antichain heuristic} as a baseline to study the
impact of suboptimal antichain selection.  It does not affect any safety proof or liveness argument.

Thus, in the actual protocol, the antichain used for block references is \emph{always} the exact maximum one computed via the Dilworth-based matching reduction, while the greedy construction appears solely as an experimental tool.

%% file: app6.tex


\section{Additional Experimental Plots in Section~\ref{sec:eval}}
\label{sec:eval-app}

To improve readability, Figures~\ref{fig:praos}, \ref{fig:adversarial},\ref{fig:c2-delay}, and \ref{fig:window}, originally shown in \S\ref{sec:eval} have been relocated to this appendix section.  They reproduce the underlying experimental plots with identical settings and labeling, and should be read as the detailed visual companion to the results summarized in \S\ref{sec:eval}.

\begin{figure}[htbp!]
\includegraphics[width=0.7\textwidth]{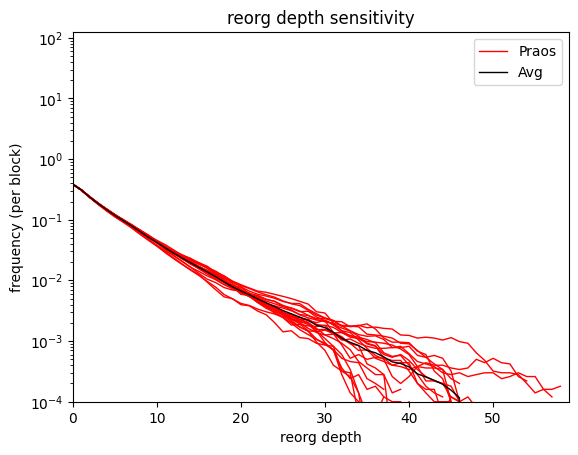}
\caption{Reorg depth sensitivity of the Ouroboros Praos with $f=0.05$. Red lines present individual chain reorg depths, while black presents an average.}
\label{fig:praos}
\end{figure}

\begin{figure}[htbp!]
\subfloat[\ProjBase, $\alpha=0.30$]{\includegraphics[width=\textwidth/2]{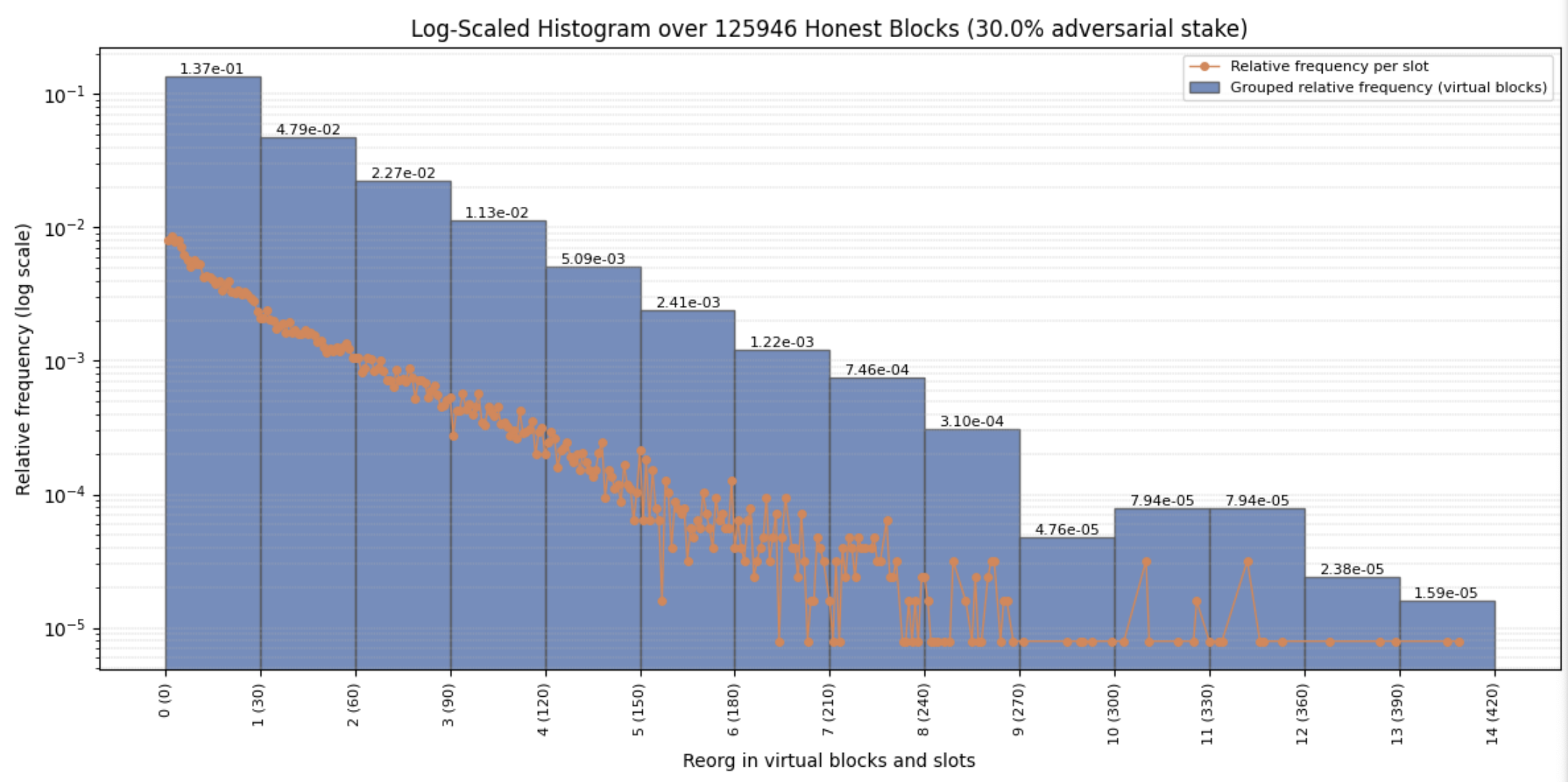}}
\subfloat[Praos, $\alpha=0.30$]{\includegraphics[height=41mm, keepaspectratio]{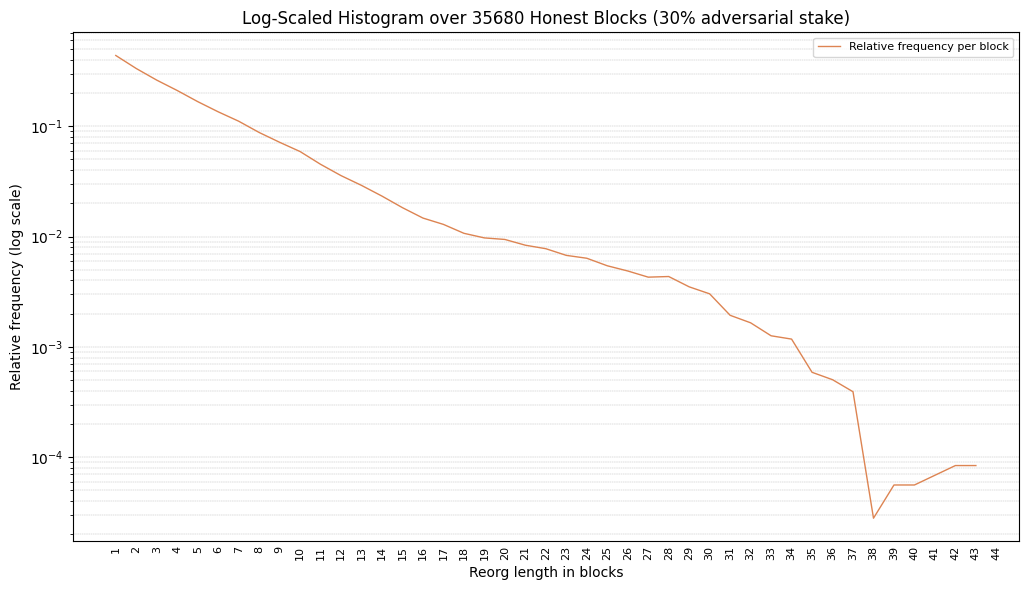}}\\
\subfloat[\ProjBase, $\alpha=0.40$]{\includegraphics[width=\textwidth/2]{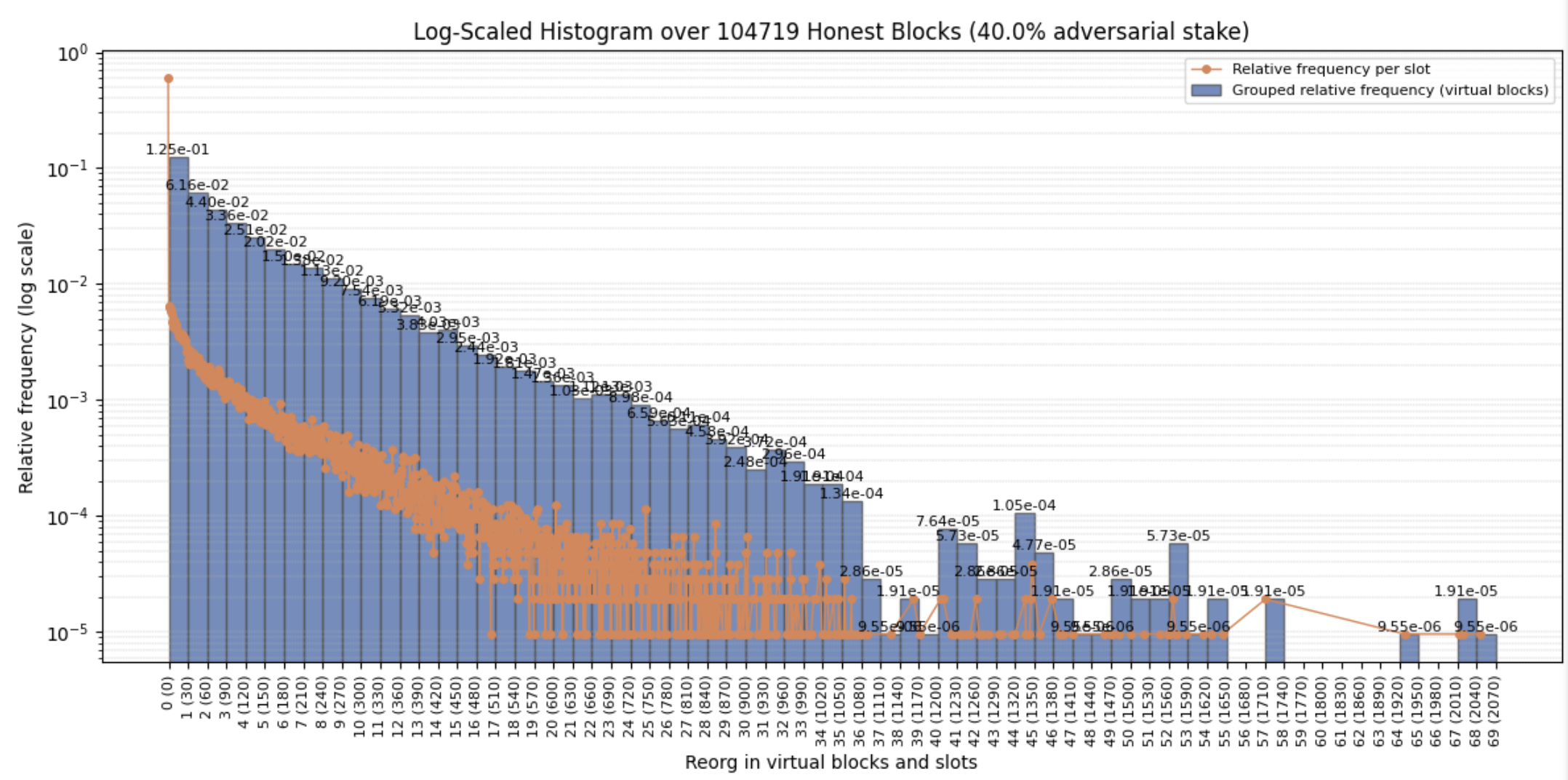}}
\subfloat[Praos, $\alpha=0.40$]{\includegraphics[height=41mm, keepaspectratio]{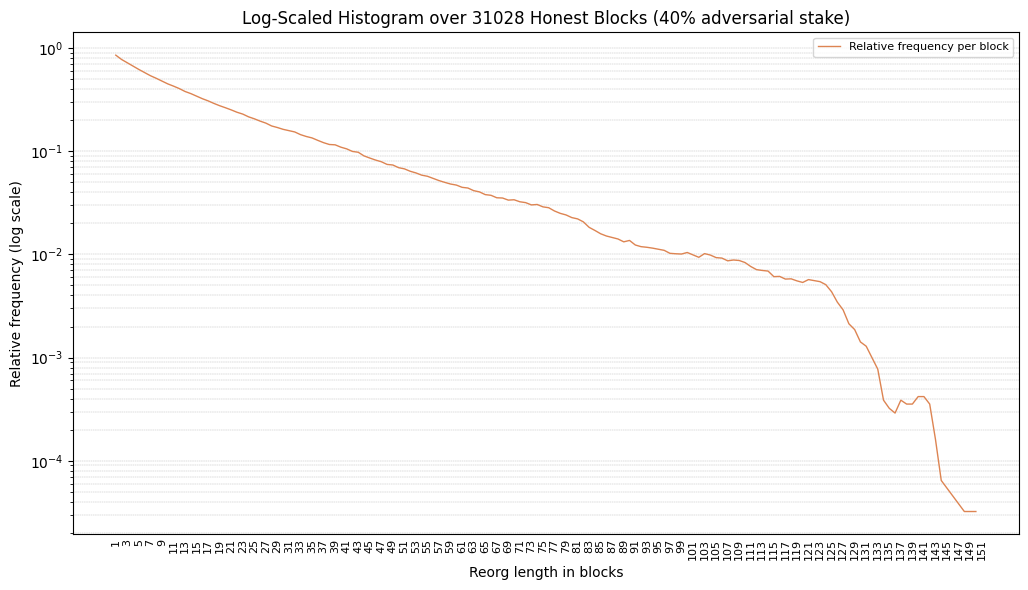}}\\
\subfloat[\ProjBase, $\alpha=0.45$]{\includegraphics[width=\textwidth/2]{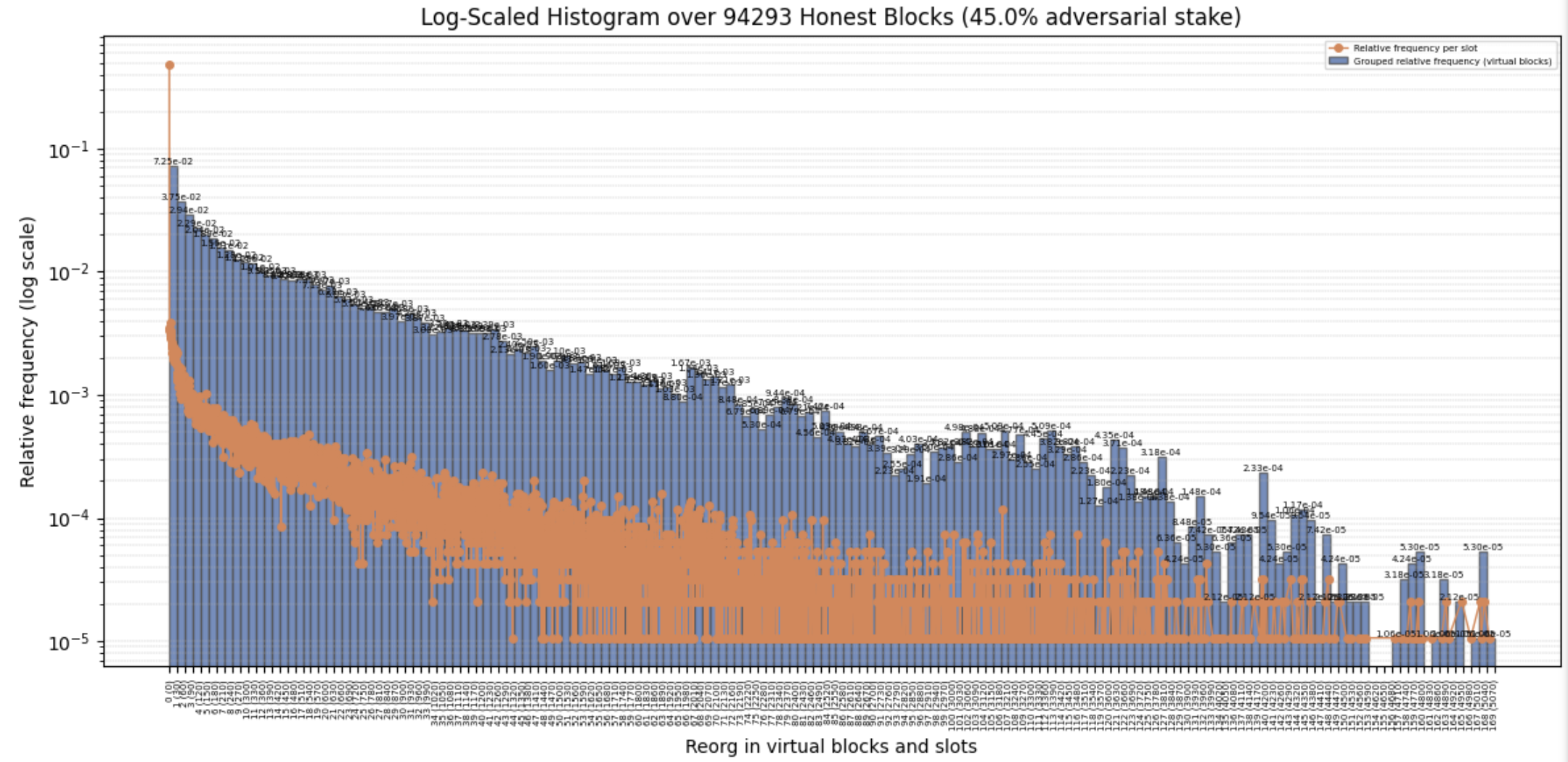}}
\subfloat[Praos, $\alpha=0.45$]{\includegraphics[height=41mm, keepaspectratio]{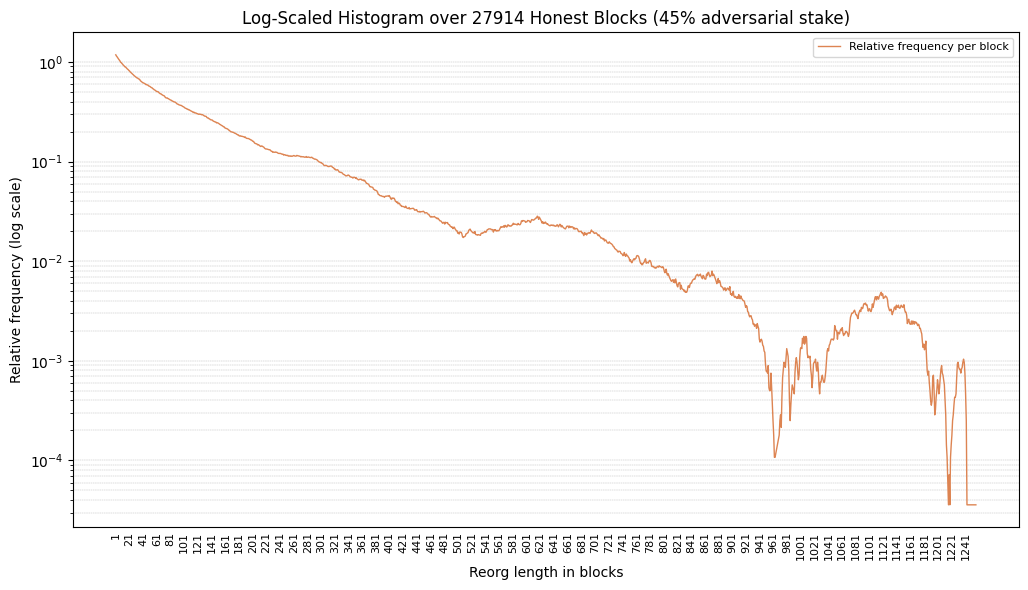}}\\
\subfloat[\ProjBase, $\alpha=0.49$]{\includegraphics[width=\textwidth/2]{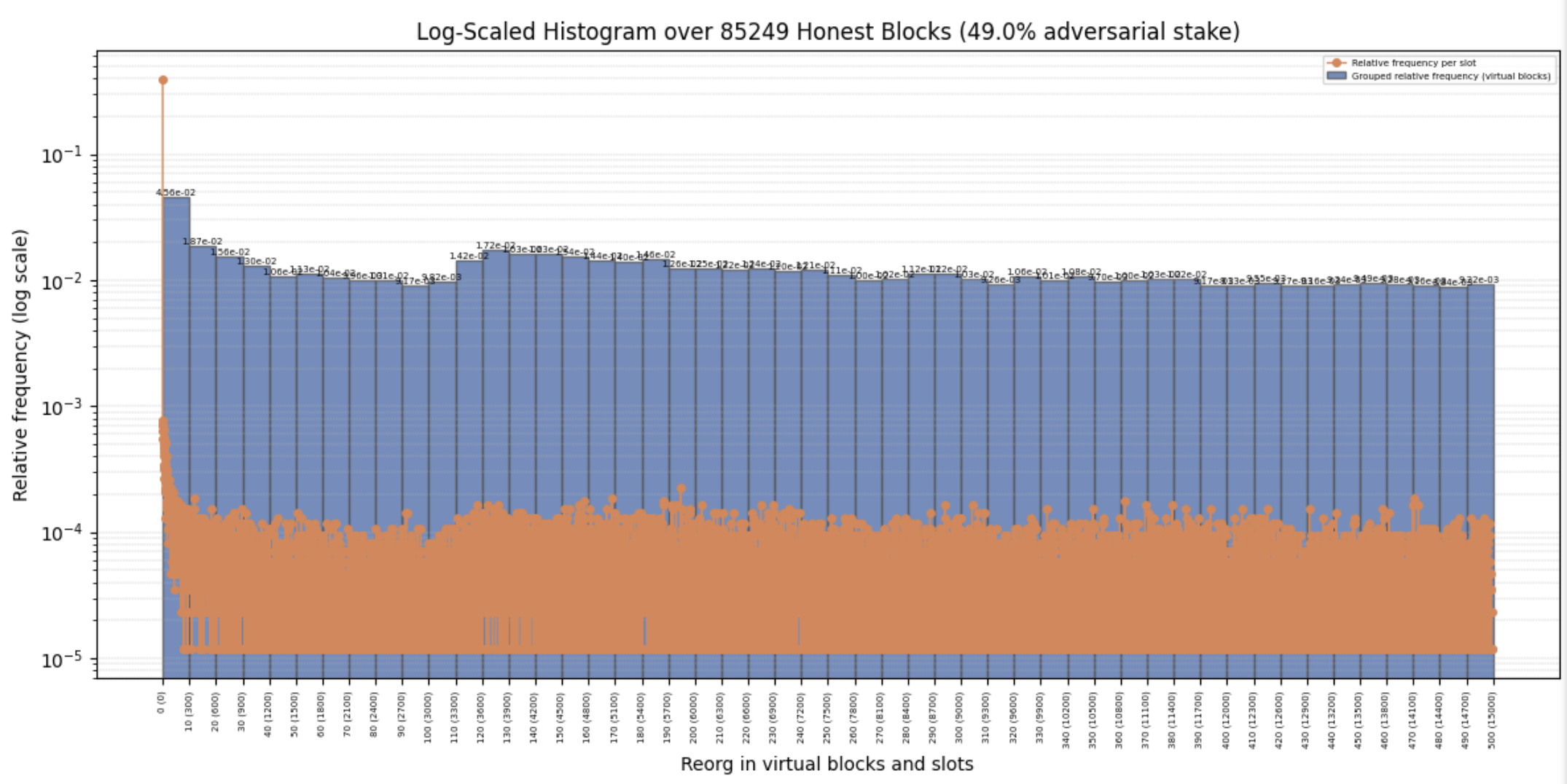}}
\subfloat[Praos, $\alpha=0.49$]{\includegraphics[height=41mm, keepaspectratio]{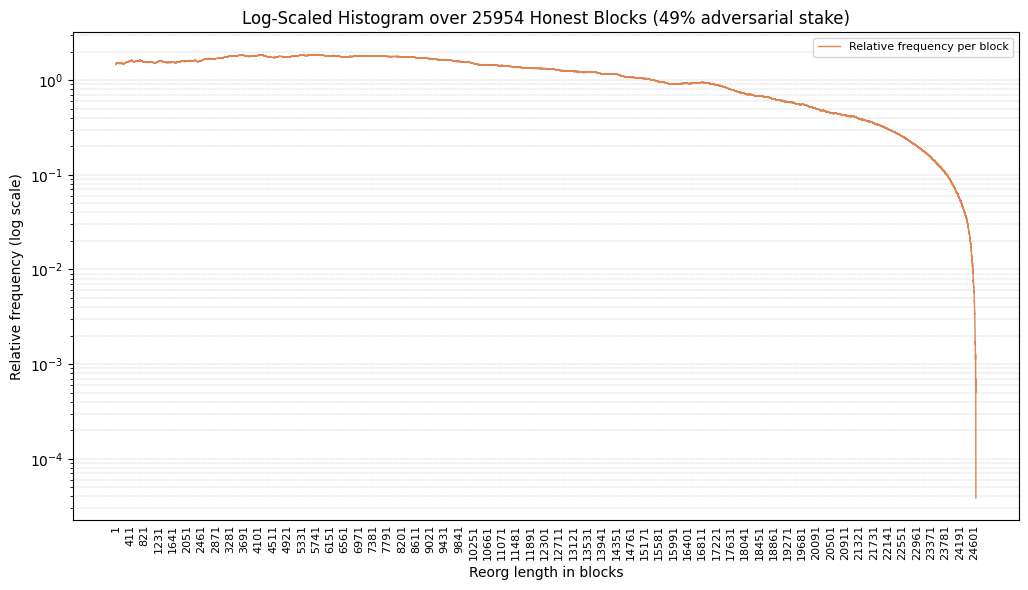}}\\
\caption{Comparison of the impact of the adversarial stake on the reorg depth for \ProjBase ($f=0.25$) and Praos ($f=0.05$).}
\label{fig:adversarial}
\end{figure}

\begin{figure}[htbp!]
\subfloat[$f=0.05$ and $\Delta=2$s]{\includegraphics[width=\textwidth/4]{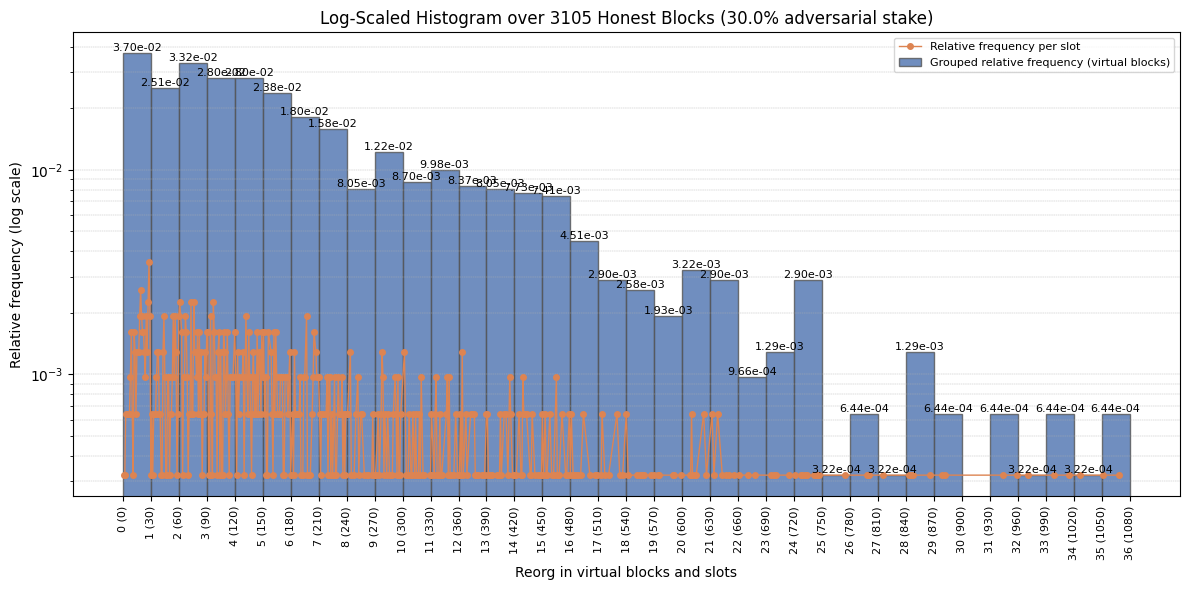}}
\subfloat[$f=0.1$ and $\Delta=2$s]{\includegraphics[width=\textwidth/4]{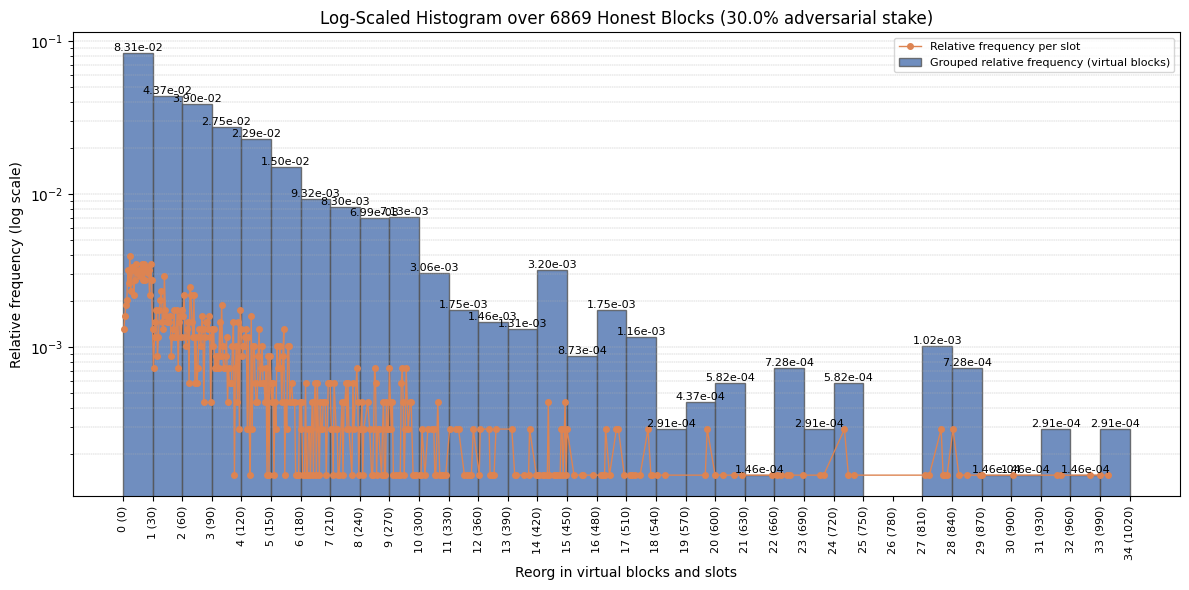}}
\subfloat[$f=0.2$ and $\Delta=2$s]{\includegraphics[width=\textwidth/4]{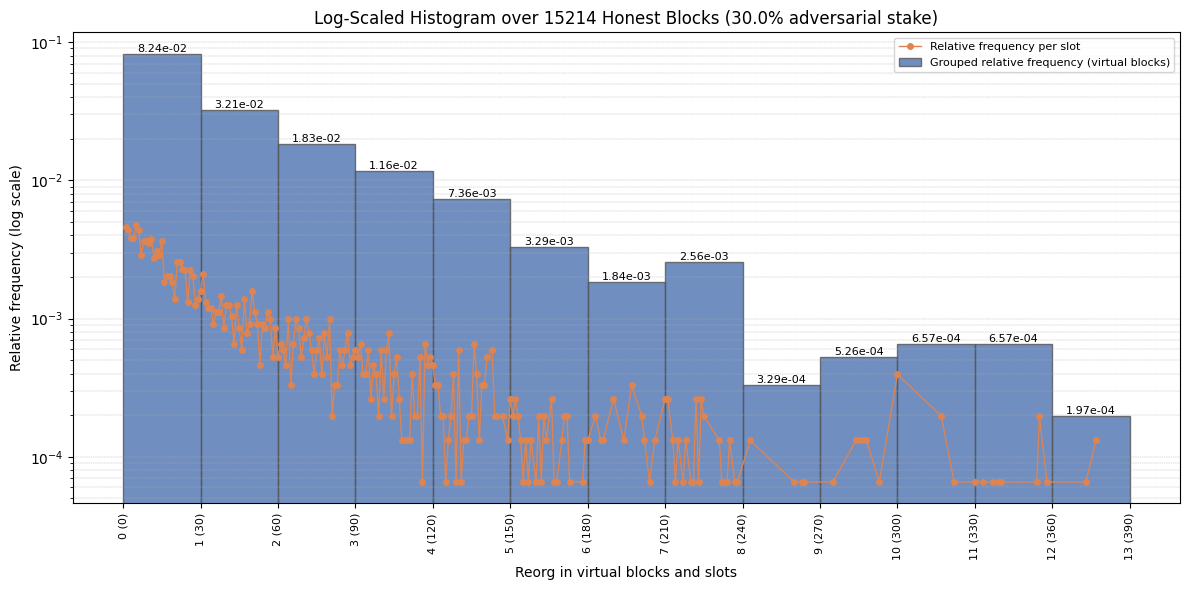}}
\subfloat[Praos and $\Delta=2$s]{\includegraphics[height=21mm, keepaspectratio]{figs2/reorg-praos-f0.05-as30-d2-1M-OK}}\\
\subfloat[$f=0.05$ and $\Delta=5$s]{\includegraphics[width=\textwidth/4]{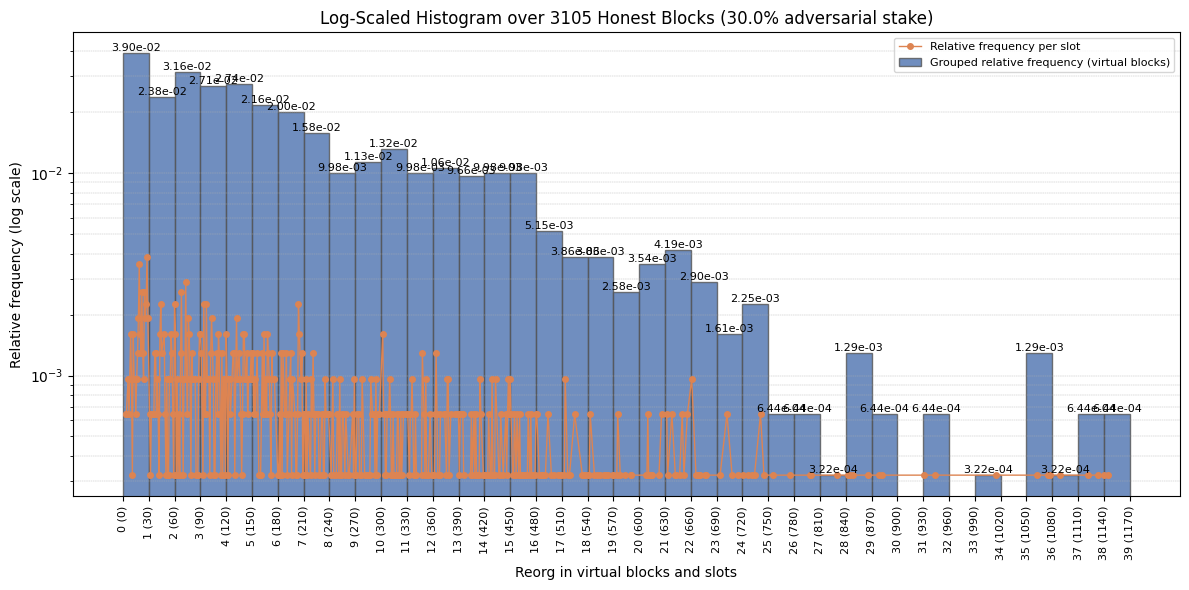}}
\subfloat[$f=0.1$ and $\Delta=5$s]{\includegraphics[width=\textwidth/4]{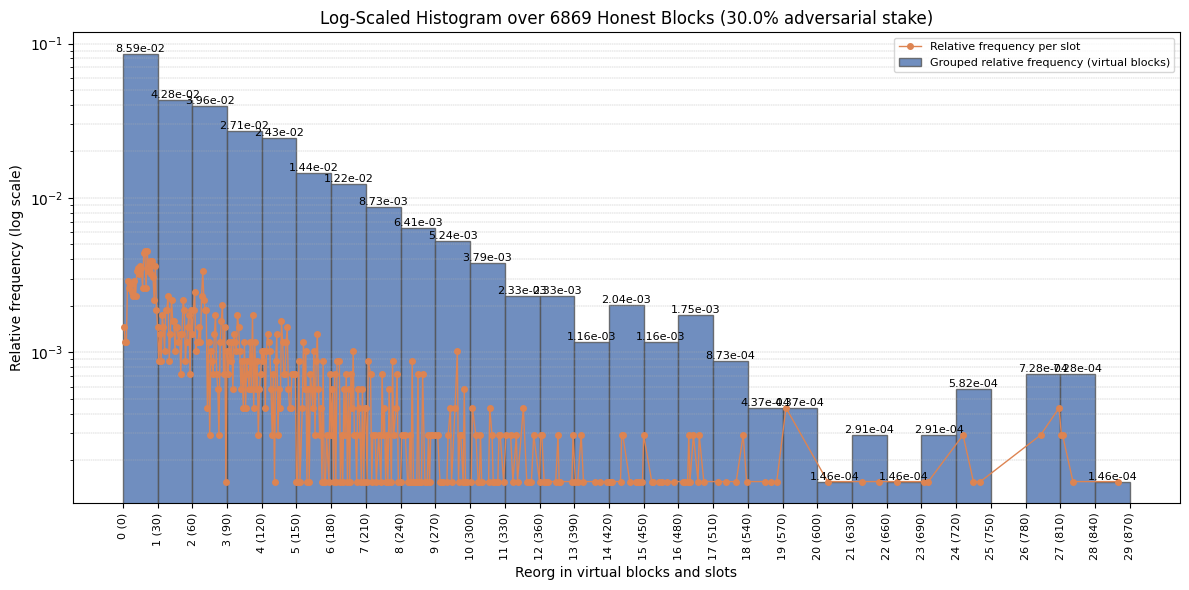}}
\subfloat[$f=0.2$ and $\Delta=5$s]{\includegraphics[width=\textwidth/4]{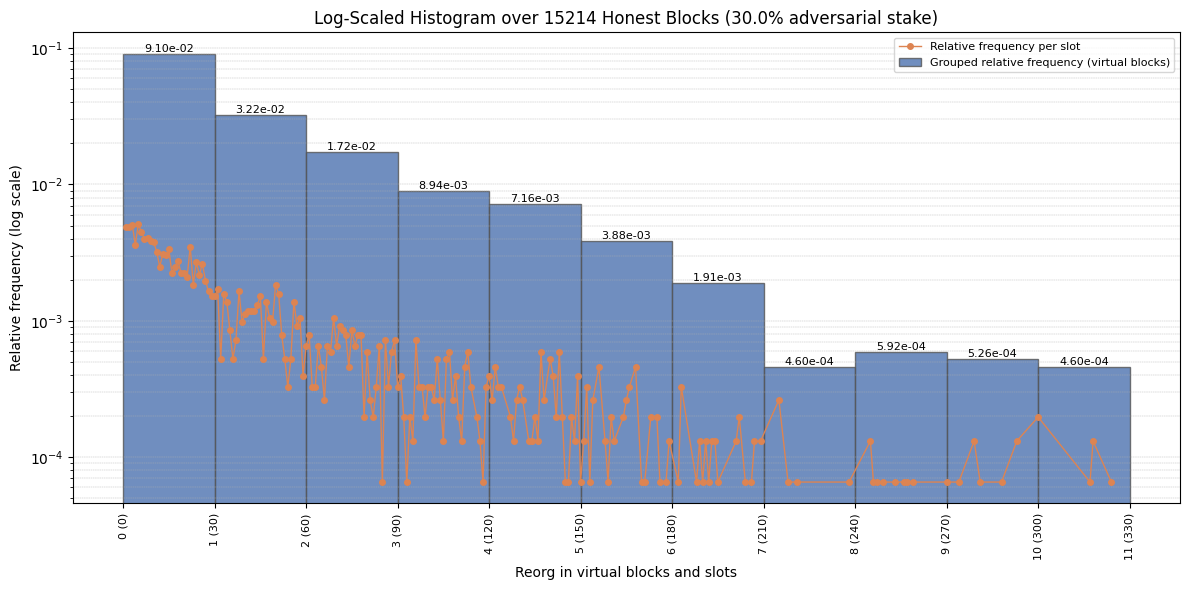}}
\subfloat[Praos and $\Delta=5$s]{\includegraphics[height=21mm, keepaspectratio]{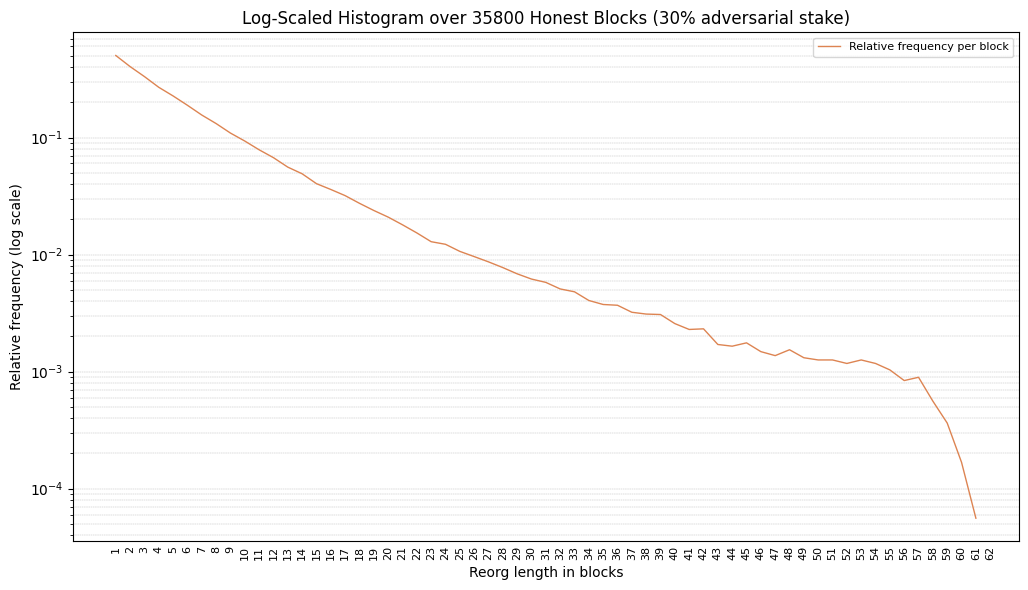}}\\
\subfloat[$f=0.05$ and $\Delta=8$s]{\includegraphics[width=\textwidth/4]{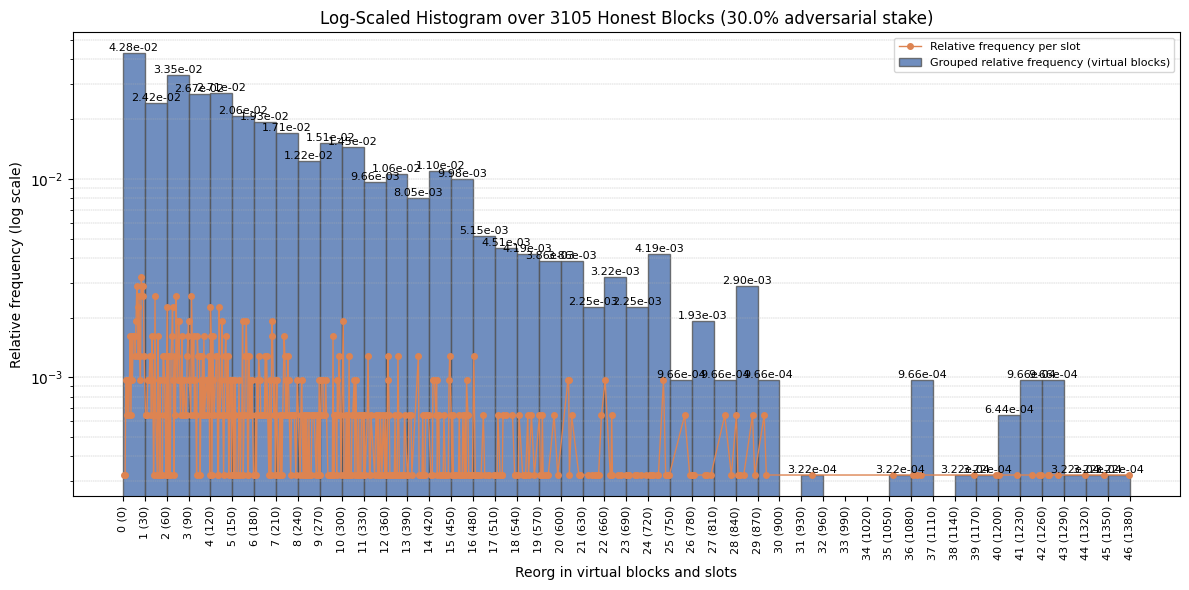}}
\subfloat[$f=0.1$ and $\Delta=8$s]{\includegraphics[width=\textwidth/4]{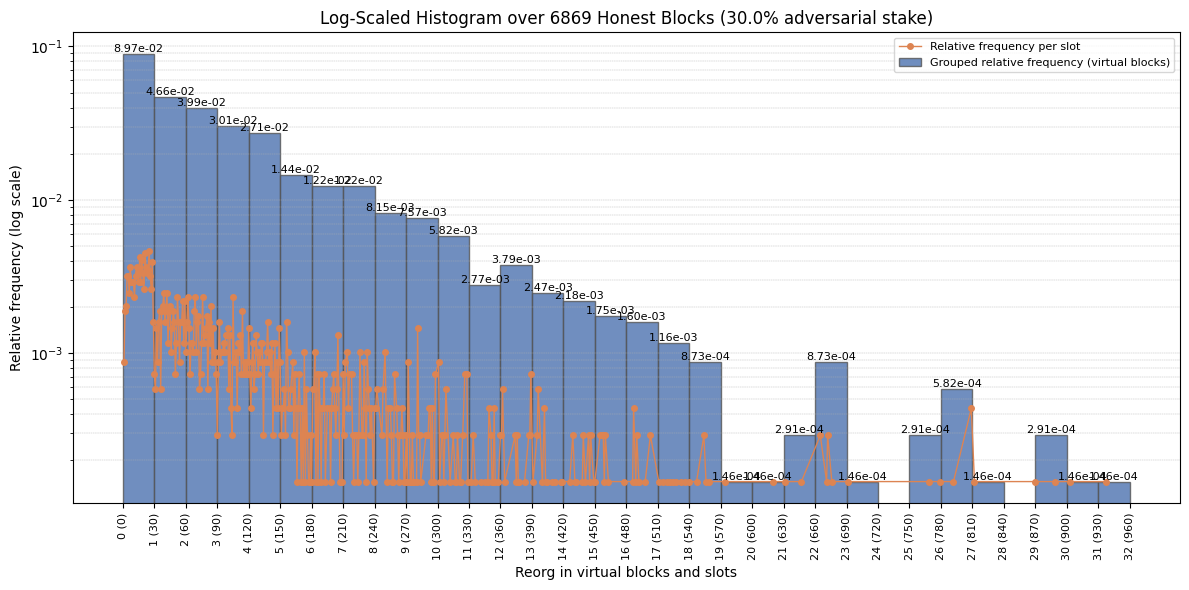}}
\subfloat[$f=0.2$ and $\Delta=8$s]{\includegraphics[width=\textwidth/4]{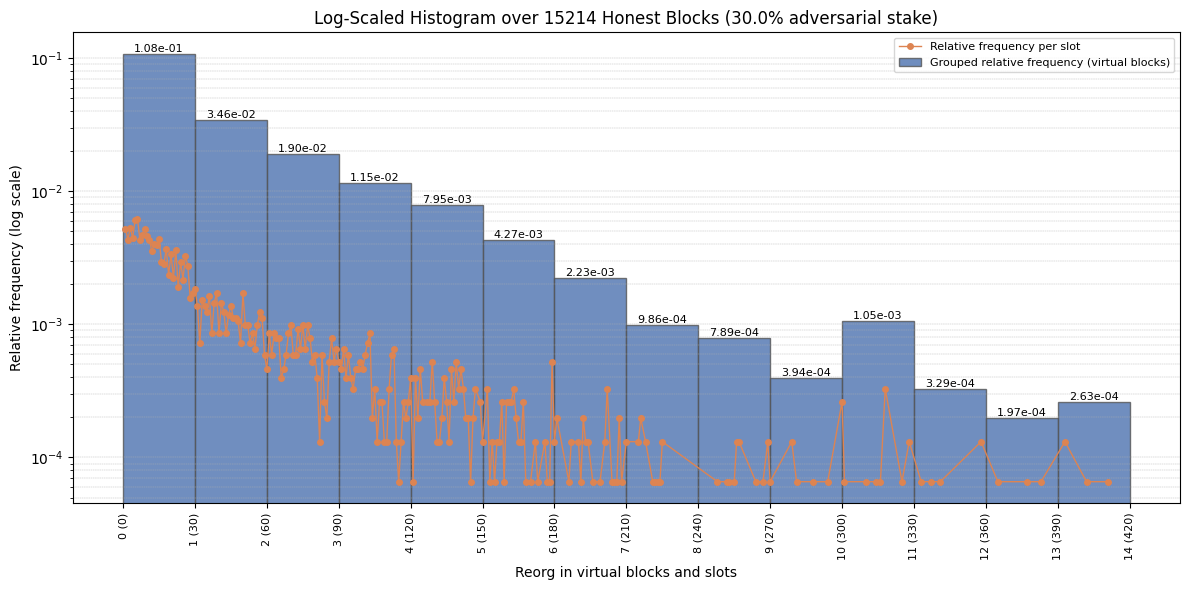}}
\subfloat[Praos and $\Delta=8$s]{\includegraphics[height=21mm, keepaspectratio]{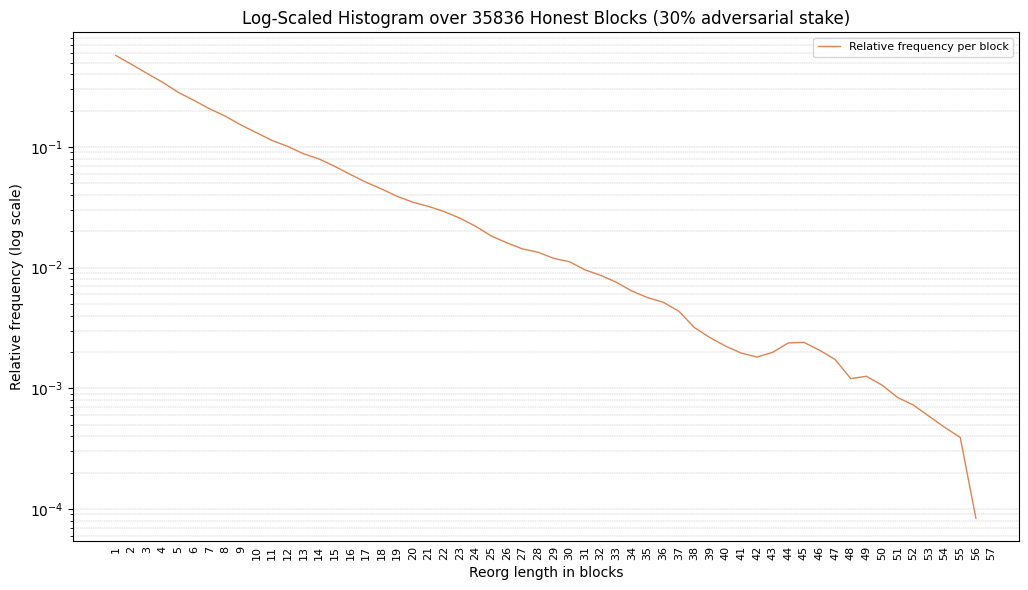}}\\
\subfloat[$f=0.05$ and $\Delta=10$s]{\includegraphics[width=\textwidth/4]{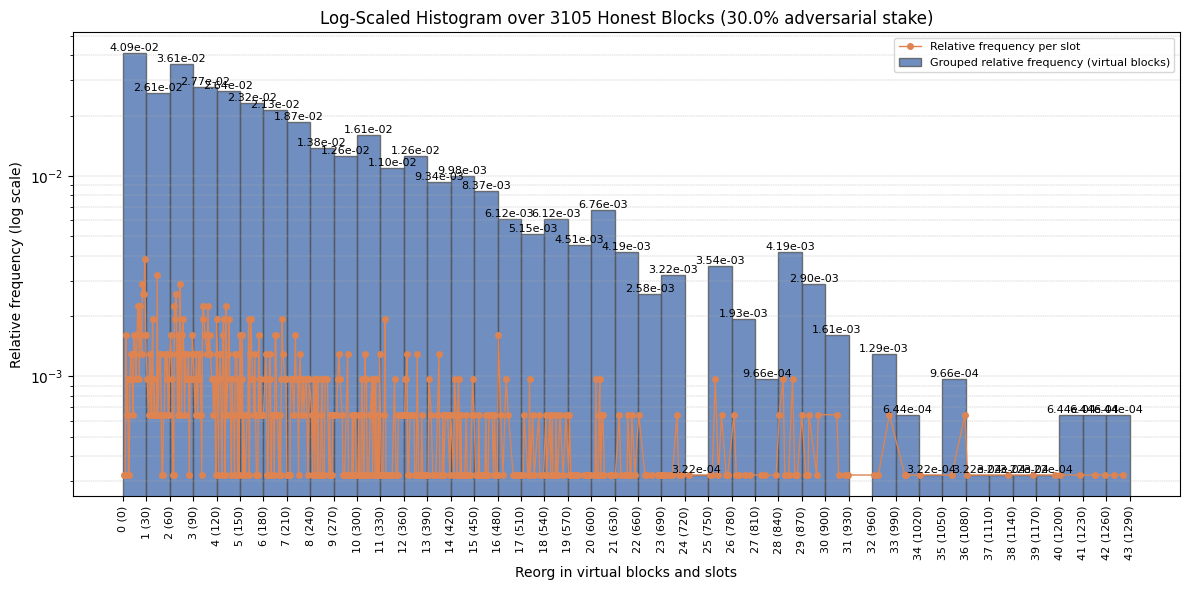}}
\subfloat[$f=0.1$ and $\Delta=10$s]{\includegraphics[width=\textwidth/4]{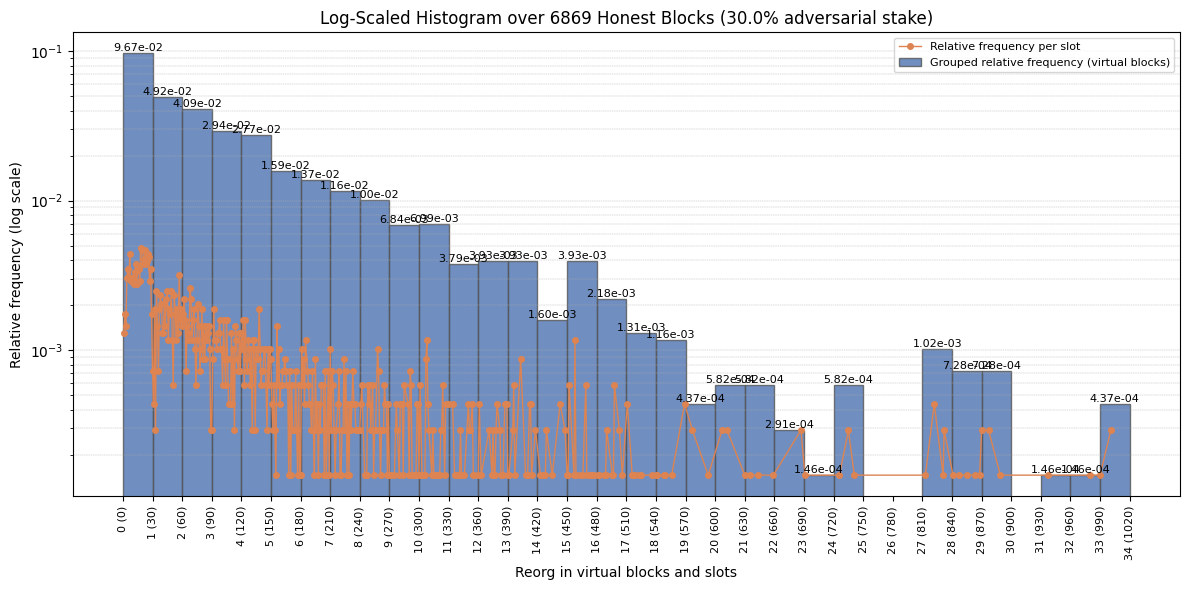}}
\subfloat[$f=0.2$ and $\Delta=10$s]{\includegraphics[width=\textwidth/4]{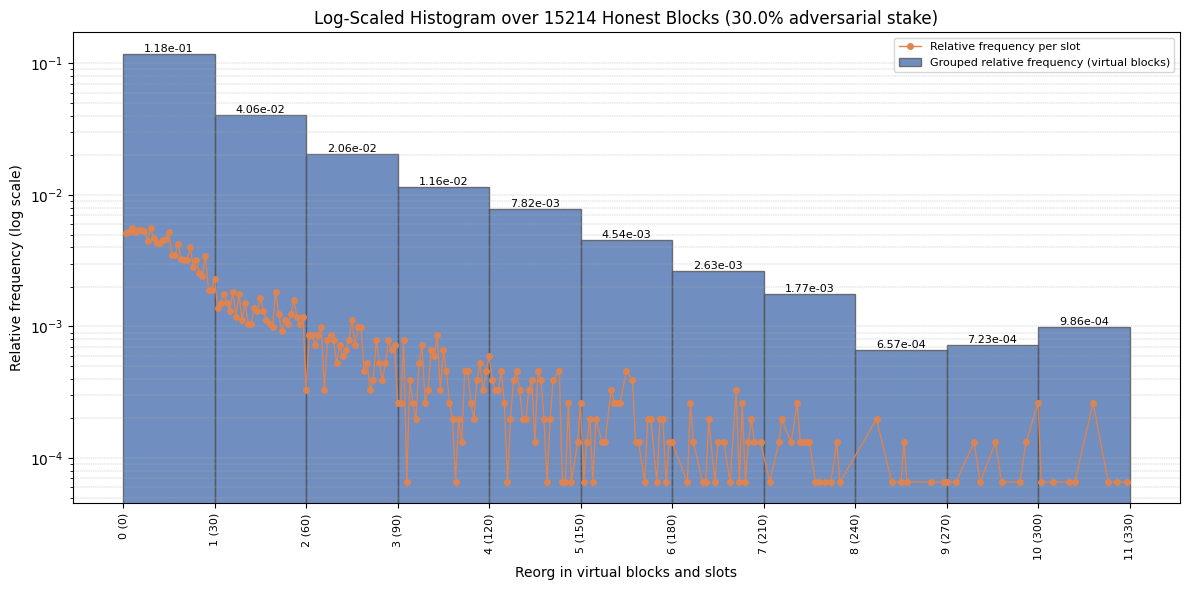}}
\subfloat[Praos and $\Delta=10$s]{\includegraphics[height=21mm, keepaspectratio]{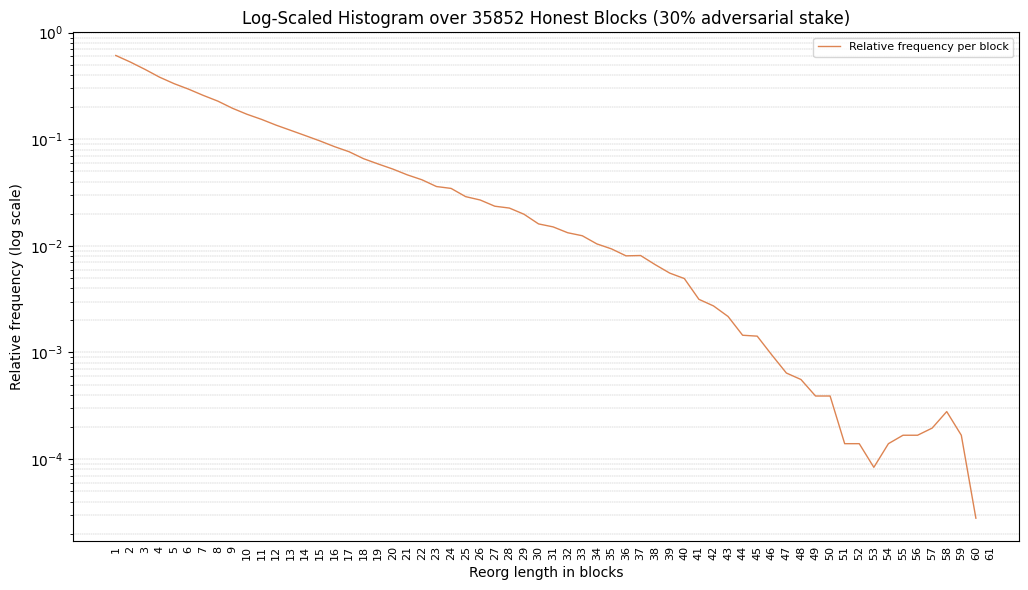}}\\
\subfloat[$f=0.05$ and $\Delta=12$s]{\includegraphics[width=\textwidth/4]{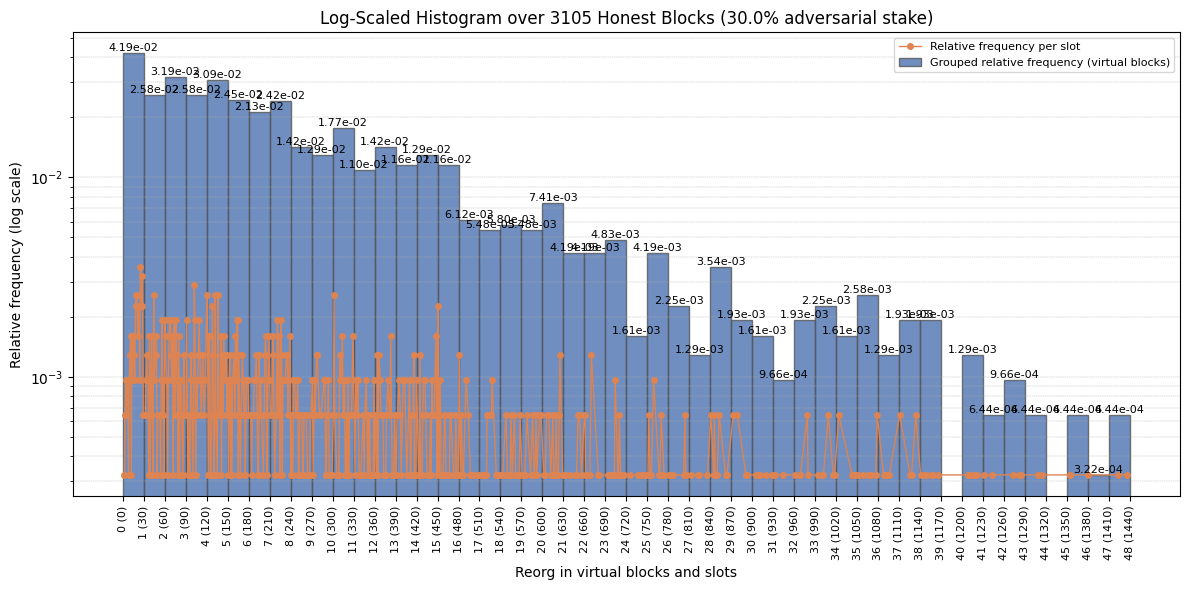}}
\subfloat[$f=0.1$ and $\Delta=12$s]{\includegraphics[width=\textwidth/4]{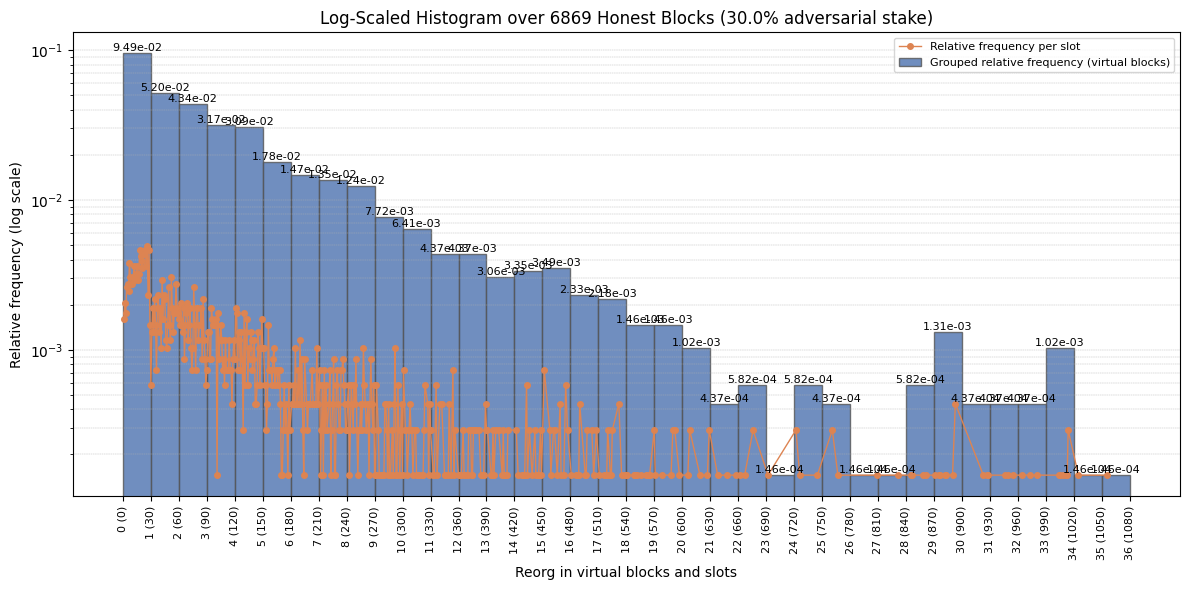}}
\subfloat[$f=0.2$ and $\Delta=12$s]{\includegraphics[width=\textwidth/4]{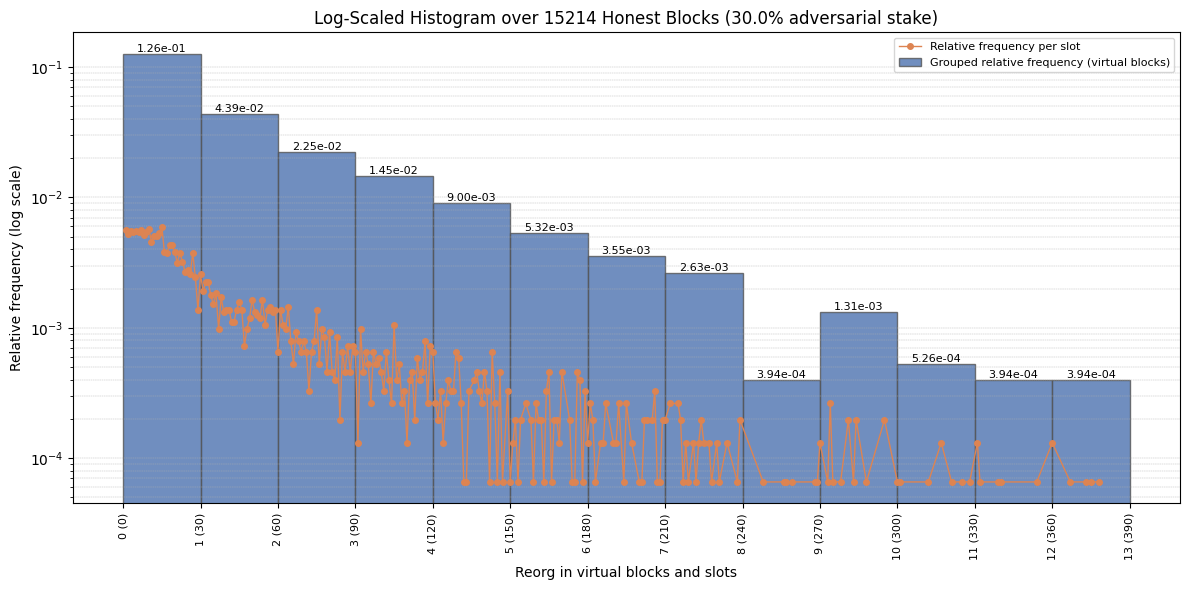}}
\subfloat[Praos and $\Delta=12$s]{\includegraphics[height=21mm, keepaspectratio]{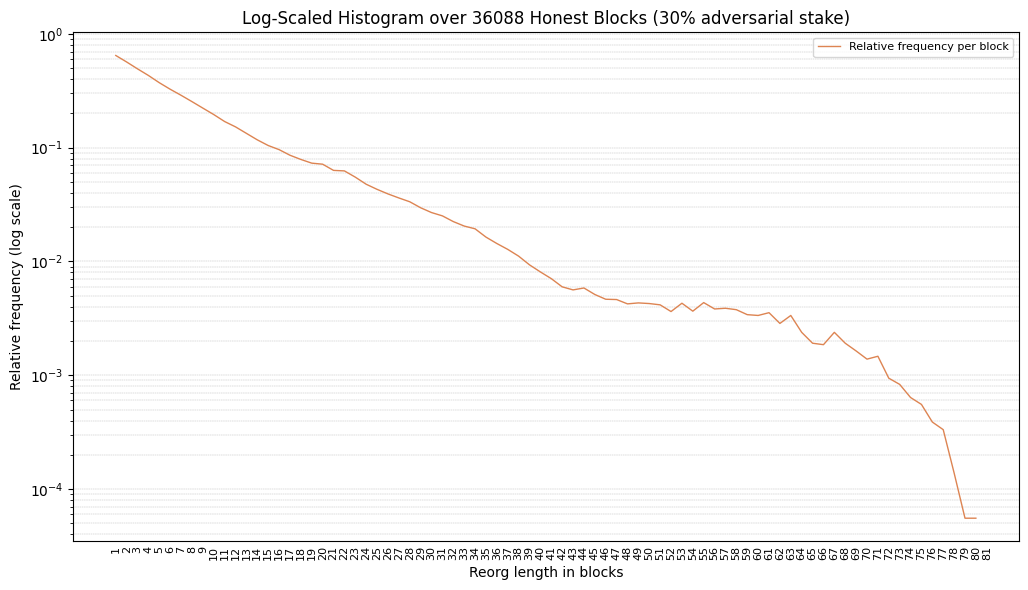}}\\
\subfloat[$f=0.05$ and $\Delta=15$s]{\includegraphics[width=\textwidth/4]{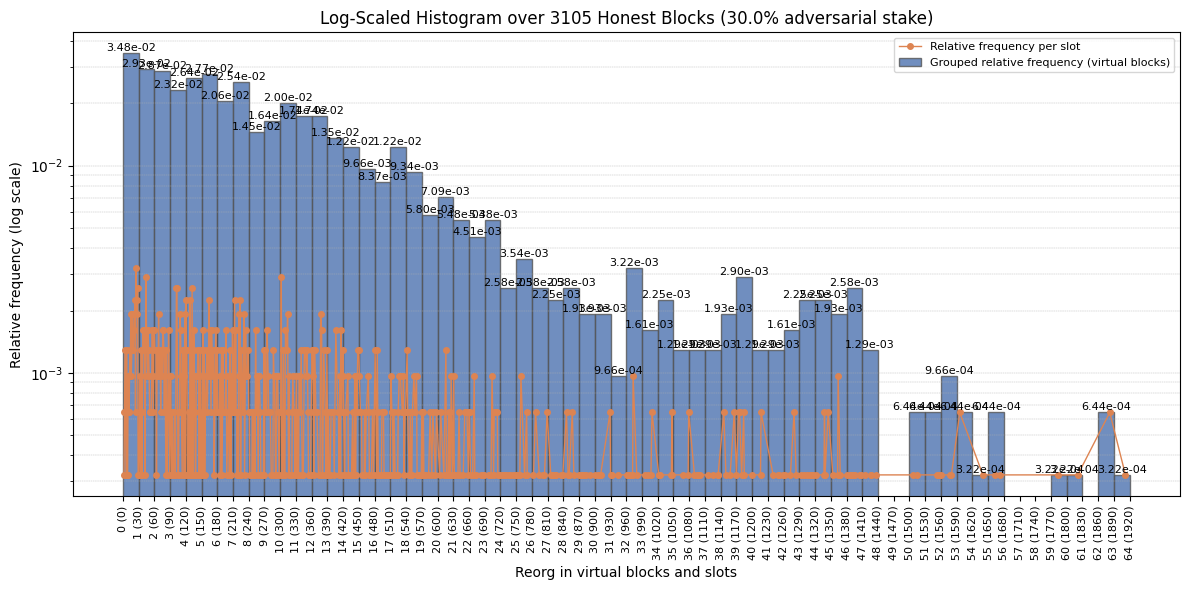}}
\subfloat[$f=0.1$ and $\Delta=15$s]{\includegraphics[width=\textwidth/4]{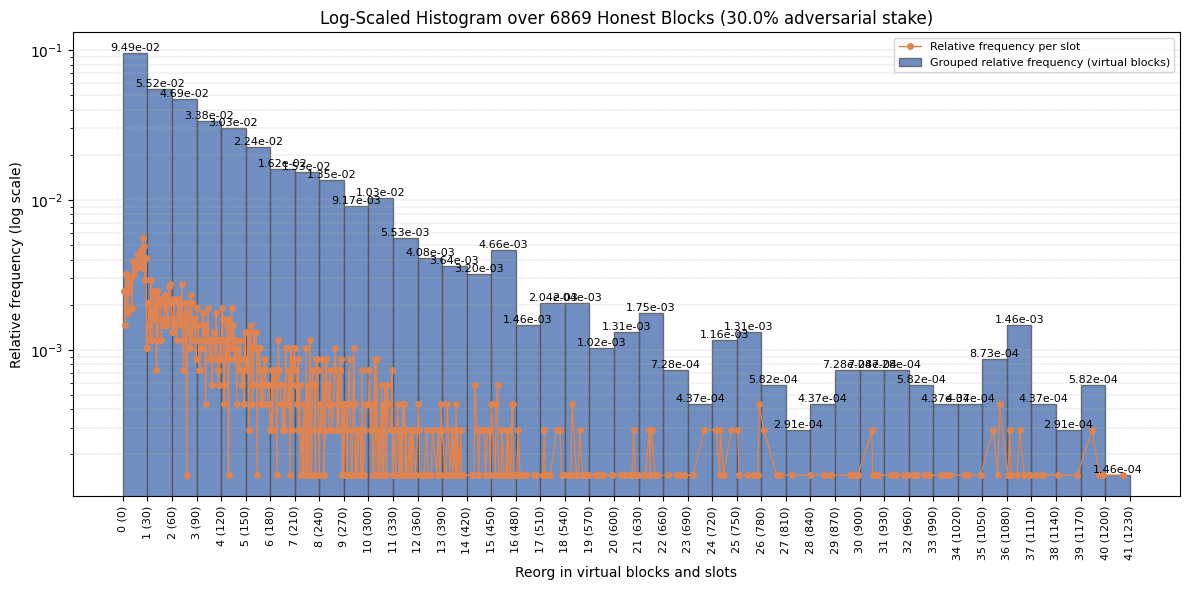}}
\subfloat[$f=0.2$ and $\Delta=15$s]{\includegraphics[width=\textwidth/4]{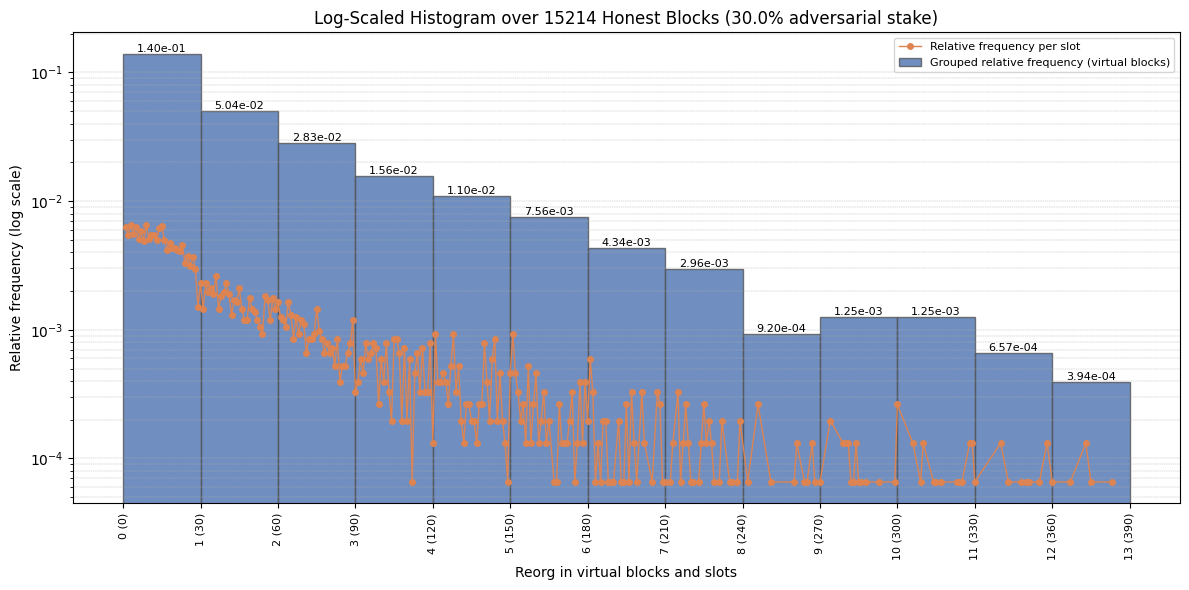}}
\subfloat[Praos and $\Delta=15$s]{\includegraphics[height=21mm, keepaspectratio]{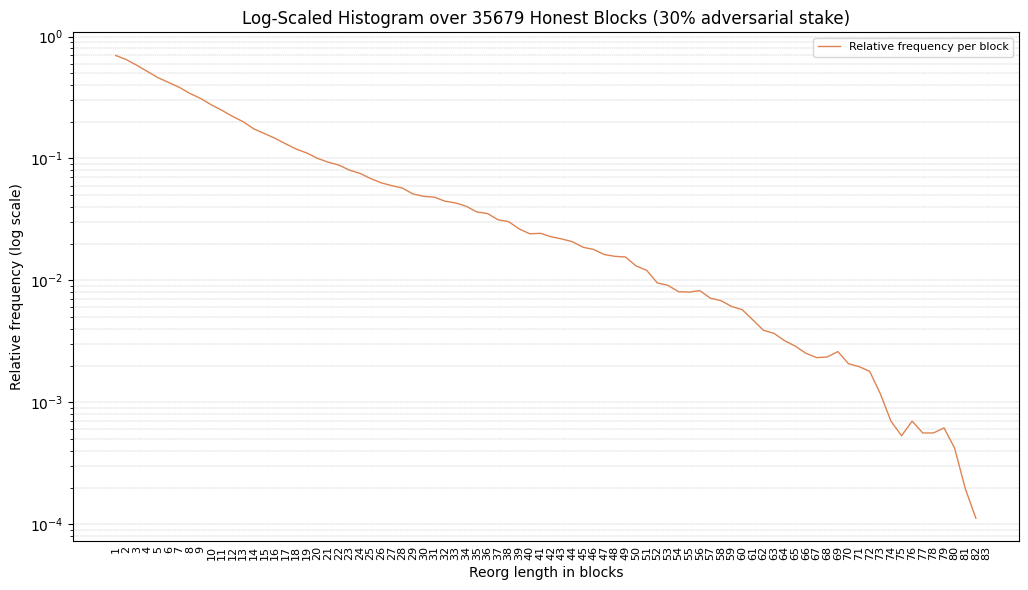}}\\
\caption{Comparison of the impact of block production rate and mean broadcast delay on the reorg depth for \ProjBase with 30s window and 30\% adversarial stake to Praos. As the mean broadcast delay increases, the reorg depth also increases. Conversely, when the block production rate increases, the amplitude of the reorg depth increase diminishes. In essence, with higher block production, the reorg depth becomes less responsive to the delay. }
\label{fig:c2-delay}
\end{figure}

\begin{figure}[htbp!]
\subfloat[15 seconds mean broadcast delay]{\includegraphics[width=\textwidth/3]{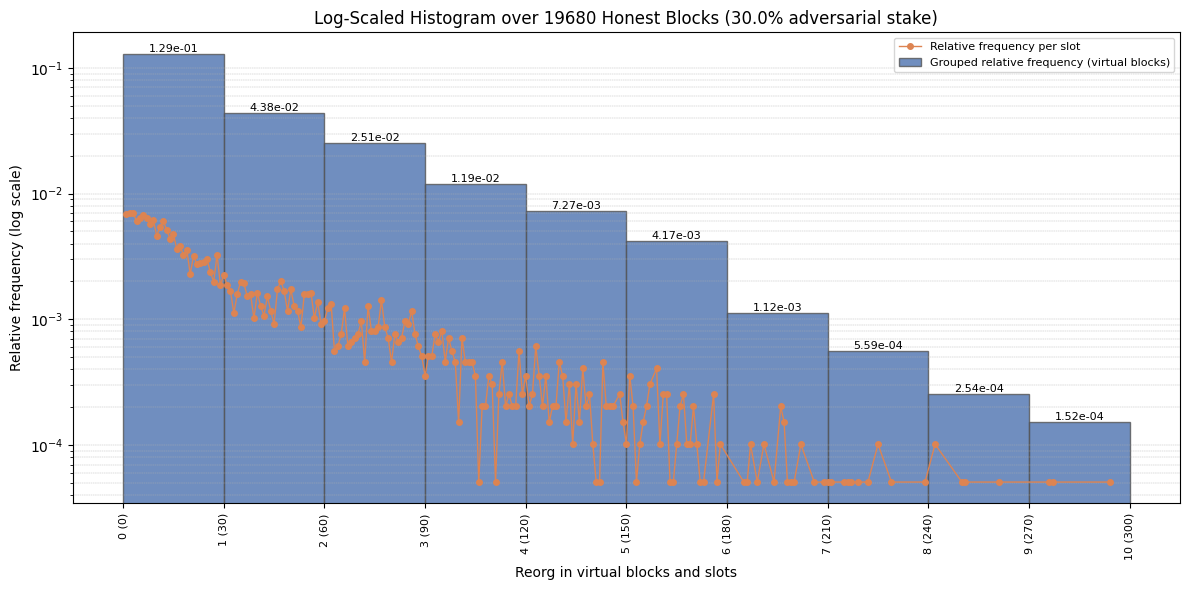}}
\subfloat[25 seconds mean broadcast delay]{\includegraphics[width=\textwidth/3]{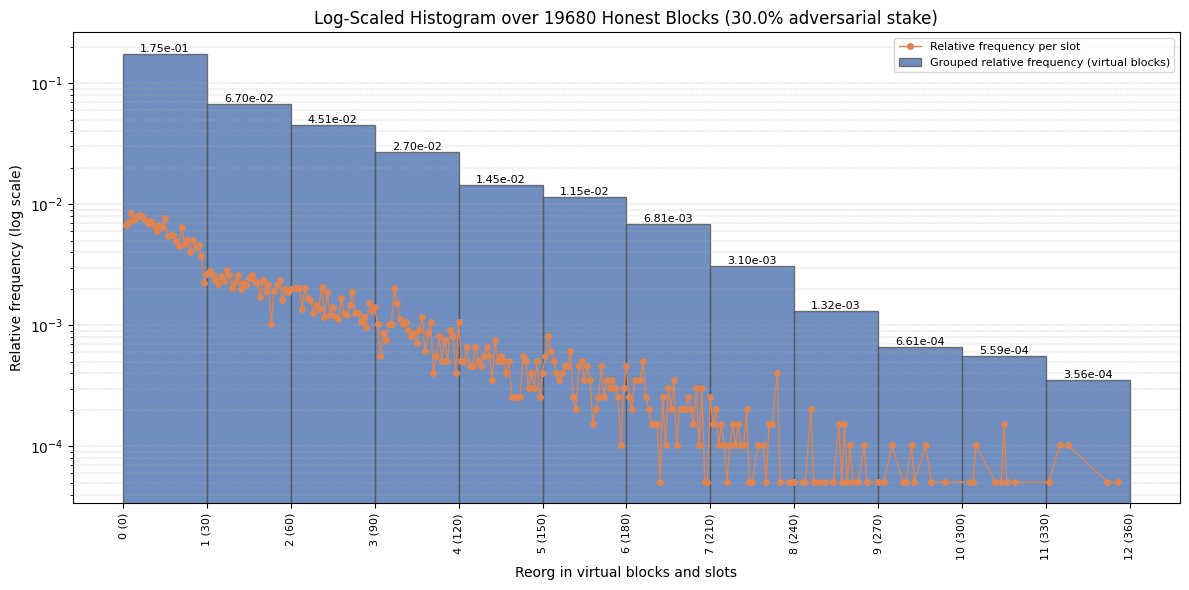}}
\subfloat[30 seconds mean broadcast delay]{\includegraphics[width=\textwidth/3]{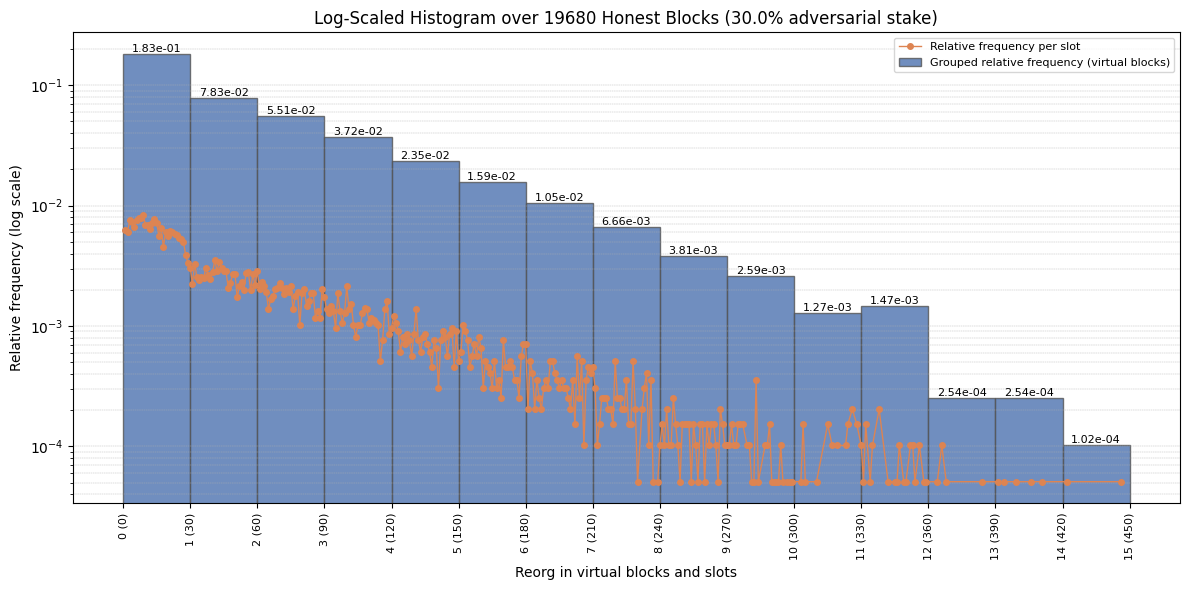}}\\
\subfloat[35 seconds mean broadcast delay]{\includegraphics[width=\textwidth/3]{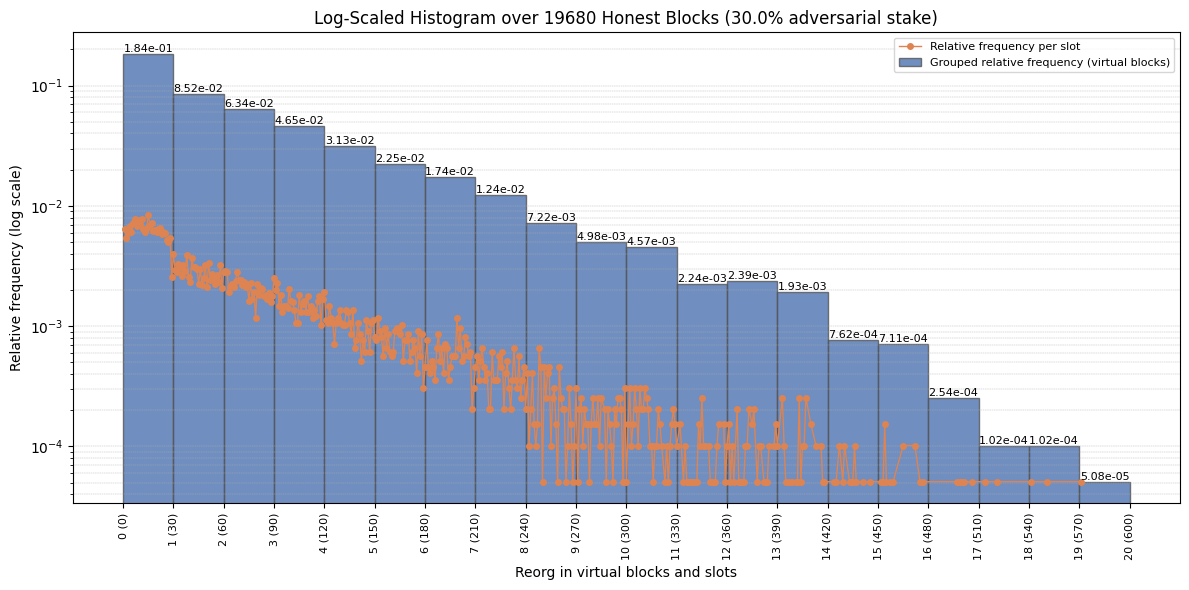}}
\subfloat[40 seconds mean broadcast delay]{\includegraphics[width=\textwidth/3]{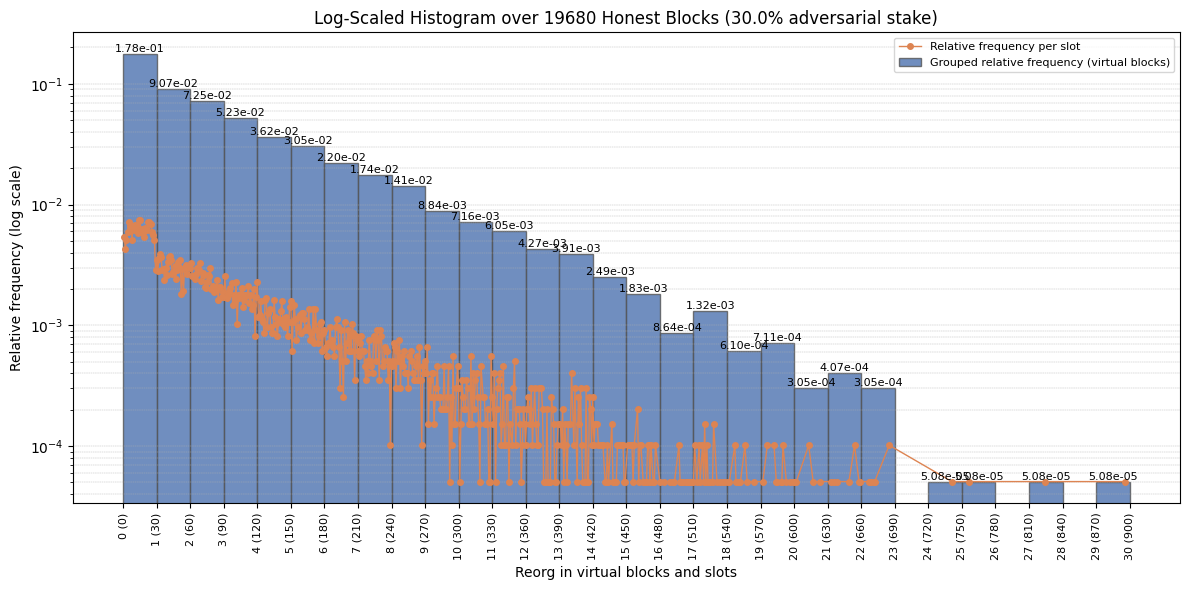}}
\subfloat[50 seconds mean broadcast delay]{\includegraphics[width=\textwidth/3]{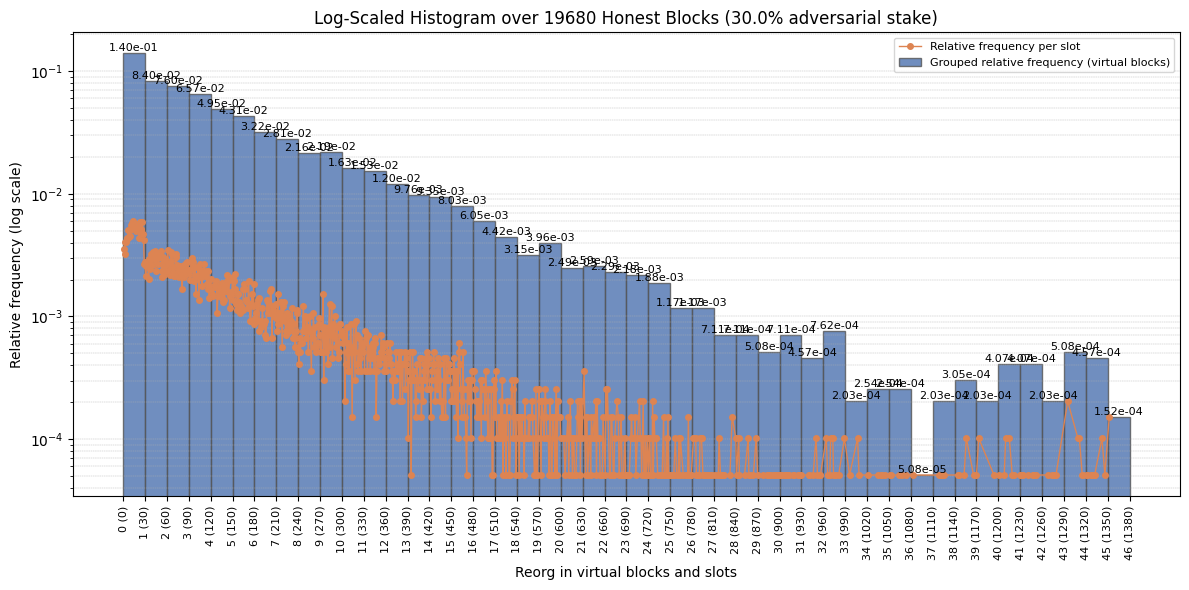}}\\
\subfloat[55 seconds mean broadcast delay]{\includegraphics[width=\textwidth/3]{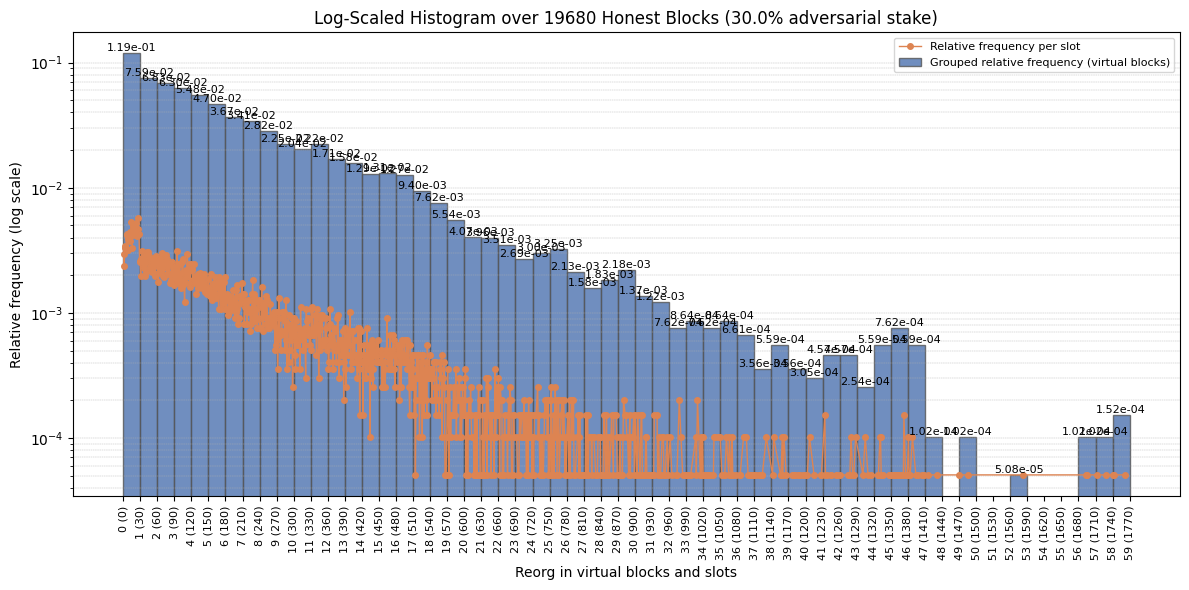}}
\subfloat[60 seconds mean broadcast delay]{\includegraphics[width=\textwidth/3]{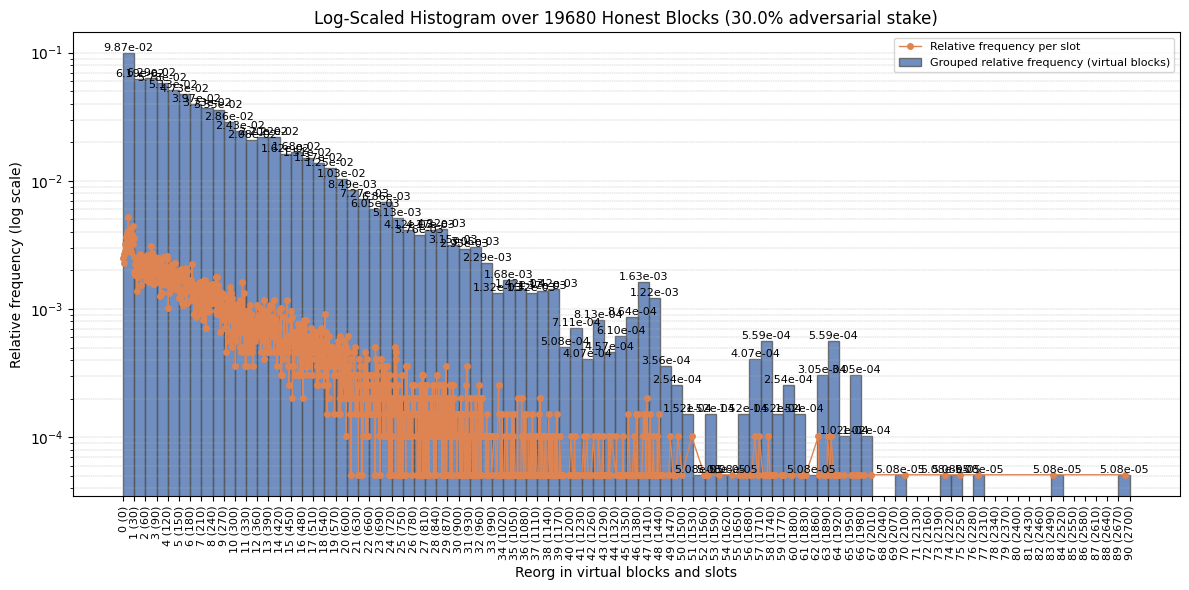}}
\subfloat[65 seconds mean broadcast delay]{\includegraphics[width=\textwidth/3]{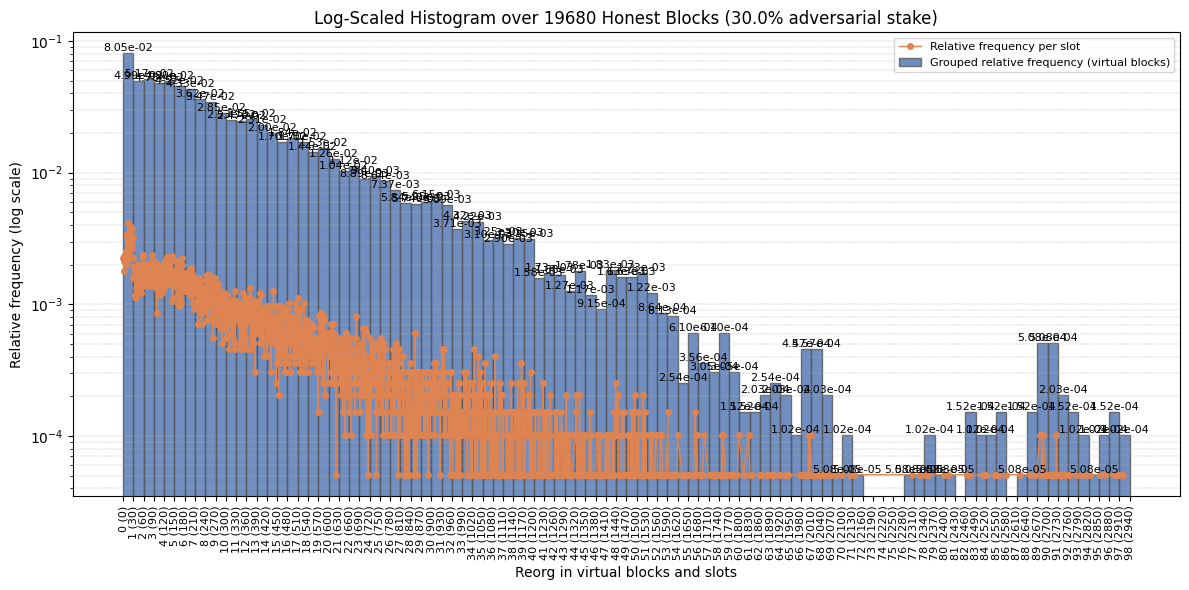}}\\
\subfloat[70 seconds mean broadcast delay]{\includegraphics[width=\textwidth/3]{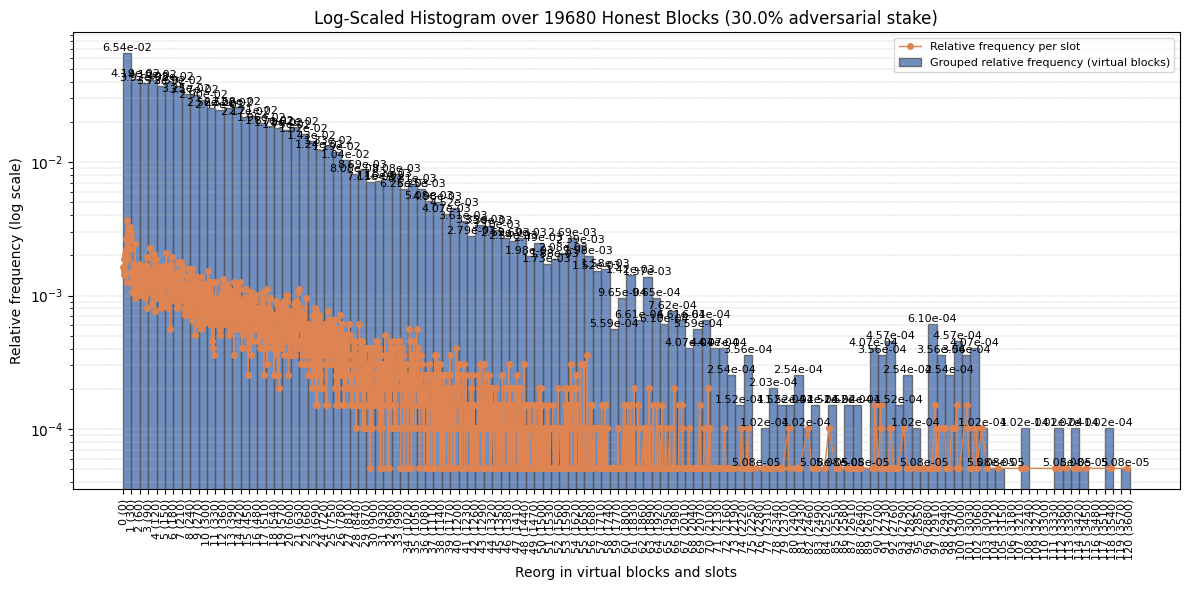}}
\subfloat[80 seconds mean broadcast delay]{\includegraphics[width=\textwidth/3]{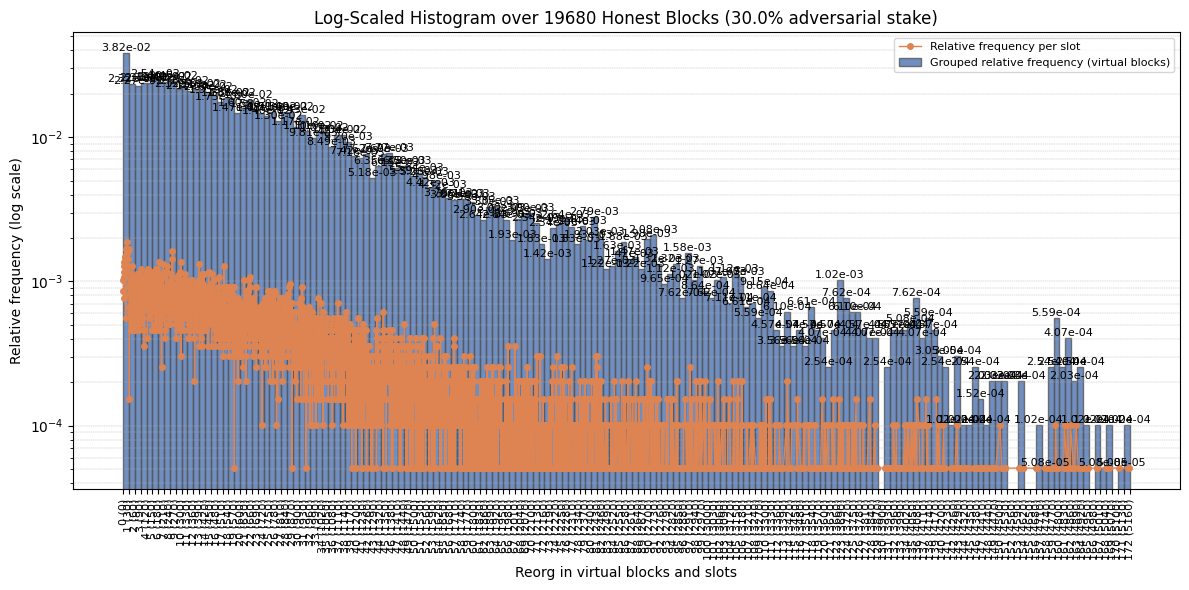}}
\subfloat[90 seconds mean broadcast delay]{\includegraphics[width=\textwidth/3]{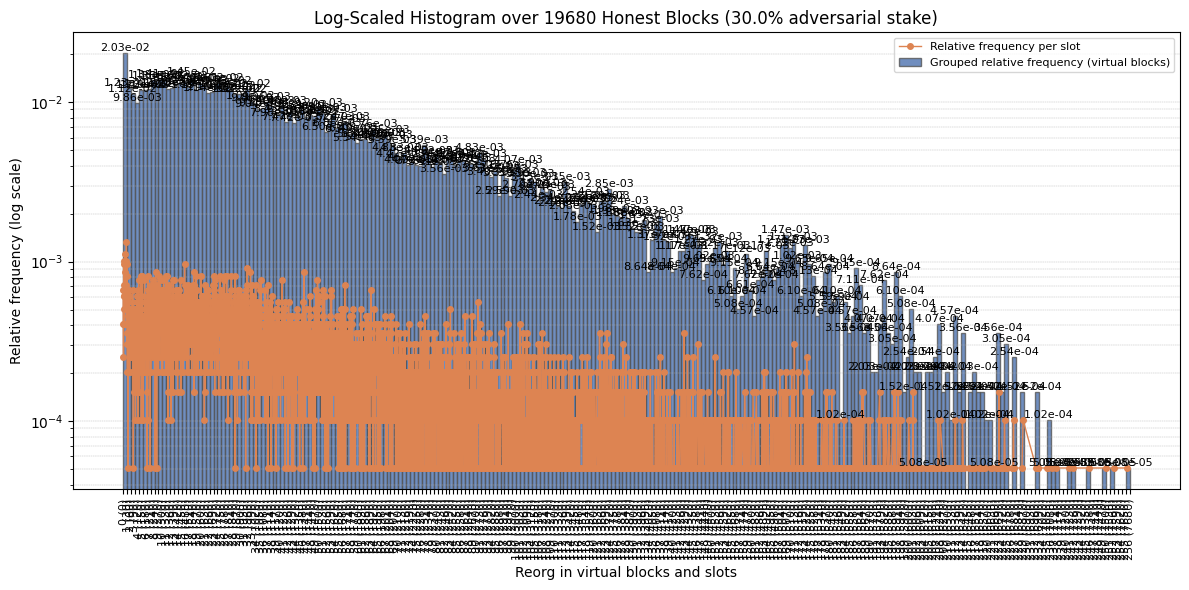}}\\
\caption{Impact of the mean broadcast delay on the reorg depth of \ProjBase with 30\% adversarial stake, window size 30 slots, and $f=0.25$. The reorg depth increases as the mean broadcast delay increases.}
\label{fig:window}
\end{figure}